\documentclass[aps,prl,10pt,twocolumn,superscriptaddress,nofootinbib]{revtex4-2}

\usepackage{amsmath,amssymb,amsthm}
\usepackage{bm,mathrsfs,xcolor}
\usepackage{titletoc}
\usepackage{braket}
\usepackage{graphicx}
\usepackage[caption=false]{subfig}
\usepackage[colorlinks=true,citecolor=magenta,urlcolor=cyan,linkcolor=black]{hyperref}

\usepackage[capitalise]{cleveref}

\newtheorem{theorem}{Theorem}
\newtheorem{lemma}{Lemma}
\newtheorem{corollary}{Corollary}

\newcommand{\id}{\mathbb{I}}
\newcommand{\dt}{\delta t}
\newcommand{\bigO}{\mathcal{O}}
\newcommand{\mc}{\mathcal}
\newcommand{\norm}[1]{\left\lVert #1 \right\rVert}
\newcommand{\abs}[1]{\left\lvert #1 \right\rvert}
\newcommand{\Tr}{\operatorname{Tr}}
\newcommand{\supp}{\operatorname{supp}}

\newcommand{\F}{\mathrm{F}}

\DeclareMathOperator{\tr}{Tr}
\DeclareMathOperator{\support}{supp}
\DeclareMathOperator*{\rprod}{\overrightarrow{\prod}}
\DeclareMathOperator*{\lprod}{\overleftarrow{\prod}}
\newcommand{\s}{\mathbb{S}}
\newcommand{\re}{\mathrm{re}}
\renewcommand{\d}{\mathrm{d}}
\hypersetup{pdftitle={Trotterization, Operator Scrambling, and Entanglement},
  pdfauthor={Tianfeng Feng, Yue Cao, Qi Zhao}}

\begin{document}

\title{Trotterization, Operator Scrambling, and Entanglement}

\author{Tianfeng Feng}
\thanks{These authors contributed equally.}
\affiliation{QICI Quantum Information and Computation Initiative, School of Computing and Data Science, The University of Hong Kong, Pokfulam Road, Hong Kong SAR, China}

\author{Yue Cao}
\thanks{These authors contributed equally.}
\affiliation{QICI Quantum Information and Computation Initiative, School of Computing and Data Science, The University of Hong Kong, Pokfulam Road, Hong Kong SAR, China}

\author{Qi Zhao}
\email{zhaoqi@cs.hku.hk}
\affiliation{QICI Quantum Information and Computation Initiative, School of Computing and Data Science, The University of Hong Kong, Pokfulam Road, Hong Kong SAR, China}

\date{November 2, 2026}

\begin{abstract}
 Operator scrambling describes how initially local operators spread
under many-body dynamics and is closely connected to quantum chaos. In this work, we reveal a connection between operator scrambling and the reliability of quantum simulations. We show that the {Trotter error in an observable expectation value} is bounded by a scrambling correlator between the evolved observable and the multiplicative Trotter error operator. This bound retains the dependence on both the input state and
the observable. We then derive entanglement-based bounds showing that sufficient entanglement brings both the
observable and error-operator factors in the bound to their normalized Frobenius scales. This indicates that, for observables that are normalized sums of local Pauli terms, the leading system-size dependence of the upper bound may be reduced 
from $\bigO(N)$
to $\bigO(1)$ relative to the spectral-norm estimate in systems undergoing entanglement growth and quantum thermalization.
Even when the evolving state is weakly entangled, the effective states generated by the observable or error operator can have larger subsystem entanglement, thereby improving the error bound.
Our results clarify the fundamental connection between Trotterization, scrambling, and entanglement, and provide refined guarantees for quantum simulation of many-body dynamics.
\end{abstract}

\maketitle

\emph{Introduction}.—{Understanding how quantum information 
spreads in many-body systems lies at the core of quantum 
chaos~\cite{hosur2016chaos,roberts2017chaos,
swingle2018unscrambling,mezei2017entanglement} and 
thermalization~\cite{srednicki1994chaos,banuls2011strong}. 
A central diagnostic of this spreading is {operator 
scrambling}: as a local Heisenberg operator evolves in 
time, it grows increasingly non-local, and its commutator 
with other local observables---commonly quantified by 
out-of-time-order correlators---grows 
correspondingly~\cite{Shenker_2014otoc,xu2024scrambling,
hosur2016chaos,roberts2017chaos,liu2014entanglement,banuls2011strong,DowlingPhysRevLett.131.180403}. 
Here we show that operator scrambling also controls the accuracy
of observables in digital quantum simulation.
}

{Digital simulation of quantum many-body dynamics ~\cite{feynman2018simulating,lloyd1996universal,georgescu2014quantum}---for 
example via Trotterization---remains a leading strategy 
for near-term and future quantum 
devices~\cite{lloyd1996universal,childs2019nearly,
chailds2021TrottertheoryPhysRevX.11.011020}. However, 
existing error bounds for Trotter methods are largely 
worst-case estimates, typically based on spectral-norm 
distances between unitaries or naive norm estimates of 
Heisenberg-evolved operators. These bounds often greatly 
overestimate the errors encountered in practical 
simulations of local observables~\cite{heyl2019quantum,
LiPhysRevA.110.062614,
yu2024observabledrivenspeedupsquantumsimulations}, leaving 
a wide gap between formal error guarantees and observed 
performance.}

In this Letter, we bridge this gap by establishing a 
rigorous, quantitative link between Trotter errors, 
operator scrambling, and entanglement. We first prove that 
the error in the evolution of an observable under 
Trotterization is bounded by the operator scrambling 
between that observable and the multiplicative error 
operator of the approximate unitary. 
{While the local expansion of the multiplicative error follows standard product-formula analysis \cite{chailds2021TrottertheoryPhysRevX.11.011020}, the key step here is to retain both the input state and the observable and identify this error operator as the probe in the scrambling correlator.}

{Building on this scrambling picture, we further derive an 
entanglement-based bound. 
When the evolving state attains  sufficient entanglement across the relevant partitions, the observable and error-operator factors in
our bound approach their normalized Frobenius norms
~\cite{chen2024average,
yu2024observabledrivenspeedupsquantumsimulations,
zhao2022hamiltonian, zhao2025entanglement}, leading to a two-fold quadratic reduction in system-size dependence and operator-structure dependence.
For example, for local lattice Hamiltonians, the leading error operator
typically has system-size-dependent spectral norm $\bigO(N)$ but normalized Frobenius
norm $O(\sqrt N)$. 
 For a normalized sum of distinct local Pauli terms (or more generally weakly overlapping Pauli sums),
$O=N^{-1}\sum^{N}_{j=1} O_j$ with $\|O_j\|\le 1$, one has $\|O\|\le 1$ while $\|O\|_F=O(N^{-1/2})$. 
Thus, in sufficiently
entangling regimes, our analysis suggests that the observable
error in simulating lattice Hamiltonians can scale as $\bigO(1)$ rather than the $\bigO(N)$ scaling
of spectral norm estimates.
}

{
Unlike the entanglement-based bound of Ref.~\cite{zhao2025entanglement}, which bounds the observable contribution by its spectral norm, our analysis retains the vector norms associated with the observable, allowing this contribution to benefit from the entanglement of the propagated states.
This reveals an additional,
observable-dependent contribution to the long-time error
that is governed primarily by the entanglement of 
the final-stage states, which in the entangling regimes considered here are more strongly entangled than the relevant early states.
We further show that the error suppression may partially 
persist even when the final state itself is weakly 
entangled, provided the normalized effective states obtained by applying the observable or the error operator to 
that state are themselves sufficiently 
entangled~\cite{Nozaki_2014OE}. }

We validate our theoretical analysis with extensive 
numerical simulations, confirming that the scrambling- and 
entanglement-based bounds accurately predict the scaling 
behavior of Trotter errors in practice. 
{Our results connect entanglement-based and observable-dependent
error estimates and clarify how both can improve the accuracy
of product-formula simulation.}

A key physical message is that stronger interactions may have competing effects: they increase commutators and thus scrambling, tending to enlarge Trotter errors, but they also accelerate entanglement growth, which can tighten the entanglement-based bound and is accompanied by reduced errors. {Our framework quantifies this competition and can sharpen resource estimates}, i.e., gate and communication complexity \cite{childs2018toward,PhysRevLett.133.010603Zhaohongzhen,feng2024distributedquantumsimulation}, for quantum simulators.



\begin{figure*}[tb]
\centering
\subfloat[]{\includegraphics[width=0.45\textwidth]{\detokenize{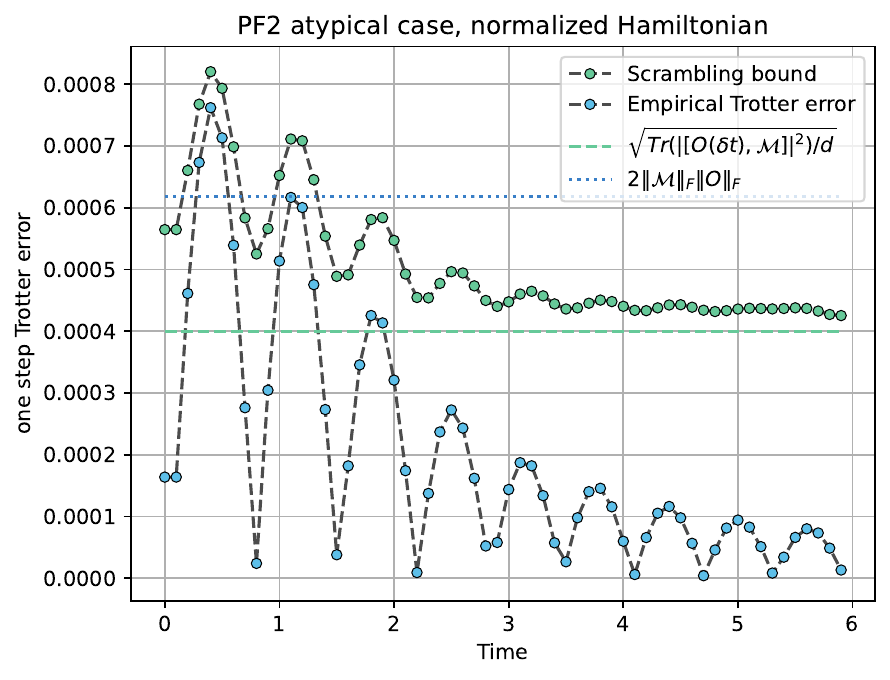}}}
\hspace{0.03\textwidth}
\subfloat[]{\includegraphics[width=0.45\textwidth]{\detokenize{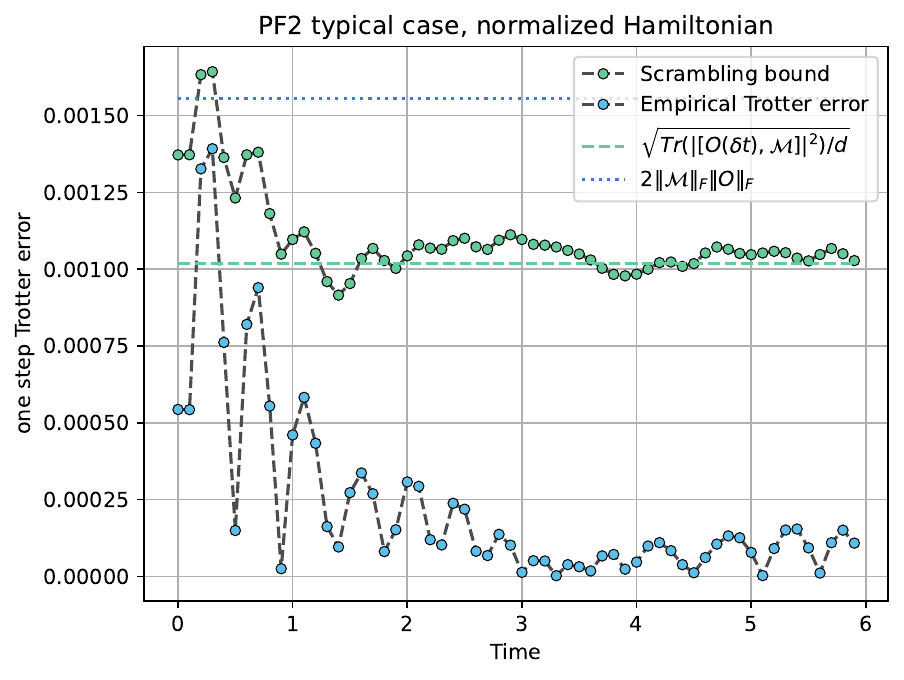}}}
    \caption{
    The one-step Trotter error for the expectation value of the normalized Hamiltonian ($O=H/\norm{H}$), given by $|\langle\psi(t)|U_p^\dagger (\delta t) HU_p(\delta t)-U_0^\dagger(\delta t)  HU_0(\delta t)|\psi(t)\rangle| /\norm{H}$, and the scrambling bound given by Eq. \eqref{eq:scrambling-bound}, are presented. We consider a one-dimensional mixed-field quantum Ising model (QIMF) Hamiltonian, given as $H=h_x\sum_{j=1}^{N}X_j+h_y\sum_{j=1}^NY_j+J\sum_{j=1}^{N-1}X_jX_{j+1}$. The evolution operator $ U_0=e^{-iH \delta t}$ can be approximated using a second-order product formula (PF2) $U_2=e^{-iA\delta t/2}e^{-iB\delta t}e^{-iA\delta t/2}$, where $A=h_x\sum_{j=1}^N X_j+J\sum_{j=1}^{N-1}X_jX_{j+1}$ and $B=h_y\sum_{j=1}^NY_j$. The multiplicative Trotter error $\mathcal{M}$ is calculated according to $\mathcal{M}=U_0^\dagger(\delta t) U_2(\delta t)-I$ for PF2.    
    The initial state is $|01\rangle^{\otimes N/2}$.  Each Trotter segment has a duration of  $t/r=0.1$. In panels (a) and (b), the results of the PF2 method are presented for an atypical parameter set $(h_x,h_y,J)=(0,0.9045,0.4)$ and a typical parameter set $(h_x,h_y,J)=(0.809, 0.9045, 1)$, respectively. {The labels ``typical'' and ``atypical'' refer to the observed entanglement behavior: rapidly entangling/ETH-like \cite{PhysRevA.43.2046ETH91,PhysRevE.50.888ETH96} versus weak-entanglement trajectories, respectively.} 
    In both cases, the scrambling bound closely tracks the error in both parameter
sets and is particularly tight in the weakly entangling case.
{For comparison, the exact worst-case commutator norms 
\(\|[O(\delta t),\mathcal M]\|\) are 
\(1.0912\times10^{-3}\) and \(3.3254\times10^{-3}\) for 
the atypical and typical cases, respectively. 
The green dashed lines show the normalized Frobenius scrambling scale
\(\|[O(\delta t),\mathcal M]\|_F\), which is smaller by 
factors \(2.73\) and \(3.26\), consistent with the expected 
\(\sqrt N\)-type reduction for \(N=10\). The looser spectral product bounds
$2\|O\|\|\mathcal M\|=2\|\mathcal M\|$ are
$5.3180\times10^{-3}$ and $2.0784\times10^{-2}$, whereas the blue dotted Frobenius product scales
$2\|O\|_F\|\mathcal M\|_F$ are
$6.1893\times10^{-4}$ and $1.5549\times10^{-3}$, giving
reduction factors $8.59$ and $13.37$. Thus, the first
comparison illustrates the $\sqrt N$ separation between the
worst-case commutator norm and normalized-Frobenius scrambling scales, while the second
illustrates the two-fold Frobenius reduction that may
improve the scaling from $O(N)$ to $O(1)$ for normalized
physical observables.}}
\label{fig:scrambling}
\end{figure*}

\emph{Setup.}---Consider a Hamiltonian $H = \sum_{l=1}^L H_l$ 
decomposed into $L$ locally implementable terms. The exact 
evolution operator over time $t$ is $U_0(t) = e^{-iHt}$. 
A $p$-th order product formula (PF$p$) with step size 
$\delta t = t/r$ approximates the single-step evolution 
$e^{-iH\delta t}$ by a sequence of local gates; we write 
the resulting one-step approximate unitary as
\begin{equation}
    U_p(\delta t) = U_0(\delta t)\,(I + \mathcal{M}),
    \label{eq:multiplicative}
\end{equation}
where $\mathcal{M}$ is the \emph{multiplicative error 
operator}. Iterating $r$ times gives $U_p^r(\delta t) 
\approx U_0(t)$, with the conventional spectral-norm 
error scaling as $\bigO(t^{p+1}/r^p)$ 
\cite{childs2019nearly,chailds2021TrottertheoryPhysRevX.11.011020}. 
{For PF$p$, the error operator admits a local expansion 
$
    \mathcal{M} = \sum_j M_j\,\delta t^{p+1} 
    + \mathcal{M}_{\rm re} $ and $
    \|\mathcal{M}_{\rm re}\| = \bigO(\delta t^{p+2}),
$
where $M_j$ are local operators supported on a bounded 
number of nested commutators of the $H_l$'s, and we 
denote the leading-order coefficient by 
$M := \sum_j M_j$ \cite{footnote1}. In all numerical results below, we evaluate the exact multiplicative error $\mc M=U^\dagger _0(\delta t)U_p(\delta t) -I$, whereas $M$ denotes the leading-order coefficient used in the analytical expansion $\mathcal{M} = \sum_j M_j\,\delta t^{p+1} +\bigO (\delta t^{p+2})$.}

{Let $O$ be a Hermitian observable and $|\psi\rangle$ an 
input state. The exact single-step Heisenberg evolution 
of $O$ is
$
    O(\delta t) := U_0^\dagger(\delta t)\,O\,U_0(\delta t),
$
and the one-step observable error is defined as
\begin{equation}
    \epsilon_O(\delta t) := \bigl|\langle\psi|\,
    O(\delta t) - U_p^\dagger(\delta t)\,O\,U_p(\delta t)
    \,|\psi\rangle\bigr|.
    \label{eq:epsO-def}
\end{equation}
This is the quantity we bound in the remainder of the 
Letter, first in terms of operator scrambling and then 
in terms of entanglement.}

\emph{Scrambling-based bound on Trotter error of observables.}---
To connect $\epsilon_O$ to operator scrambling, we recall 
the standard scrambling correlator 
\cite{landsman2019verified,Shenker_2014otoc,xu2024scrambling}
\begin{equation}
    C_{O,V}(t) := \bigl\langle 
    [O(t),V]^\dagger[O(t),V] \bigr\rangle,
\end{equation}
where $\langle A\rangle := \langle\psi|A|\psi\rangle$. 
When $O$ and a probe $V$ are initially disjoint, 
$C_{O,V}(t)$ vanishes at $t=0$ and becomes appreciable as the evolved support of \(O(t)\) reaches that of \(V\). 
Here the relevant probe is the multiplicative error 
operator $\mathcal{M}$ itself, e.g., $V=\mathcal{M}$.

{Using 
Eq.~\eqref{eq:multiplicative} and the Cauchy-Schwarz inequality, a calculation (see Section III A in Supplemental Material) shows that the one-step observable error obeys}
\begin{equation}
    \epsilon_O^2 (\delta t) \le 
    \langle\psi|\,[O(\delta t),\mathcal{M}]^\dagger
    [O(\delta t),\mathcal{M}]\,|\psi\rangle
    \;=:\; C_{O,\mathcal{M}}(\delta t).
    \label{eq:scrambling-bound}
\end{equation}
The right-hand side is precisely the operator-scrambling 
correlator between the evolved observable and the 
multiplicative error operator: it vanishes whenever 
$O(\delta t)$ and $\mathcal{M}$ commute, and grows as the effective support of $O(\delta t)$ overlaps with that of 
$\mathcal{M}$.  The worst-case commutator bound is 
recovered as $\max_{|\psi\rangle} \epsilon_O = 
\|[O(\delta t),\mathcal{M}]\|$ ({note that Lieb-Robinson estimates provide further upper bounds on this worst-case commutator norm.}), so 
Eq.~\eqref{eq:scrambling-bound} is 
retains the state dependence that is lost in the corresponding worst-case commutator norm
\cite{heyl2019quantum,LiPhysRevA.110.062614,chailds2021TrottertheoryPhysRevX.11.011020,yu2024observabledrivenspeedupsquantumsimulations}. {For Haar-random or 1-design
inputs, the relevant scale is bounded by the normalized 
Frobenius norm 
\begin{equation}
\mathbb E_\psi \sqrt{C_{O,\mc M}(\dt)}  \le\|[O(\delta t),\mathcal{M}]\|_F, 
\end{equation}
where $\norm{A}_{\F}=\sqrt{\Tr(A^\dagger A)/d},$
and $d$ is the Hilbert-space dimension. This is 
consistent with the Frobenius light-cone picture 
\cite{tran2020hierarchylightcone}. 
Figure~\ref{fig:scrambling} tests the scrambling-based bound and compares the empirical one-step error with Eq.~\eqref{eq:scrambling-bound} for both weak- and strong-entangling regimes.} 
{For random inputs, and more generally in the sufficiently entangling regimes considered below, the relevant scrambling scale can approach the normalized Frobenius norm \(\|[O(\delta t),\mathcal M]\|_F\), rather than the worst-case commutator norm.}
{This normalized Frobenius norm may give the expected
$\sqrt N$ separation from the worst-case commutator norm for the local Pauli-sum structures considered here (Sec. III D of the Supplemental Material for the proof and the End Matter for tests up to 1000 qubits).}


{For long-time evolution $t = r\delta t$, we compare the 
Trotterized expectation value to the exact one via a 
telescoping expansion over the $r$ segments. Writing 
$O(k\delta t) = e^{iHk\delta t}Oe^{-iHk\delta t}$ and 
$|\psi_k\rangle = U_p^{r-k}|\psi\rangle$, one obtains 
(see Section III C in Supplemental Material):}

\begin{theorem}[Accumulated scrambling bound]
\label{thm:scramble}
Given $U_0(t) = e^{-iHt}$ and $U_p^r$ with 
$U_p = e^{-iH\delta t}(I+\mathcal{M})$, 
$t = r\delta t$, and $\mathcal{M} = M\delta t^{p+1} + 
\mathcal{M}_{\rm re}$, the observable error satisfies
\begin{equation}
    \epsilon_O(t) \le \sum_{k=1}^{r} 
    \sqrt{C^{(k)}_{O,{M}}}\;\delta t^{p+1} 
    + 2r\|O\|\,\|\mathcal{M}_{\rm re}\|,
\end{equation}
where $C^{(k)}_{O,{M}} := \langle\psi_k|\,|[O(k\delta t),M]|^2\,
|\psi_k\rangle$ is the operator scrambling evaluated at 
the $k$-th segment, and $\|\mathcal{M}_{\rm re}\| = 
O\bigl(
\delta t^{p+2}\bigr)$.
\end{theorem}
{
The index $k$ measures the ideal evolution time remaining between
an error insertion and the final measurement.
The theorem therefore accounts for the growth of the scrambling
correlator across all Trotter segments.
For local Hamiltonians, contributions from error terms outside
the observable's backward light cone are suppressed by locality.
We examine this segment dependence numerically in Section~IV E
of the Supplemental Material.
} 


\emph{Entanglement-based bound on observable error.}--- 
Scrambling is typically accompanied by entanglement growth. 
We now show how entanglement brings the vector norms entering
the observable-error bound toward their normalized Frobenius values \cite{zhao2025entanglement}.
{The basic idea is a vector-norm 
estimate: if the reduced state of $\ket{\chi}$ is nearly 
maximally mixed on the support where $A$ acts, then 
$\norm{A\ket{\chi}}$ is close to the normalized Frobenius 
norm of $A$, rather than its spectral norm. 
For any 
operator $A=\sum_\alpha A_\alpha$ and pure state 
$\ket{\chi}$, 
let 
$R_{\alpha\beta}=\supp(A_\alpha^\dagger A_\beta)$, let 
$d_{\alpha\beta}$ be the dimension of this subsystem, and 
let $\rho^\chi_{\alpha\beta}$ be the reduced state of 
$\ket{\chi}$ on $R_{\alpha\beta}$. The entropy deficit 
from the maximally mixed state is
$
  D_{\alpha\beta}(\chi)=\log d_{\alpha\beta}
  -S(\rho^\chi_{\alpha\beta}).
$
Using Pinsker's inequality as in entanglement-based 
bounds~\cite{zhao2025entanglement}, 
one obtains
$
  \norm{A\ket{\chi}}^2=\bra{\chi} A^\dagger A\ket{ \chi}
  \le
  \norm{A}_{\F}^2+\Delta_A(\chi),
$
with
$
  \Delta_A(\chi)
  =
  \sum_{\alpha,\beta}
  \norm{A_\alpha^\dagger A_\beta}
  \sqrt{2D_{\alpha\beta}(\chi)} .
$
The term $\Delta_A(\chi)$ is controlled by the entropy 
deficit. It is small when the relevant reduced states are close to maximally mixed on the supports of the products $A_\alpha^\dagger 
A_\beta$. 
For notational 
convenience, we write
\begin{equation}
  \norm{A}_\chi := \sqrt{\norm{A}_{\F}^2+\Delta_A(\chi)},
  \label{eq:effective-norm}
\end{equation}
so that 
$\norm{A\ket{\chi}}\le \norm{A}_\chi$. 
When $\Delta_A(\chi)\ll\norm{A}_{\F}^2$,
$\norm{A}_\chi$ approaches $\norm{A}_{\F}$.
The two are equal when the entropy-deficit correction vanishes.
}

{To apply this estimate to the observable error, 
expanding the commutator in Eq.~\eqref{eq:scrambling-bound}
and applying the triangle inequality gives
(see Section IV A in Supplemental Material)
\begin{equation}
\begin{split}
  \epsilon_O(\delta t)
  \le
  \norm{O\ket{\psi(\delta t)}}\,\norm{\mc M\ket{\psi_{O(\delta t)}}}\\
  +
  \norm{\mc M\ket{\psi}}\,\norm{O\ket{\psi_{U_0 (\delta t)\mc M}}},
  \label{eq:vector-bound}
  \end{split}
\end{equation}
where $\ket{\psi(\delta t)}=U_0(\delta t)\ket{\psi}=e^{-iH\delta t}\ket{\psi}$ and 
 $   \ket{\psi_{X}}=
   \frac{X\ket{\psi}}{\norm{X\ket{\psi}}}.
 $
Substituting 
Eq.~\eqref{eq:effective-norm} and the leading expansion 
of $\mc M=\sum_j M_j \delta t^{p+1} +\bigO (\delta t^{p+2})$ gives the following entanglement-based bound (see Section IV B in Supplemental Material).}

{\begin{theorem}[Entanglement bound for observable error]
\label{thm:entanglement}
Let $M=\sum_a M_a$ and 
$O=\sum_j O_j$  be the local 
decompositions, with $\mc M=M\delta t^{p+1}+\bigO(\delta t^{p+2})$. 
For a single segment,
\begin{equation}
\begin{aligned}
  \epsilon_O(\delta t)
  &\le
  \delta t^{p+1}
  \left[
  \norm{O}_{\psi(\delta t)}\,\norm{M}_{\psi_{O(\delta t)}}
+\norm{M}_\psi\,\norm{O}_{\psi_{U_0 (\delta t) M}}
  \right]\\
  & \quad+\bigO(\delta t^{p+2}),
\end{aligned}
\label{eq:entanglement-bound}
\end{equation}
where $\norm{A}_\chi$ is defined in 
Eq.~\eqref{eq:effective-norm}, with the entropy-deficit 
corrections $\Delta_O$ and $\Delta_M$ evaluated on the 
supports of $O_j^\dagger O_{j'}$ and $M_a^\dagger M_b$, 
respectively.
\end{theorem}
When the entropy-deficit corrections are negligible for all four effective norms appearing in
Eq.~\eqref{eq:entanglement-bound}—Theorem~
\ref{thm:entanglement} reduces to
\begin{equation}
  \epsilon_O(\delta t)
  \lesssim
  2\norm{O}_{\F}\norm{M}_{\F}\,\delta t^{p+1 }
  +\mathcal{O}(\delta t^{p+2}).
  \label{eq:frobenius}
\end{equation}
This is the observable analogue of average-case 
product-formula scaling~\cite{LiPhysRevA.110.062614,yu2024observabledrivenspeedupsquantumsimulations}, but 
the averaging is replaced by entanglement properties of 
the particular states that appear in the observable 
error. This gives a simple system-size estimate. For a local 
lattice Hamiltonian, the leading-order coefficient $\norm{M}$ typically satisfies
 $\norm{M}=\bigO(N)$ while $\norm{M}_{\F}
=\bigO(\sqrt N)$~\cite{chailds2021TrottertheoryPhysRevX.11.011020}. For a normalized physical observable 
$O=N^{-1}\sum_{j=1}^N P_j$ with distinct local Pauli terms, $\norm{O}\le 1$ while $\norm{O}_{\F}=\bigO(N^{-1/2})$. 
At leading order, the corresponding spectral-product estimate $\|[O(\delta t),\mathcal{M}]\|\le 2\norm{O}\norm{M}\,\delta 
t^{p+1}+\bigO(\delta t^{p+2})$ scales as $\bigO(N\delta t^{p+1})$, whereas 
Eq.~\eqref{eq:frobenius} scales as $\bigO(\delta t^{p+1})$—a 
two-fold quadratic improvement in system-size dependence. This indicates that, for such normalized local observables, the upper bound on the observable error may be reduced from $\bigO(N)$
 to  $\bigO(1)$ in systems undergoing entanglement growth and thermalization. }

\begin{figure}[tb]
\centering
\includegraphics[width=0.95\linewidth]{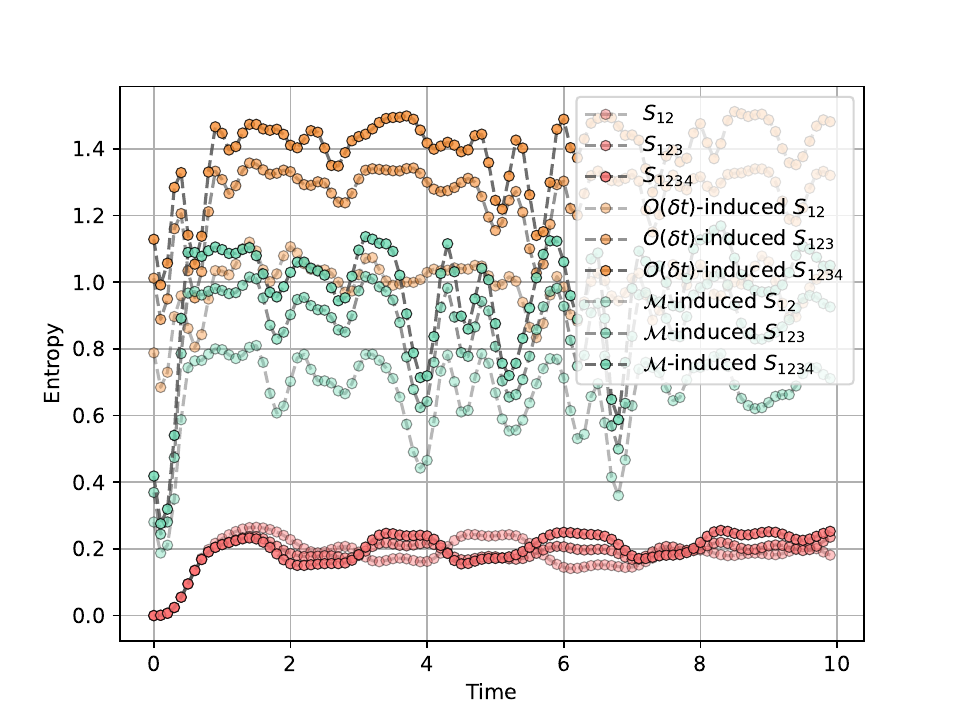}
\caption{{Operator-induced entanglement here is the entanglement of the normalized states generated by applying
$O(\delta t)$ or $\mc M$ to the evolving state.}
This should not be confused with operator entanglement of the operator itself (e.g., its Choi-state entanglement).
Specifically, it means the von Neumann entropy of normalized states such as 
$O(\delta t)\ket{\psi}/\norm{O(\delta t)\ket{\psi}}$ and 
$\mc M\ket{\psi}/\norm{\mc M\ket{\psi}}$. The model is the 
mixed-field Ising chain with $(h_x,h_y,J)=
(0.809,0.9045,1)$, initial state $\ket{+}^{\otimes N}$, 
and observable $O=\sum_i Z_iZ_{i+1}/
\norm{\sum_iZ_iZ_{i+1}}$. {
Even when the evolving state has low subsystem entropy,
applying $O(\delta t)$ or $\mathcal M$ can produce effective states
with substantially higher entropy and smaller entropy deficits
on the displayed subsystems.
This illustrates a mechanism that can partially improve
the entanglement-based estimate.
}}
\label{fig:induced}
\end{figure}

{Theorem~\ref{thm:entanglement} also reveals another route to
error suppression. Even when the evolved state itself is weakly
entangled, for example, in the case of many-body localization  \cite{nandkishore2015manybodylocalize,MBLRevModPhys.91.021001}, the operator-generated effective states may nevertheless have smaller entropy deficits.}
Specifically, the states $\ket{\psi_{O(\delta t)}}$ and $\ket{\psi_{U_0 (\delta t) M}}$ in 
Theorem~\ref{thm:entanglement} can explain why 
observable simulation can improve even when the evolved 
state $\ket{\psi}$ itself is not highly entangled. The 
effective norms $\norm{O}_{\psi_{U_0 (\delta t) M}}$ and 
$\norm{M}_{\psi_{O(\delta t)}}$ in Eq.~\eqref{eq:entanglement-bound} 
are evaluated after applying $O(\delta t)$ or $U_0(\delta t)M$, not 
on the original evolved state. A structured sum of local 
terms can create entanglement in the normalized vector 
$A\ket{\chi}/\norm{A\ket{\chi}}$ even when $\ket{\chi}$ 
itself carries little entanglement. We call this 
\emph{operator-induced entanglement}, following the entanglement of locally excited states~\cite{Nozaki_2014OE}. {
Figure~\ref{fig:induced} illustrates how operator-induced
entanglement can yield higher subsystem entropies in the effective states than in the evolving state itself. In this numerical diagnostic, we use the exact multiplicative error \(\mathcal M\) while Theorem 2 uses its leading coefficient \(M\).
The resulting smaller entropy deficits on these subsystems
can reduce the corresponding entropy-deficit corrections,
providing a route to partial improvement of the
entropy-based bound even when the evolving state remains
weakly entangled.
}

\emph{Long-time observable error.}---
{
A state- and
observable-dependent decomposition of the long-time error
(End Matter) reveals a sharper structure: each segment
contains one factor $\norm{O\ket{\cdot}}$ governed by the observable $O$ and one
factor $\norm{\mc M\ket{\cdot}}$ governed by the error operator $\mc M$. These two factors depend on different parts of the evolution.
The observable-dependent factor is evaluated on an effective
state that has undergone the entire evolution with an operator shift, so it is
controlled mainly by the final state. The error-operator factor mainly
depends on a partial evolution up to that segment, and
therefore can remain weakly entangled at early or intermediate
segments. 
In the entangling regimes considered here, these nearly fully propagated states inherit the small entropy deficits associated with the final-stage dynamics. Consequently, the observable-dependent factors may approach the Frobenius scale across essentially all segments,
$
  \norm{O\ket{\cdot}}
  \approx
  \norm{O}_F$. The
error-operator factors $\norm{\mc M\ket{\cdot}}$, in contrast,
depend on the entanglement of each intermediate evolved
state ~\cite{zhao2025entanglement}.
We verify this numerically in Fig.~\ref{fig:distribution} (see End Matter), where the observable-dependent norms concentrate around $\norm{O}_F$ over 1000 segments.
}


\emph{Operator scrambling versus entanglement.}---
From a mathematical perspective, the Trotter error is bounded as $\epsilon^2_O\le \bra{\psi} |[O(\delta t),\mathcal{M}]|^2\ket{\psi} $, suggesting a fundamental trade-off between the magnitude of operator commutators involved in scrambling and entanglement growth, which together govern the increase and suppression of Trotter errors.
To numerically illustrate this phenomenon, we consider the two-body observable 
$O=\sum_{i=1}^{N-1} X_iX_{i+1}/{\|\sum_{i=1}^{N-1}X_iX_{i+1}\|_\infty}$
, and compare its Trotter error growth under two Hamiltonians: $H_1$ and $H_2$
 with parameters $(h_x,h_y,J)=(0.8,0.9,1)$ and $(h_x,h_y,J)=(1.12,1.26,1.4)$, respectively, such that $H_2=1.4H_1$.
For short times, the leading Trotter error is governed by $[O(\delta t),M]$
, where $M$ is the leading-order coefficient of the multiplicative error; notably, $1.4^2[O(\delta t), M^{(1)}]\approx [O(\delta t), M^{(2)}]$ 
 for $H_1$ and $H_2$. With the worst-case initial state (spectral norm case), as shown in Figure~\ref{fig:interaction_of_different_strengths}, the initial Trotter error for the strongly coupled Hamiltonian $H_2$ exceeds that for $H_1$  by a factor of $1.9585\approx1.4^2$, in agreement with the increased magnitude in the commutator. However, at later times, the Trotter error for $H_2$ decays more rapidly than for $H_1$, concomitant with faster entanglement growth. This crossover behavior is well captured by fitting a damped oscillatory envelope to the error dynamics (see dashed line with light color in Figure~\ref{fig:interaction_of_different_strengths}). 

\begin{figure}[tb]
    \centering
    \includegraphics[width=0.98\linewidth]
    {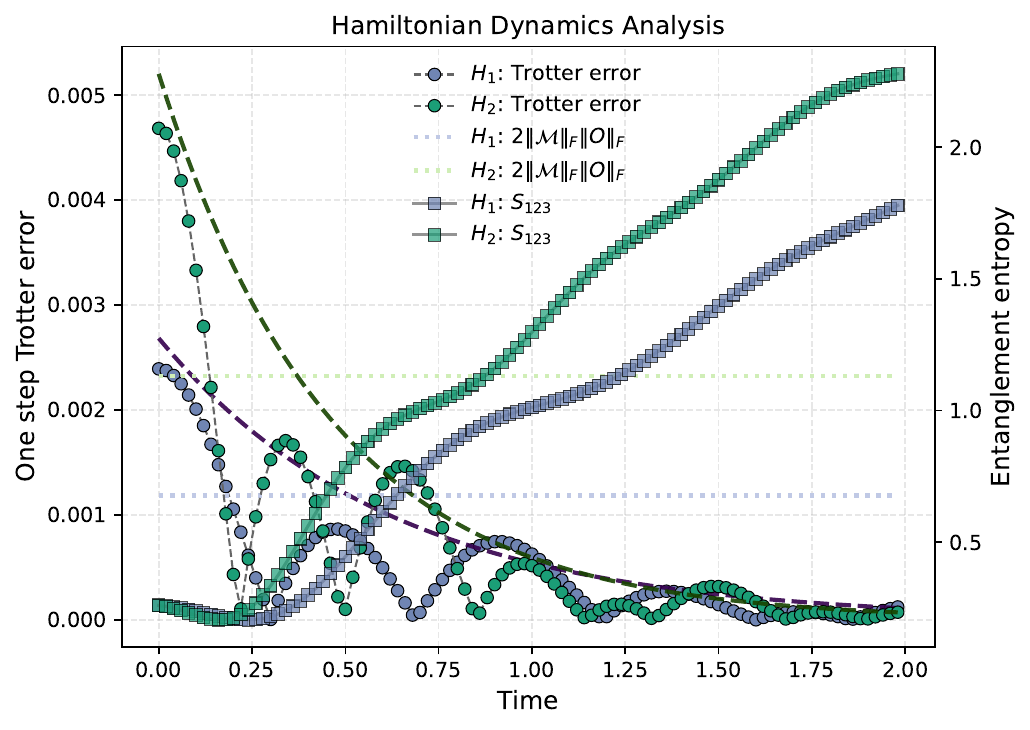}
    \caption{One-step Trotter error for the two-body 
    observable $O$ under $H_1$ (weak coupling, 
    $(h_x,h_y,J)=(0.8,0.9,1)$) and $H_2$ (strong 
    coupling, $H_2 = 1.4\,H_1$) using PF1, with the 
    worst-case (spectral-norm) initial state. The error 
    is initially larger for $H_2$, consistent with the 
    stronger commutator $[O(\delta t), M^{(2)}] \approx 
    1.4^2\,[O(\delta t),M^{(1)}]$, but decays faster at later 
    times concomitant with more rapid entanglement growth (tracked 
    by the three-site entropy $S_{123}$, right axis). 
    Dashed curves show fits to a damped oscillatory 
    envelope, illustrating the crossover associated with the competition between operator scrambling and 
    entanglement.}
    \label{fig:interaction_of_different_strengths}
\end{figure}

\emph{Discussion and conclusions.---}
Recently, local operator entanglement has emerged as a powerful complementary measure, linking operator growth and operator scrambling to the classical simulability of quantum evolution \cite{DowlingPhysRevLett.131.180403,dowling2025bridging}.
{We have connected observable Trotter errors to operator scrambling
and entanglement.
The scrambling bound retains the input-state dependence, while
the entanglement-based estimate identifies when both the observable
and error-operator factors approach their normalized Frobenius scales.
}
Additionally, even in the weak-entanglement case of input states \cite{zhao2025entanglement}, we suggest that the observable-dependent error reduction remains possible.
A further direction is to study the entanglement rate in Trotterization \cite{Bravyi_2007,_Anthony_Chen_2023}, which constrains the minimum evolution time required to achieve a given level of entanglement.

 While our analysis primarily focuses on Trotter errors of observables, it is also applicable to coherent noise in general quantum circuits and to analog quantum simulations subject to perturbation noise \cite{buluta2009quantum,georgescu2014quantum_anolog,cai_et_al:LIPIcs.TQC.2024.2,PhysRevX.15.021017}. For example, the Trotterized unitary ${U}^r_p$ can be replaced by any approximate circuit $\prod_{k=1}^r V_k$ of $U_0$ where $V_k$ is the approximate unitary for a single step. Thus, our result also suggests that the observable error for evolution under imperfect unitary circuits is bounded by the accumulated operator scrambling with respect to the observable $O$ and the corresponding multiplicative error operator $\mathcal{M}_k$.

Tighter error bounds can reduce the estimated gate and communication costs
\cite{chailds2021TrottertheoryPhysRevX.11.011020,childs2018toward,childs2019nearly,PhysRevLett.133.010603Zhaohongzhen,zhao2022hamiltonian,DQSPhysRevResearch.5.L022003,feng2024distributedquantumsimulation}.
Consequently, our work not only reveals connections between Trotterization and many-body dynamics but can also provide tighter gate- and communication-resource estimates for both current noisy quantum processors and future fault-tolerant quantum computers.

\begin{acknowledgments}
T. F. and Y. C. contributed equally to this work.
    Q.Z. acknowledges funding from Quantum Science and Technology-National Science and Technology Major Project 2024ZD0301900, National Natural Science Foundation of China (NSFC) via Project No. 12347104 and No. 12305030, Hong Kong Research Grant Council (RGC) via No. 27300823, 17310926, N\_HKU718/23, R6010-23, and the CCF-Tencent Rhino-Bird Open Research Fund (No. T105-TEG-2024080200001).
\end{acknowledgments}

\begin{figure*}[htp]
    \centering
    \subfloat[]{
    \includegraphics[width=0.45\linewidth]{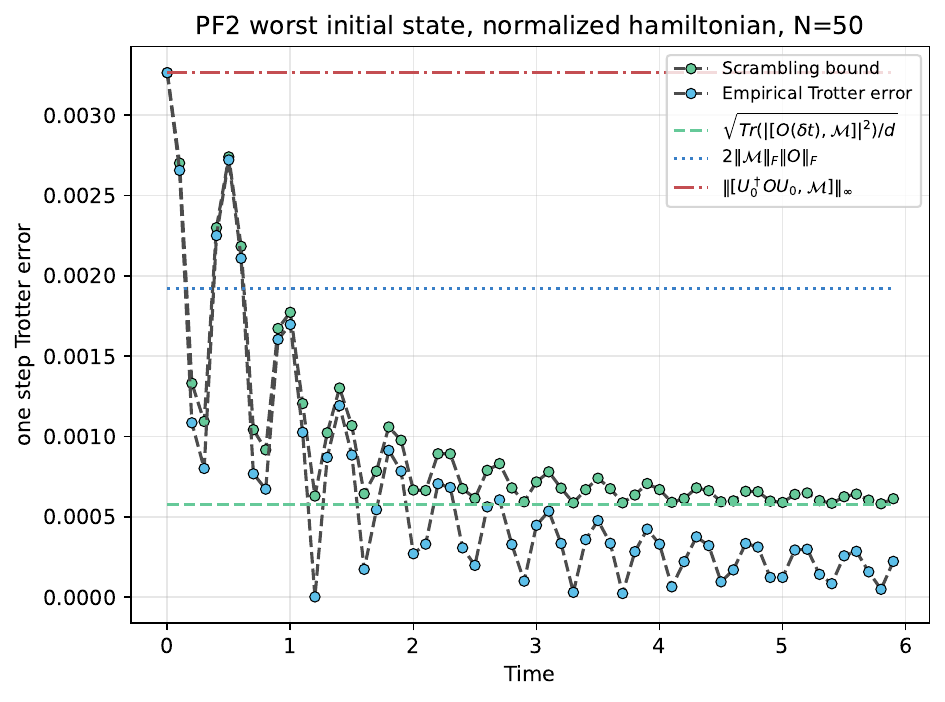}}
    \subfloat[]{
    \includegraphics[width=0.45\linewidth]{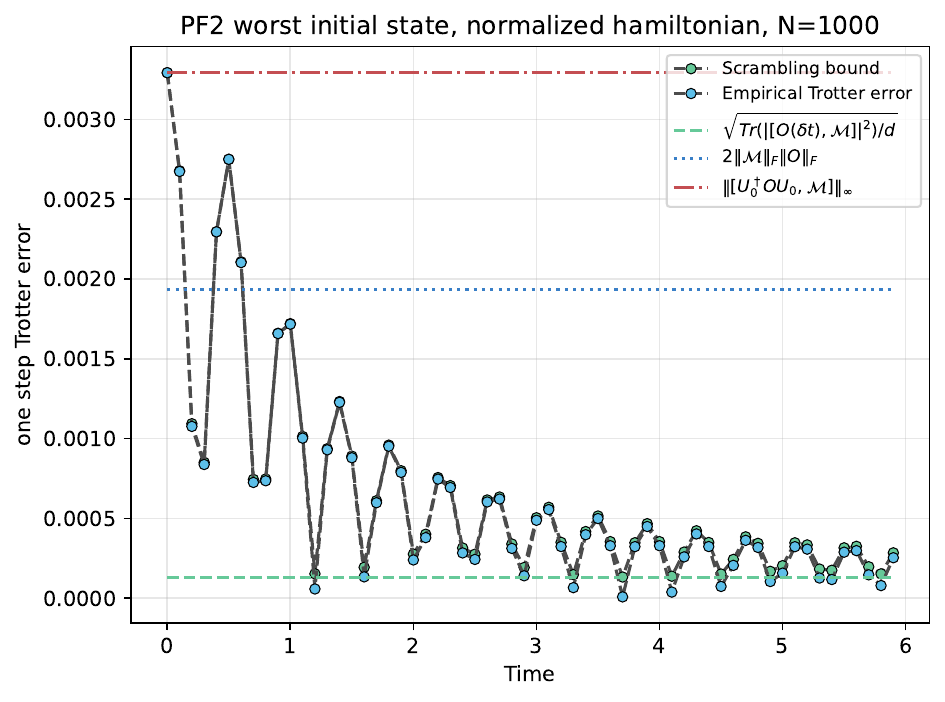}}
    \caption{{ Large-system verification of scrambling-based bound for (a) $N=50$ and (b) $N=1000$, for Hamiltonian $  H=g\sum_{j=1}^{N}Y_j+J\sum_{j=1}^{N-1}X_jX_{j+1}$ with
    $(g,J)=(0.9045,1)$, PF2 step size $\delta t=0.1$, and
    normalized observable $O=H/\norm{H}$. The initial pure
    Gaussian state is chosen to maximize the one-step
    observable error within the pure-Gaussian-state family. We compare the empirical error with the
    state-dependent scrambling bound in
    Eq.~\eqref{eq:scrambling-bound}, the average
    scrambling scale
    $\sqrt{\Tr(|[O(\delta t),\mc M]|^2)/d}$, the Frobenius
    norm scale $2\norm{\mc M}_{\F}\norm{O}_{\F}$, and the
    worst-case commutator norm. The calculation uses the
    fermionic Gaussian covariance-matrix formalism of
    Ref.~\cite{surace2022fermionic}.}}
    \label{fig:largeN-gaussian}
\end{figure*}

\section*{End Matter}

\subsection{General Scrambling-based bound of observable error}

In the main text, we give a scrambling-based bound for the observable in short-time quantum dynamics $e^{-iH\delta t}$. Here, we give a more general scrambling-based bound, which can be applied to any quantum circuit or unitary (see section III A in Supplemental Material).

\begin{lemma} [Operator scrambling-based bound of quantum circuits] \label{scrambling1}
Given an ideal quantum circuit $U_0$ and approximate unitary $U$, the error of the observable $O$ is bounded by   
\begin{equation}
 \epsilon_O=\abs{\langle \psi| U^\dagger_0 O U_0- U^\dagger O U \ket{\psi}} \le \sqrt{\bra{\psi} |[O_{U_0},\mathcal{M} ]|^2  \ket{\psi} },
\end{equation}
where $U=U_0(I+\mathcal{M})$, and  $O_{U_0}=U^\dagger_0 O U_0$ denotes the ideal evolution of observable $O$.
\end{lemma}
Suppose $\mathcal{M}$ and $O$ admit local decompositions and we can further bound the observable error by summing the operator scrambling of the local error term $\mathcal{M}_a$ using the triangle inequality. 

\begin{corollary}
[Operator scrambling-based bound (local case)] \label{scrambling}
Given an ideal quantum circuit $U_0$ and approximate unitary $U$, the error of the observable $O$ is bounded by 
 \begin{equation}
    \begin{aligned}
    \epsilon_O&\le \sqrt{\bra{\psi} |[O_{U_0},\mathcal{M} ]|^2  \ket{\psi} }\le\sum_{x,a} \sqrt{\bra{\psi} |[O_{x,U_0},\mathcal M_a ]|^2  \ket{\psi} },
    \end{aligned}
\end{equation}   
where $U=U_0(I+\mathcal{M})$, $\mathcal{M}=\sum_a \mathcal M_a$, $O=\sum_x O_x$ and  $O_{x,U_0}=U^\dagger_0 O_x U_0$ denote the ideal evolution of observable $O$.
\end{corollary}

{For a
local Hamiltonian, let
$
  O=\sum_x c_x O_x,$ and consider the Trotter error of $p$-th product formula $\mathcal M=\sum_a  M_a \delta t^{p+1} +\bigO{(\delta t^{p+2})}
$. A Lieb--Robinson bound
implies
$
 \|[O_x(\tau), M_a]\|
 \le
 C_{\rm LR} \|O_x\|\| M_a\|
 e^{-\mu(d(x,a)-v_{\rm LR}\tau)} ,
$
where $C_{\rm LR}$ and $v_{\rm LR}$ can be chosen uniformly in $x,a$ and independently of $N$~\cite{lieb1972finite}.
Therefore, for all terms, we have
\begin{equation}
\begin{split}
\sqrt{C^{(k)}_{O, M}}
&\le
\|[O(k\delta t), M]\|\\
&\le
C_{\rm LR} 
\sum_{x,a}
|c_x|\|O_x\|\|M_a\|
e^{-\mu(d(x,a)-v_{\rm LR}k\delta t)} .
\end{split}
\end{equation}
This shows that the growth of $C^{(k)}_{O, M}$ is
controlled by the backward Lieb--Robinson light cone of the
final observable. Although this is a double sum, it is not $\mathcal O(N^2)$ for
local systems: for each fixed $x$,
$\sum_a e^{-\mu d(x,a)}=\mathcal O(1)$ on a bounded-degree lattice.
Thus the $k$-dependent growth is explicitly included in
the accumulated bound and is bounded by locality.}

\subsection{Large-system Gaussian verification}

{To further test the  scrambling bound beyond the  worst-case commutator norm estimate, we consider the zero-longitudinal-field Ising chain
\begin{equation}
    H=g\sum_{j=1}^{N}Y_j+J\sum_{j=1}^{N-1}X_jX_{j+1}.
    \label{Gaussian}
\end{equation}
After a Jordan--Wigner transformation, the state and dynamics
can be represented by $2N\times2N$ Majorana covariance
matrices using the standard fermionic Gaussian
formalism~\cite{surace2022fermionic}. This allows us to evaluate the same one-step error and bounds
for systems up to $N=1000$.
As shown in Fig.~\ref{fig:largeN-gaussian}, the
state-dependent scrambling bound continues to track the
empirical error and, after the initial worst-case commutator point, both
lie well below the worst-case scale. Note that compared to the worst-case commutator norms, 
the normalized-Frobenius scrambling scale generally gives a \(\sqrt N\)-type reduction up to a constant factor,
which is consistent with our numerical results $ 7.07$ and $  25.3$ respectively.}

\subsection{Accumulated vector-norm bound of observable error}

To analyze the long-time evolution error, we combine 
the triangle inequality and 
Eq.~\eqref{eq:vector-bound}, leading to the following 
lemma (see Section IV A in Supplemental Material).

\begin{lemma}[Accumulated vector-norm bound]
\label{main:eq:bound_for_lt}
For a quantum circuit $U_0^r$ and an approximate unitary 
$U_p^r$ with $U_p = U_0(I+\mc M)$, the observable error 
for input state $\ket{\psi}$ satisfies
\begin{equation}
    \begin{aligned}
    \epsilon_O &\le \sum_{k=1}^{r}
    \Bigl(\norm{O\ket{\tilde\psi_{r_k}}}\,
    \norm{\mc M\ket{\psi_{k,O_k}}}
    + \norm{O\ket{\tilde\psi_{r_k,\mc M}}}\,
    \norm{\mc M\ket{\psi_k}}\Bigr),
    \end{aligned}
\end{equation}
where 
$
\ket{\tilde\psi_{r_k}} = U_0^k U_p^{r-k}\ket{\psi}, $ $
    \ket{\psi_k} = U_p^{r-k}\ket{\psi}, $ $
    O_k = U_0^{\dagger k}\,O\,U_0^k,$ $
    \ket{\psi_{k,O_k}} = \frac{O_k\ket{\psi_k}}
{\norm{O_k\ket{\psi_k}}}, $ $
    \ket{\tilde\psi_{r_k,\mc M}} = \frac{U_0^k\,\mc M\, 
    U_p^{r-k}\ket{\psi}}
    {\norm{U_0^k\,\mc M\, U_p^{r-k}\ket{\psi}}}.
$
\end{lemma}

This bound is expressed in terms of vector norms, which 
can be converted to entanglement-entropy statements via 
Eq.~\eqref{eq:effective-norm}. Although the structure 
superficially resembles a straightforward accumulation of 
the single-step result, it exhibits qualitatively 
different behavior at long times, and is fundamentally 
distinct from a naive accumulation of independent 
state-dependent errors as in 
Ref.~\cite{zhao2025entanglement}.

In the entangling regimes considered here, we find that when the final-stage dynamics has small entropy deficits on the relevant supports, the corresponding effective states remain close to the Frobenius regime for almost all Trotter segments:
\begin{equation}
\norm{O\ket{\tilde\psi_{r_k}}}
  \approx
  \norm{O\ket{\tilde\psi_{r_k,\mc M}}}
  \approx
  \norm{O}_F.
  \label{eq:observation-1}
\end{equation}
That is, the observable-dependent part of the long-time 
error can be governed by the normalized Frobenius norm 
across the entire trajectory, rather than only during those 
segments whose intermediate state happens to already be 
highly entangled.

Numerical results 
confirming this behavior are shown in 
Fig.~\ref{fig:distribution}, demonstrating that the 
observable-driven suppression of Trotter error can be 
realized uniformly across all Trotter segments, rather 
than only during those segments whose intermediate states 
happen to be highly entangled \cite{zhao2025entanglement}.


\begin{figure*}[tb]
    \centering
    \subfloat[]{\includegraphics[width=0.44\textwidth]{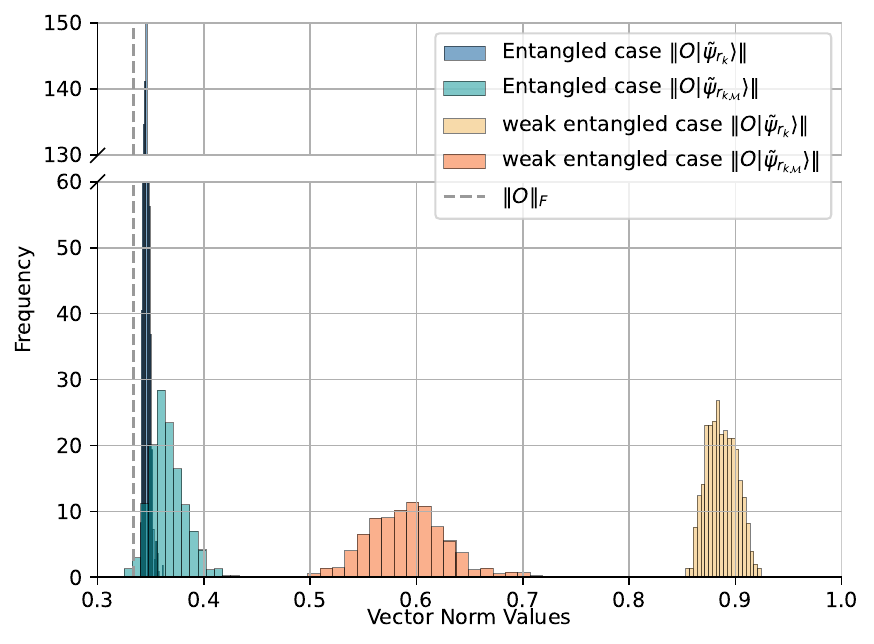}}
\hspace{0.03\textwidth}
\subfloat[]{\includegraphics[width=0.43\textwidth]{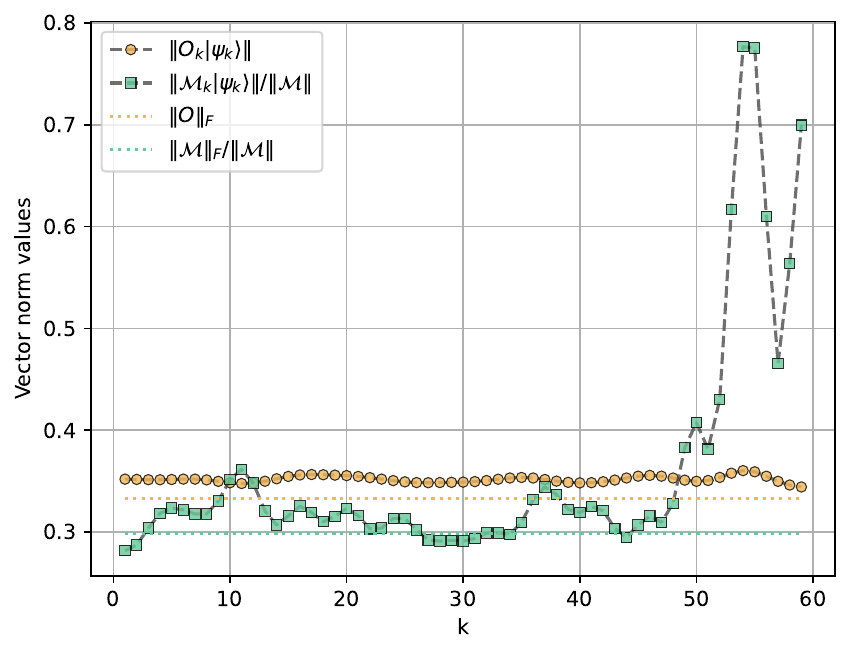}}
    \caption{
    (a) Distribution of the vector norms 
    $\norm{O\ket{\tilde\psi_{r_k}}}$ and 
    $\norm{O\ket{\tilde{\psi}_{r_k,\mc M}}}$ from 
    Lemma~\ref{main:eq:bound_for_lt}, evaluated over all 
    segments $k\in[1,r]$ with $r=1000$, $\delta t=0.1$. 
    Two initial states are compared: for $\ket{\psi} = 
    \ket{0}^{\otimes N}$, the final state 
    $U_p^r\ket{\psi}$ is highly entangled; for 
    $\ket{\psi} = \ket{+}^{\otimes N}$, it remains weakly 
    entangled. In the entangled case, both norms 
    concentrate sharply around $\norm{O}_\F$ across all 
    $r$ segments, confirming that they are nearly 
    independent of $k$, consistent with the behavior described above. In 
    the weakly entangled case, the values remain far from 
    $\norm{O}_\F$.
    (b) The vector norms $\norm{O\ket{\tilde\psi_{r_k}}}$ 
    and $\norm{\mc M\ket{\psi_k}}/\norm{\mc M}$ versus $k$ 
    for the entangled case $\ket{\psi}=\ket{0}^{\otimes N}$, 
    with $r=30$, $\delta t=0.1$. While 
    $\norm{O\ket{\tilde\psi_{r_k}}}$ remains close to 
    $\norm{O}_\F$ for all $k$, the ratio 
    $\norm{\mc M\ket{\psi_k}}/\norm{\mc M}$ deviates 
    increasingly from $\norm{\mc M}_\F/\norm{\mc M}$ as $k$ 
    grows and the intermediate state $\ket{\psi_k}=
    U_p^{r-k}\ket{\psi}$ becomes less entangled, consistent 
    with Ref.~\cite{zhao2025entanglement}. 
    This asymmetry supports the distinct segment dependence of the observable- and error-operator factors
    }
    \label{fig:distribution}
\end{figure*}

\bibliographystyle{apsrev4-2}
\bibliography{ref}

@article{srednicki1994chaos,
  title={Chaos and quantum thermalization},
  author={Srednicki, Mark},
  journal={Physical review e},
  volume={50},
  number={2},
  pages={888},
  year={1994},
  publisher={APS}
}

@article{georgescu2014quantum,
  title={Quantum simulation},
  author={Georgescu, Iulia M and Ashhab, Sahel and Nori, Franco},
  journal={Reviews of Modern Physics},
  volume={86},
  number={1},
  pages={153--185},
  year={2014},
  publisher={APS}
}

@article{lieb1972finite,
  title={The finite group velocity of quantum spin systems},
  author={Lieb, Elliott H and Robinson, Derek W},
  journal={Communications in mathematical physics},
  volume={28},
  number={3},
  pages={251--257},
  year={1972},
  publisher={Springer}
}

@article{DowlingPhysRevLett.131.180403,
  title = {Scrambling Is Necessary but Not Sufficient for Chaos},
  author = {Dowling, Neil and Kos, Pavel and Modi, Kavan},
  journal = {Phys. Rev. Lett.},
  volume = {131},
  issue = {18},
  pages = {180403},
  numpages = {6},
  year = {2023},
  month = {Nov},
  publisher = {American Physical Society},
  doi = {10.1103/PhysRevLett.131.180403},
  url = {https://link.aps.org/doi/10.1103/PhysRevLett.131.180403}
}

@article{dowling2025bridging,
  title={Bridging entanglement and magic resources within operator space},
  author={Dowling, Neil and Modi, Kavan and White, Gregory AL},
  journal={Physical Review Letters},
  volume={135},
  number={16},
  pages={160201},
  year={2025},
  publisher={APS}
}

@article{PhysRevX.15.021017,
  title = {Accuracy Guarantees and Quantum Advantage in Analog Open Quantum Simulation with and without Noise},
  author = {Kashyap, Vikram and Styliaris, Georgios and Mouradian, Sara and Cirac, J. Ignacio and Trivedi, Rahul},
  journal = {Phys. Rev. X},
  volume = {15},
  issue = {2},
  pages = {021017},
  numpages = {71},
  year = {2025},
  month = {Apr},
  publisher = {American Physical Society},
  doi = {10.1103/PhysRevX.15.021017},
  url = {https://link.aps.org/doi/10.1103/PhysRevX.15.021017}
}

@article{hosur2016chaos,
  title={Chaos in quantum channels},
  author={Hosur, Pavan and Qi, Xiao-Liang and Roberts, Daniel A and Yoshida, Beni},
  journal={Journal of High Energy Physics},
  volume={2016},
  number={2},
  pages={1--49},
  year={2016},
  publisher={Springer}
}

@article{roberts2017chaos,
  title={Chaos and complexity by design},
  author={Roberts, Daniel A and Yoshida, Beni},
  journal={Journal of High Energy Physics},
  volume={2017},
  number={4},
  pages={1--64},
  year={2017},
  publisher={Springer}
}

@article{swingle2018unscrambling,
  title={Unscrambling the physics of out-of-time-order correlators},
  author={Swingle, Brian},
  journal={Nature Physics},
  volume={14},
  number={10},
  pages={988--990},
  year={2018},
  publisher={Nature Publishing Group UK London}
}

@article{mezei2017entanglement,
  title={On entanglement spreading in chaotic systems},
  author={Mezei, M{\'a}rk and Stanford, Douglas},
  journal={Journal of High Energy Physics},
  volume={2017},
  number={5},
  pages={1--24},
  year={2017},
  publisher={Springer}
}

@article{liu2014entanglement,
  title={Entanglement tsunami: universal scaling in holographic thermalization},
  author={Liu, Hong and Suh, S Josephine},
  journal={Physical review letters},
  volume={112},
  number={1},
  pages={011601},
  year={2014},
  publisher={APS}
}

@article{banuls2011strong,
  title={Strong and weak thermalization of infinite nonintegrable quantum systems},
  author={Ba{\~n}uls, Mari Carmen and Cirac, J Ignacio and Hastings, Matthew B},
  journal={Physical review letters},
  volume={106},
  number={5},
  pages={050405},
  year={2011},
  publisher={APS}
}

@article{zhao2025entanglement,
  title={Entanglement accelerates quantum simulation},
  author={Zhao, Qi and Zhou, You and Childs, Andrew M},
  journal={Nature Physics},
  volume={21},
  number={8},
  pages={1338--1345},
  year={2025},
  publisher={Nature Publishing Group UK London}
}

@article{footnote1,
   title={},
author={},
   volume={},
   ISSN={},
   url={},
   DOI={},
   number={},
   journal={
For fixed product-formula order, the operators $M_j$ are obtained
explicitly from the finite set of nonvanishing local nested
commutators generated by the Hamiltonian decomposition. We state this
leading-order structure first and give its detailed derivation in
Sec.~II of the Supplemental Material},
   }

@article{chailds2021TrottertheoryPhysRevX.11.011020,
  title = {Theory of Trotter Error with Commutator Scaling},
  author = {Childs, Andrew M. and Su, Yuan and Tran, Minh C. and Wiebe, Nathan and Zhu, Shuchen},
  journal = {Phys. Rev. X},
  volume = {11},
  issue = {1},
  pages = {011020},
  numpages = {49},
  year = {2021},
  month = {Feb},
  publisher = {American Physical Society},
  doi = {10.1103/PhysRevX.11.011020},
  url = {https://link.aps.org/doi/10.1103/PhysRevX.11.011020}
}

@article{lloyd1996universal,
  title={Universal quantum simulators},
  author={Lloyd, Seth},
  journal={Science},
  volume={273},
  number={5278},
  pages={1073--1078},
  year={1996},
  publisher={American Association for the Advancement of Science}
}

@article{zhao2022hamiltonian,
  title={Hamiltonian simulation with random inputs},
  author={Zhao, Qi and Zhou, You and Shaw, Alexander F and Li, Tongyang and Childs, Andrew M},
  journal={Physical Review Letters},
  volume={129},
  number={27},
  pages={270502},
  year={2022},
  publisher={APS}
}

@article{childs2019nearly,
  title={Nearly optimal lattice simulation by product formulas},
  author={Childs, Andrew M and Su, Yuan},
  journal={Physical review letters},
  volume={123},
  number={5},
  pages={050503},
  year={2019},
  publisher={APS}
}

@incollection{feynman2018simulating,
  title={Simulating physics with computers},
  author={Feynman, Richard P},
  booktitle={Feynman and computation},
  pages={133--153},
  year={2018},
  publisher={cRc Press}
}

@article{buluta2009quantum,
  title={Quantum simulators},
  author={Buluta, Iulia and Nori, Franco},
  journal={Science},
  volume={326},
  number={5949},
  pages={108--111},
  year={2009},
  publisher={American Association for the Advancement of Science}
}

@article{childs2018toward,
  title={Toward the first quantum simulation with quantum speedup},
  author={Childs, Andrew M and Maslov, Dmitri and Nam, Yunseong and Ross, Neil J and Su, Yuan},
  journal={Proceedings of the National Academy of Sciences},
  volume={115},
  number={38},
  pages={9456--9461},
  year={2018},
  publisher={National Academy of Sciences}
}

@article{tran2020hierarchylightcone,
  title={Hierarchy of linear light cones with long-range interactions},
  author={Tran, Minh C and Chen, Chi-Fang and Ehrenberg, Adam and Guo, Andrew Y and Deshpande, Abhinav and Hong, Yifan and Gong, Zhe-Xuan and Gorshkov, Alexey V and Lucas, Andrew},
  journal={Physical Review X},
  volume={10},
  number={3},
  pages={031009},
  year={2020},
  publisher={APS}
}

@misc{yu2024observabledrivenspeedupsquantumsimulations,
      title={Observable-Driven Speed-ups in Quantum Simulations}, 
      author={Wenjun Yu and Jue Xu and Qi Zhao},
      year={2024},
      eprint={2407.14497},
      archivePrefix={arXiv},
      primaryClass={quant-ph},
      url={https://arxiv.org/abs/2407.14497}, 
}

@article{PhysRevLett.133.010603Zhaohongzhen,
  title = {Adaptive Trotterization for Time-Dependent Hamiltonian Quantum Dynamics Using Piecewise Conservation Laws},
  author = {Zhao, Hongzheng and Bukov, Marin and Heyl, Markus and Moessner, Roderich},
  journal = {Phys. Rev. Lett.},
  volume = {133},
  issue = {1},
  pages = {010603},
  numpages = {8},
  year = {2024},
  month = {Jul},
  publisher = {American Physical Society},
  doi = {10.1103/PhysRevLett.133.010603},
  url = {https://link.aps.org/doi/10.1103/PhysRevLett.133.010603}
}

@article{heyl2019quantum,
  title={Quantum localization bounds Trotter errors in digital quantum simulation},
  author={Heyl, Markus and Hauke, Philipp and Zoller, Peter},
  journal={Science advances},
  volume={5},
  number={4},
  pages={eaau8342},
  year={2019},
  publisher={American Association for the Advancement of Science}
}

@article{chen2024average,
  title={Average-case speedup for product formulas},
  author={Chen, Chi-Fang and Brand{\~a}o, Fernando GSL},
  journal={Communications in Mathematical Physics},
  volume={405},
  number={2},
  pages={32},
  year={2024},
  publisher={Springer}
}

@article{xu2024scrambling,
  title={Scrambling dynamics and out-of-time-ordered correlators in quantum many-body systems},
  author={Xu, Shenglong and Swingle, Brian},
  journal={PRX quantum},
  volume={5},
  number={1},
  pages={010201},
  year={2024},
  publisher={APS}
}

@article{georgescu2014quantum_anolog,
  title={Quantum simulation},
  author={Georgescu, Iulia M and Ashhab, Sahel and Nori, Franco},
  journal={Reviews of Modern Physics},
  volume={86},
  number={1},
  pages={153--185},
  year={2014},
  publisher={APS}
}

@article{nandkishore2015manybodylocalize,
  title={Many-body localization and thermalization in quantum statistical mechanics},
  author={Nandkishore, Rahul and Huse, David A},
  journal={Annu. Rev. Condens. Matter Phys.},
  volume={6},
  number={1},
  pages={15--38},
  year={2015},
  publisher={Annual Reviews}
}

@article{landsman2019verified,
  title={Verified quantum information scrambling},
  author={Landsman, Kevin A and Figgatt, Caroline and Schuster, Thomas and Linke, Norbert M and Yoshida, Beni and Yao, Norman Y and Monroe, Christopher},
  journal={Nature},
  volume={567},
  number={7746},
  pages={61--65},
  year={2019},
  publisher={Nature Publishing Group UK London}
}

@article{DQSPhysRevResearch.5.L022003,
  title = {Simulating time evolution on distributed quantum computers},
  author = {Buessen, Finn Lasse and Segal, Dvira and Khait, Ilia},
  journal = {Phys. Rev. Res.},
  volume = {5},
  issue = {2},
  pages = {L022003},
  numpages = {6},
  year = {2023},
  month = {Apr},
  publisher = {American Physical Society},
  doi = {10.1103/PhysRevResearch.5.L022003},
  url = {https://link.aps.org/doi/10.1103/PhysRevResearch.5.L022003}
}

@InProceedings{cai_et_al:LIPIcs.TQC.2024.2,
  author =	{Cai, Yiyi and Tong, Yu and Preskill, John},
  title =	{{Stochastic Error Cancellation in Analog Quantum Simulation}},
  booktitle =	{19th Conference on the Theory of Quantum Computation, Communication and Cryptography (TQC 2024)},
  pages =	{2:1--2:15},
  series =	{Leibniz International Proceedings in Informatics (LIPIcs)},
  ISBN =	{978-3-95977-328-7},
  ISSN =	{1868-8969},
  year =	{2024},
  volume =	{310},
  editor =	{Magniez, Fr\'{e}d\'{e}ric and Grilo, Alex Bredariol},
  publisher =	{Schloss Dagstuhl -- Leibniz-Zentrum f{\"u}r Informatik},
  address =	{Dagstuhl, Germany},
  URL =		{https://drops.dagstuhl.de/entities/document/10.4230/LIPIcs.TQC.2024.2},
  URN =		{urn:nbn:de:0030-drops-206720},
  doi =		{10.4230/LIPIcs.TQC.2024.2}
}

@article{Shenker_2014otoc,
   title={Black holes and the butterfly effect},
   volume={2014},
   ISSN={1029-8479},
   url={http://dx.doi.org/10.1007/JHEP03(2014)067},
   DOI={10.1007/jhep03(2014)067},
   number={3},
   journal={Journal of High Energy Physics},
   publisher={Springer Science and Business Media LLC},
   author={Shenker, Stephen H. and Stanford, Douglas},
   year={2014},
   month=mar }

@article{suzuki1991general,
  title={General theory of fractal path integrals with applications to many-body theories and statistical physics},
  author={Suzuki, Masuo},
  journal={Journal of mathematical physics},
  volume={32},
  number={2},
  pages={400--407},
  year={1991},
  publisher={American Institute of Physics}
}

@article{Nozaki_2014OE,
   title={Quantum Entanglement of Local Operators in Conformal Field Theories},
   volume={112},
   ISSN={1079-7114},
   url={http://dx.doi.org/10.1103/PhysRevLett.112.111602},
   DOI={10.1103/physrevlett.112.111602},
   number={11},
   journal={Physical Review Letters},
   publisher={American Physical Society (APS)},
   author={Nozaki, Masahiro and Numasawa, Tokiro and Takayanagi, Tadashi},
   year={2014},
   month=mar }

@misc{low2010pseudorandomnesslearningquantumcomputation,
      title={Pseudo-randomness and Learning in Quantum Computation}, 
      author={Richard A. Low},
      year={2010},
      eprint={1006.5227},
      archivePrefix={arXiv},
      primaryClass={quant-ph},
      url={https://arxiv.org/abs/1006.5227}, 
}

@misc{harrow2013churchsymmetricsubspace,
      title={The Church of the Symmetric Subspace}, 
      author={Aram W. Harrow},
      year={2013},
      eprint={1308.6595},
      archivePrefix={arXiv},
      primaryClass={quant-ph},
      url={https://arxiv.org/abs/1308.6595}, 
}

@article{LiPhysRevA.110.062614,
  title = {Principal Trotter observation error with truncated commutators},
  author = {Li, Langyu},
  journal = {Phys. Rev. A},
  volume = {110},
  issue = {6},
  pages = {062614},
  numpages = {12},
  year = {2024},
  month = {Dec},
  publisher = {American Physical Society},
  doi = {10.1103/PhysRevA.110.062614},
  url = {https://link.aps.org/doi/10.1103/PhysRevA.110.062614}
}

@article{_Anthony_Chen_2023,
   title={Speed limits and locality in many-body quantum dynamics},
   volume={86},
   ISSN={1361-6633},
   url={http://dx.doi.org/10.1088/1361-6633/acfaae},
   DOI={10.1088/1361-6633/acfaae},
   number={11},
   journal={Reports on Progress in Physics},
   publisher={IOP Publishing},
   author={(Anthony) Chen, Chi-Fang and Lucas, Andrew and Yin, Chao},
   year={2023},
   month=sep, pages={116001} }

@article{Bravyi_2007,
   title={Upper bounds on entangling rates of bipartite Hamiltonians},
   volume={76},
   ISSN={1094-1622},
   url={http://dx.doi.org/10.1103/PhysRevA.76.052319},
   DOI={10.1103/physreva.76.052319},
   number={5},
   journal={Physical Review A},
   publisher={American Physical Society (APS)},
   author={Bravyi, Sergey},
   year={2007},
   month=nov }

@misc{feng2024distributedquantumsimulation,
      title={Distributed Quantum Simulation}, 
      author={Tianfeng Feng and Jue Xu and Wenjun Yu and Zekun Ye and Penghui Yao and Qi Zhao},
      year={2024},
      eprint={2411.02881},
      archivePrefix={arXiv},
      primaryClass={quant-ph},
      url={https://arxiv.org/abs/2411.02881}, 
}

@article{MBLRevModPhys.91.021001,
  title = {Colloquium: Many-body localization, thermalization, and entanglement},
  author = {Abanin, Dmitry A. and Altman, Ehud and Bloch, Immanuel and Serbyn, Maksym},
  journal = {Rev. Mod. Phys.},
  volume = {91},
  issue = {2},
  pages = {021001},
  numpages = {26},
  year = {2019},
  month = {May},
  publisher = {American Physical Society},
  doi = {10.1103/RevModPhys.91.021001},
  url = {https://link.aps.org/doi/10.1103/RevModPhys.91.021001}
}

@article{surace2022fermionic,
  title={Fermionic Gaussian states: an introduction to numerical approaches},
  author={Surace, Jacopo and Tagliacozzo, Luca},
  journal={SciPost Physics Lecture Notes},
  pages={054},
  year={2022}
}

@article{PhysRevA.43.2046ETH91,
  title = {Quantum statistical mechanics in a closed system},
  author = {Deutsch, J. M.},
  journal = {Phys. Rev. A},
  volume = {43},
  issue = {4},
  pages = {2046--2049},
  numpages = {0},
  year = {1991},
  month = {Feb},
  publisher = {American Physical Society},
  doi = {10.1103/PhysRevA.43.2046},
  url = {https://link.aps.org/doi/10.1103/PhysRevA.43.2046}
}

@article{PhysRevE.50.888ETH96,
  title = {Chaos and quantum thermalization},
  author = {Srednicki, Mark},
  journal = {Phys. Rev. E},
  volume = {50},
  issue = {2},
  pages = {888--901},
  numpages = {0},
  year = {1994},
  month = {Aug},
  publisher = {American Physical Society},
  doi = {10.1103/PhysRevE.50.888},
  url = {https://link.aps.org/doi/10.1103/PhysRevE.50.888}
}

\clearpage
\onecolumngrid
\setcounter{section}{0}
\setcounter{secnumdepth}{2}
\setcounter{equation}{0}
\setcounter{figure}{0}
\setcounter{table}{0}
\setcounter{theorem}{0}
\setcounter{lemma}{0}
\setcounter{corollary}{0}
\setcounter{observation}{0}
\numberwithin{equation}{section}
\renewcommand{\theequation}{S\arabic{section}.\arabic{equation}}
\renewcommand{\thefigure}{S\arabic{figure}}
\renewcommand{\thetable}{S\arabic{table}}
\renewcommand{\thetheorem}{S\arabic{theorem}}
\renewcommand{\thelemma}{S\arabic{lemma}}
\renewcommand{\thecorollary}{S\arabic{corollary}}
\renewcommand{\theobservation}{S\arabic{observation}}
\renewcommand{\theHsection}{supp.\arabic{section}}
\renewcommand{\theHequation}{supp.\arabic{section}.\arabic{equation}}
\renewcommand{\theHfigure}{supp.\arabic{figure}}
\renewcommand{\theHtable}{supp.\arabic{table}}
\renewcommand{\theHtheorem}{supp.\arabic{theorem}}
\renewcommand{\theHlemma}{supp.\arabic{lemma}}
\renewcommand{\theHcorollary}{supp.\arabic{corollary}}
\renewcommand{\theHobservation}{supp.\arabic{observation}}

\begin{center}
  {\large\bfseries Supplementary Materials of
  ``Trotterization, Operator Scrambling, and Entanglement''\par}
  \medskip
  Tianfeng Feng, Yue Cao, and Qi Zhao\par
  \smallskip
  November 2, 2026
\end{center}

\contentsmargin{2.55em}
\dottedcontents{section}[3.8em]{}{2.3em}{1pc}
\dottedcontents{subsection}[6.1em]{}{2.2em}{1pc}
\dottedcontents{subsubsection}[9em]{}{2.1em}{1pc}
\startcontents[supp]
\noindent\textbf{Supplementary Materials Contents}\par\medskip
\printcontents[supp]{}{1}{}
\medskip

\section{Brief introduction  of Trotterization}
In this work, we mainly focus on product formula (PF) algorithms \cite{suzuki1991general,chailds2021TrottertheoryPhysRevX.11.011020}.
For a short evolution time $\delta t$, the first-order product formula (PF1) algorithm for a Hamiltonian $\sum_{l=1}^LH_l$ applies the unitary operation
\begin{align}
    U_1(\delta t):=e^{-iH_1\delta t}e^{-iH_2\delta t}\cdots e^{-iH_L\delta t}=\rprod_l e^{-iH_l \delta t}.
\end{align}
Here the right arrow indicates the product is in the order of increasing indices (i.e., for $l=1,2,\ldots,L$). Second-order product formulas (PF$2$)
can be obtained by combining evolutions in both increasing and decreasing orders of indices, with
\begin{align}
  U_2(\delta t):= \rprod_l e^{-iH_l\delta t/2} \lprod_l e^{-iH_l\delta t/2},
\end{align}
where the left arrow indicates the product in decreasing order (i.e., for $l=L,L-1,\ldots,1$).

More generally, Suzuki constructed $p$th-order product formulas (PF$p$ for even $p=2k$) recursively from the second-order formula as~\cite{suzuki1991general}
\begin{align}
    U_{2k}(\delta t)&= [U_{2k-2}(p_k \delta t)]^2 U_{2k-2}((1-4p_k)\delta t) [U_{2k-2}(p_k \delta t)]^2,
\end{align}
where $p_k:=\frac{1}{4-4^{1/(2k-1)}}$ and $k>1$ \cite{suzuki1991general}.
Overall, we have $S=2\cdot 5^{k-1}$ stages (operators of the form $U_1=\rprod_l e^{-iH_l\delta t}$ or its reverse ordering $\lprod_l e^{-iH_l\delta t}$) and the evolution can be rewritten as \cite{chailds2021TrottertheoryPhysRevX.11.011020}
\begin{equation}\label{eq:pf2k}
  U_{2k}(\delta t)=  \prod_{s=1}^{S} \prod_{l=1}^{L} e^{-i \delta t a_s H_{\pi_s(l)}},
\end{equation}
where each $\pi_s$ is the identity permutation or reversal permutation and $\delta t a_s$ is the simulation time for stage $s$.

To simulate the evolution of a quantum system for a long time, we divide the entire duration into many smaller time segments. Each segment represents a short-time evolution that can be simulated with small error. The total error can be upper-bounded using the triangle inequality as the sum of the Trotter errors for each segment, which we also refer to as Trotter steps.


\section{Analysis of the leading-order multiplicative error in the \texorpdfstring{$\text{PF}p$}{PFp}}

{We first state the main calculation result of this section before giving its
derivation. For a $p$th-order product formula, with the multiplicative
error defined consistently through
$  U_p(\delta t)
    =
    U_0(\delta t)
    (I+\mathcal M(\delta t)),$
the leading multiplicative error can be written as
\begin{equation}
\begin{split}
 \mathcal M(\delta t)
    =
    M\,\delta t^{p+1}
    +
    \bigO(\delta t^{p+2})
     &=  \sum_{i=1,2} \sum_{\vec{j}_{p+1}\in \Gamma_i} \frac1{p+1} c'_{\vec{j}_{p+1}} \delta t^{p+1}  N_{\vec{j}_{p+1}} + \bigO(\delta t^{p+2})\\
 &= \sum_j M_j  \delta t^{p+1} +\bigO(\delta t^{p+2}),
\label{SM:leading-M-summary}
    \end{split}
\end{equation}
Here the operators $M_j$ are linear combinations of the nested
commutators
$
    N_{\vec j_{p+1}}
    =
    [H_{j_1},
    [H_{j_2},
    \ldots,
    [H_{j_p},H_{j_{p+1}}]
    \ldots]],
$
generated by the Hamiltonian terms entering the product formula. The values $c_{\vec{j}_{p+1}}$ are real coefficients that are functions of $\vec{j}_{p+1}$ and $p$, satisfying $|c_{\vec{j}_{p+1}}|\le 1$ (see below). 
Equivalently, the leading coefficient can be written as
$
    M
    =
    \sum_{\vec j_{p+1}}
    \kappa_{\vec j_{p+1}}
    N_{\vec j_{p+1}},
$
where the coefficients $\kappa_{\vec j_{p+1}}$ are fixed by the
chosen product formula.
For a local Hamiltonian $H=\sum_l H_l$, these contributions are
themselves local: a nested commutator vanishes whenever its local
Hamiltonian terms do not form an overlapping sequence of supports.
Thus, at fixed product-formula order, the nonzero terms contributing
to $M$ can be constructed directly from local Hamiltonian terms,
without constructing the full many-body operator.
 We now derive these expressions
and the remainder estimate explicitly.}

We now analyze the leading error terms $\mathcal{M}$ of $\text{PF}p$ in detail.
First define an order for pairs $(s,l)$ of stage indices $s \in \{1,\ldots,S\}$ and term indices $l \in \{1,\ldots,L\}$.
Let $(s,l)\prec (s',l')$ when $s<s'$ or $s=s',l<l'$. Furthermore, let $(s,l)\preceq (s',l')$ when $s<s'$ or $s=s',l\le l'$. Set $e^{-iH\delta t}:=U_0(\delta t)$, 
according to Theorem 3 in Ref.~\cite{chailds2021TrottertheoryPhysRevX.11.011020}, the additive error $\mathscr{E} (\delta t)$
can be expressed as
\begin{equation}\label{app:aderror}
\mathscr{E}(\delta t):=U_0(\delta t)-U_p(\delta t)=\int_0^{\delta t} \d\tau_1 \, e^{-i( \delta t-\tau_1)H} U_p(\tau_1) \mathscr{N}(\tau_1),
\end{equation}
where $U_p$ is the $p$th-order Trotter formula and
\begin{equation}
\begin{aligned}
 \mathscr{N}(\tau_1)&= \sum_{(s,l)}
 \rprod_{(s',l')\prec(s,l)}
e^{i\tau_1 a_s H_{\pi_{s'}(l')}} \left(- ia_s H_{\pi_s(l)}\right)       \lprod_{(s',l')\prec(s,l)}
e^{-i\tau_1 a_s H_{\pi_{s'}(l')}}  \\
&\quad-  \rprod_{(s',l')}
e^{i\tau_1 a_s H_{\pi_{s'}(l')}} (-iH) \lprod_{(s',l')}
e^{-i\tau_1 a_s H_{\pi_{s'}(l')}},
\end{aligned}
\end{equation}
where $\mathscr{N}(\tau_1)=\mathcal O(\tau^p_1)$.
Since we are concerned about the  multiplicative error, one can rewrite 
    \begin{equation}\label{app:mlterror}
U_p(\delta t)
=U_0(\delta t) \big[I-\int_0^{\delta t} \d\tau_1 \, e^{i\tau_1 H} U_p(\tau_1) \mathscr{N}(\tau_1)\big]=U_0(\delta t)(I+\mathcal{M}(\delta t)).
\end{equation}
Now, the multiplicative error of PF$p$ can be defined as 
\begin{equation}
    \mathcal{M}(\delta t):= -\int_0^{\delta t} \d\tau_1 \, e^{i\tau_1 H} U_p(\tau_1) \mathscr{N}(\tau_1) .
\end{equation}
Here we define the vector $\vec{j}_{p+1}=(j_1,j_2,\dots,j_{p+1})$ with $p+1$ entries, $j_1,j_2,\dots,j_{p+1} \in \{(s,l): s\in \{1,\dots,S\},l\in\{1,\dots,L\}\}$, and the corresponding nested commutators as
\begin{equation}\label{SM:Nnest}
 N_{\vec{j}_{p+1}}=[H_{j_1},[H_{j_2},\dots,[H_{j_p},H_{j_{p+1}}]]].
\end{equation}
According to Theorem 5 in Ref.~\cite{chailds2021TrottertheoryPhysRevX.11.011020}, we can further write
\begin{align}
 \mathscr{N}(\tau_1)&=  \int_0^{\tau_1}\d\tau_2    \sum_{i=1,2} \sum_{\vec{j}_{p+1}\in \Gamma_i}   (\tau_1-\tau_2)^{q(\vec{j}_{p+1})-1}\tau_1^{p-q(\vec{j}_{p+1})} c_{\vec{j}_{p+1}} F_{\vec{j}_{p+1}}^{\dagger} N_{\vec{j}_{p+1}} F_{\vec{j}_{p+1}}, \label{eq:Nint}
 \end{align}
where we use the following definitions. Let
\begin{equation}
    \begin{aligned}
\Gamma_1&:=\{(j_1,j_2,\dots,j_{p+1}) : j_1\preceq j_2\preceq\dots \preceq j_{p+1}\},    \\
\Gamma_2&:=\{(j_1,j_2,\dots,j_{p+1}) : j_1\preceq j_2\dots \preceq j_{p}, \, j_{p+1}=(1,l_{p+1})\} .
    \end{aligned}
\end{equation}
The values $c_{\vec{j}_{p+1}}$ are real coefficients that are functions of $\vec{j}_{p+1}$ and $p$, satisfying $|c_{\vec{j}_{p+1}}|\le 1$. The function
$q(\vec{j}_{p+1})$
is the maximal number $q$ satisfying $j_1=j_2=\dots=j_q$.
The operators $F_{\vec{j}_{p+1}}$ are unitary, constructed as products of terms of the form $e^{-iH_l\tau_1}$.
Plugging \cref{eq:Nint} into \cref{app:aderror}, we have
  \begin{align}
 \mathcal{M}(\delta t)&=-\int_0^{\delta t} \d\tau_1 \int_0^{\tau_1}\d\tau_2
 \sum_{i=1,2} \sum_{\vec{j}_{p+1}\in \Gamma_i} (\tau_1-\tau_2)^{q(\vec{j}_{p+1})-1}\tau_1^{p-q(\vec{j}_{p+1})} c_{\vec{j}_{p+1}}  e^{i\tau_1 H} U_p(\tau_1)F_{\vec{j}_{p+1}}^{\dagger} N_{\vec{j}_{p+1}} F_{\vec{j}_{p+1}}.\label{Eint}
 \end{align}

Since $F_{\vec{j}_{p+1}}$ is a product of short-time evolutions, we can expand it as
\begin{align}
   F_{\vec{j}_{p+1}}^{\dagger} N_{\vec{j}_{p+1}} F_{\vec{j}_{p+1}}= N_{\vec{j}_{p+1}} + R_{\vec{j}_{p+1}}
\end{align}
where the nested commutator $N_{\vec{j}_{p+1}}$ is the $0$th-order term and $R_{\vec{j}_{p+1}}$ represents the remaining higher-order terms.
To realize this decomposition,
we repeatedly apply the formula
\begin{align}
    e^{iA\tau_1}B e^{-iA\tau_1}= B + \int_0^{\tau_1}  \d\tau_2 e^{iA(\tau_1-\tau_2)} [A, B] e^{-iA(\tau_1-\tau_2)},
\end{align}
satisfying
\begin{align}
    \| e^{iA\tau_1}B e^{-iA\tau_1}&- B\| \le \|[A,B]\|\tau_1,
\end{align}
which accounts for the effect of a short-time evolution generated by $A$ on $B$. For a product of short-time evolutions, each time we apply it to the innermost layer, giving
\begin{align}
 \| e^{iA_1\tau_1}\cdots e^{iA_{s-1}\tau_1}e^{iA_s\tau_1} B  e^{-iA_s\tau_1}\cdots e^{-iA_{s-1}\tau_1}e^{-iA_1\tau_1}- e^{iA_1\tau_1}\cdots e^{iA_{s-1}\tau_1} B e^{-iA_{s-1}\tau_1}e^{-iA_1\tau_1}\|\le \|[A_s,B]\|\tau_1.
\end{align}
Iterating, one has
\begin{align}
\left \| e^{iA_1\tau_1}\cdots e^{iA_{s-1}\tau_1}e^{iA_s\tau_1} B  e^{-iA_s\tau_1}\cdots e^{-iA_{s-1}\tau_1}e^{-iA_1\tau_1}- B \right\|\le \sum_s \|[A_s,B]\| \tau_1.
\end{align}
By replacing $B$ with $N_{\vec{j}_{p+1}} $, we thus can control the remaining term as $\|R_{\vec{j}_{p+1}}\|= \mathcal O(\sum_l \|[H_{l}, N_{\vec{j}_{p+1}}]\|\tau_1)$.

Replacing the dressed operator $F_{\vec{j}_{p+1}}^{\dagger} N_{\vec{j}_{p+1}} F_{\vec{j}_{p+1}}$ for $N_{\vec{j}_{p+1}}$ in $\mathscr{E}(\delta t)$ of \cref{Eint} by the bare nested commutator, we define
\begin{equation}
\begin{aligned}
 \mathcal{M}^*(\delta t):=-\int_0^{\delta t} \d\tau_1 \int_0^{\tau_1}\d\tau_2
 \sum_{i=1,2} \sum_{\vec{j}_{p+1}\in \Gamma_i} (\tau_1-\tau_2)^{q(\vec{j}_{p+1})-1}\tau_1^{p-q(\vec{j}_{p+1})} c_{\vec{j}_{p+1}}  e^{i\tau_1 H} U_p(\tau_1)N_{\vec{j}_{p+1}}.
 \end{aligned}\label{Eq.additive*}
\end{equation}
Since the higher-order terms of multiplicative error of the product-formula approximation are given by the difference of $\mathcal{M}$ and $\mathcal{M}^*$, we can bound it as follows:
\begin{equation}
    \begin{aligned}
  \|   \mathcal{M}_{\re}(\delta t)\|&= \|   \mathcal{M}(\delta t)-  \mathcal{M}^*(\delta t)\|\\
  &=\left\|\int_0^{\delta t} \!\d\tau_1 \int_0^{\tau_1} \!\d\tau_2
 \sum_{i=1,2} \sum_{\vec{j}_{p+1}\in \Gamma_i} (\tau_1-\tau_2)^{q(\vec{j}_{p+1})-1}\tau_1^{p-q(\vec{j}_{p+1})} c_{\vec{j}_{p+1}}  e^{i\tau_1 H} U_p(\tau_1)(F_{\vec{j}_{p+1}}^{\dagger} \! N_{\vec{j}_{p+1}} \! F_{\vec{j}_{p+1}} \!\! -N_{\vec{j}_{p+1}})\right\|\\
  &=\left\|\int_0^{\delta t} \d\tau_1 \int_0^{\tau_1}\d\tau_2
 \sum_{i=1,2} \sum_{\vec{j}_{p+1}\in \Gamma_i} (\tau_1-\tau_2)^{q(\vec{j}_{p+1})-1}\tau_1^{p-q(\vec{j}_{p+1})} c_{\vec{j}_{p+1}}  e^{i\tau_1 H} U_p(\tau_1)R_{\vec{j}_{p+1}}\right\|\\
 &\leq \int_0^{\delta t} \d\tau_1 \int_0^{\tau_1}\d\tau_2
 \sum_{i=1,2} \sum_{\vec{j}_{p+1}\in \Gamma_i} (\tau_1-\tau_2)^{q(\vec{j}_{p+1})-1}\tau_1^{p-q(\vec{j}_{p+1})} |c_{\vec{j}_{p+1}}| \left\|e^{i\tau_1 H} U_p(\tau_1)R_{\vec{j}_{p+1}}\right\|\\
 &= \int_0^{\delta t} \d\tau_1 \int_0^{\tau_1}\d\tau_2
 \sum_{i=1,2} \sum_{\vec{j}_{p+1}\in \Gamma_i} (\tau_1-\tau_2)^{q(\vec{j}_{p+1})-1}\tau_1^{p-q(\vec{j}_{p+1})} |c_{\vec{j}_{p+1}}| \left\|R_{\vec{j}_{p+1}}\right\| \\
  &=\mathcal O\left(\sum_{l,\vec{j}_{p+1}} \|[H_{l} , N_{\vec{j}_{p+1}}]\| \delta t^{p+2}\right)\\
 &= \mathcal O\left(\sum_{l_1,\dots,l_{p+2}=1}^L
\left\|[H_{l_1},[H_{l_2},\dots[H_{l_{p+1}},H_{l_{p+2}}]]] \right\| \delta t^{p+2} \right).
\end{aligned}
\end{equation}
Here the first inequality uses the triangle inequality for the spectral norm, and the fourth line uses the fact that $\|VA\|=\|A\|$ for any unitary $V$ and operator $A$.


We aim to derive the leading error of $\mathcal{M}(\delta t)=\mathcal{M}^*(\delta t) +\mathcal{M}_{\text{re}}(\delta t)$.
However, $\mathcal{M}^*(\delta t)$ is defined in \cref{Eq.additive*} still complicated.
To simplify $ \mathcal{M}^*(\delta t)$, we need to expand $e^{i\tau_1 H} U_p(\tau_1)$.
Since $U_p(\tau_1)=e^{-iH\tau_1}(I+\mathcal{M}(\tau_1))$, $$e^{i\tau_1 H} U_p(\tau_1)=I+\mathcal{M}(\tau_1).$$
Substituting it into $\mathcal{M}^*(\delta t)$, we have 

\begin{equation}\label{SMleadingE2}
    \begin{aligned}
 \mathcal{M}^*(\delta t) &= 
-\int_0^{\delta t} \d\tau_1 \int_0^{\tau_1}\d\tau_2
 \sum_{i=1,2} \sum_{\vec{j}_{p+1}\in \Gamma_i} (\tau_1-\tau_2)^{q(\vec{j}_{p+1})-1}\tau_1^{p-q(\vec{j}_{p+1})} c_{\vec{j}_{p+1}}  e^{i\tau_1 H} U_p(\tau_1)N_{\vec{j}_{p+1}} \\
 & =- \int_0^{\delta t} \d\tau_1 \int_0^{\tau_1}\d\tau_2 \sum_{i=1,2} \sum_{\vec{j}_{p+1}\in \Gamma_i}  (\tau_1-\tau_2)^{q(\vec{j}_{p+1})-1}\tau_1^{p-q(\vec{j}_{p+1})} c_{\vec{j}_{p+1}} [I+\mathcal{M}(\tau_1)] N_{\vec{j}_{p+1}} \\
 &=  \int_0^{\delta t} \d\tau_1  \sum_{i=1,2} \sum_{\vec{j}_{p+1}\in \Gamma_i}  c'_{\vec{j}_{p+1}}\tau_1^{p}
 N_{\vec{j}_{p+1}}   + \int_0^{\delta t} \d\tau_1  \sum_{i=1,2} \sum_{\vec{j}_{p+1}\in \Gamma_i}  c'_{\vec{j}_{p+1}}\tau_1^{p} \mathcal{M}(\tau_1)
 N_{\vec{j}_{p+1}}  \\
 &=  \sum_{i=1,2} \sum_{\vec{j}_{p+1}\in \Gamma_i} \frac1{p+1} c'_{\vec{j}_{p+1}} \delta t^{p+1}  N_{\vec{j}_{p+1}} + \mathcal{M}^{*}_{\text{re}}(\delta t)\\
 &= \sum_j M_j \delta t^{p+1} + \mathcal{M}^{*}_{\text{re}}(\delta t),
 \end{aligned}
\end{equation}
where $\mathcal{M}^{*}_{\text{re}}(\delta t)= \int_0^{\delta t} \d\tau_1  \sum_{i=1,2} \sum_{\vec{j}_{p+1}\in \Gamma_i}  c'_{\vec{j}_{p+1}}\tau_1^{p} \mathcal{M}(\tau_1)
 N_{\vec{j}_{p+1}} $.
For the third line, denote the result of integrating over $\tau_2$ as $F(\tau_1,\vec{j}_{p+1}):=\int_0^{\tau_1}\d\tau_2 (\tau_1-\tau_2)^{q(\vec{j}_{p+1})-1}\tau_1^{p-q(\vec{j}_{p+1})} c_{\vec{j}_{p+1}}$. It is straightforward to check that, for all $\vec{j}_{p+1}$, $F(\tau_1,\vec{j}_{p+1})=c'_{\vec{j}_{p+1}}\tau_1^{p}$ for some constant $c'_{\vec{j}_{p+1}}$. Here we also incorporate the negative one coefficient into the new coefficient $c'_{\vec{j}_{p+1}}$.
In the fourth line, we simply integrate $\tau_1^p$. In the last line, we decompose $\frac1{p+1} \sum c'_{\vec{j}_{p+1}} N_{\vec{j}_{p+1}} $ into local operators $M_j$. 

Now we need to bound the $\mathcal{M}^{*}_{\text{re}}(\delta t)$. Specifically, 

\begin{equation}
    \begin{split}
    \norm{\mathcal{M}^{*}_{\text{re}}(\delta t)}  &= \norm{ \int_0^{\delta t} \d\tau_1  \sum_{i=1,2} \sum_{\vec{j}_{p+1}\in \Gamma_i}  c'_{\vec{j}_{p+1}}\tau_1^{p} \mathcal{M}(\tau_1)
 N_{\vec{j}_{p+1}} } \\
    &\le \int_0^{\delta t} \d\tau_1  \sum_{i=1,2} \sum_{\vec{j}_{p+1}\in \Gamma_i}  |c'_{\vec{j}_{p+1}}|\tau_1^{p} 
\norm{\mathcal{M}(\tau_1)} \norm{ N_{\vec{j}_{p+1}} }\\
&=\int_0^{\delta t} \d\tau_1  \sum_{i=1,2} \sum_{\vec{j}_{p+1}\in \Gamma_i}  |c'_{\vec{j}_{p+1}}|\tau_1^{p} 
\bigO(\alpha_{p+1} \tau_1^{p+1})\norm{ N_{\vec{j}_{p+1}} }\\
&=  \left( \sum_{i=1,2} \sum_{\vec{j}_{p+1}\in \Gamma_i}  \frac{1}{2p+2}|c'_{\vec{j}_{p+1}}|
\norm{ N_{\vec{j}_{p+1}} } \right )\bigO(\alpha_{p+1})\delta t^{2p+2} \\
&=\bigO(\alpha^2_{p+1} \delta t^{2p+2}).
    \end{split}
\end{equation}
For the third line, we have made use of Theorem 6 in \cite{chailds2021TrottertheoryPhysRevX.11.011020}, e.g., $\norm{\mathcal{M}(\tau_1)}=\bigO(\alpha_{p+1}\tau^{p+1}_1)$, where\\ 
$\alpha_{p+1}=\mathcal O \left(\sum_{\vec{j}_{p+1}} \norm{N_{\vec{j}_{p+1}}}\right)$.

Eventually, the multiplicative error $\mathcal{M}(\delta t)$ can be expressed as
\begin{equation}\label{Mterm}
   \begin{aligned}
  \mathcal{M}(\delta t)= \sum_j M_j  \delta t^{p+1}+\mathcal{M}^*_{\text{re}} +\mathcal{M}_{\text{re}} = \sum_j M_j  \delta t^{p+1}+\mathcal{M}_{Re}
    \end{aligned}
\end{equation}
and the additive error between $\mathcal{M}(\delta t)$ and $\sum_j M_j \delta t^{p+1}$ is bounded as
\begin{equation}\label{leadingM}
   \begin{aligned}
\norm{ \mathcal{M}(\delta t)-\sum_j M_j  \delta t^{p+1}} &= \norm{\mathcal{M}^*_{\text{re}} +\mathcal{M}_{\text{re}} } \le \norm{\mathcal{M}^*_{\text{re}}}+\norm{\mathcal{M}_{\text{re}} }\\
 &=\bigO(\alpha^2_{p+1} \delta t^{2p+2})+\bigO(\alpha_{p+2} \delta t^{p+2})\\
&=\mathcal O \left(\delta t^{p+2} \right).
    \end{aligned}
\end{equation}
Therefore the leading multiplicative error of $\mathcal{M}(\delta t) $ is $\mathcal{M}_{p+1}=\sum_j M_j \delta t^{p+1}$. In the following, we denote $\mathcal{M}_{Re}=\mathcal{M}^*_{\text{re}}(\delta t)+ \mathcal{M}_{\text{re}}(\delta t)$. The upper bound of $\norm{\mathcal{M}_{\text{Re}}}$ is $\mathcal O \left( \delta t^{p+2} \right)$.

\section{Operator scrambling-based bound of observable error}

\subsection{Operator scrambling-based bound of observable error for general quantum circuits }

Here we consider the algorithm or coherence error of the quantum circuit, i.e., Trotter error. Suppose $U_0$ is the ideal unitary, and $U$ is the approximate unitary. For a given initial state $\ket{\psi}$ and the (Hermitian) observable of interest $O$ ($O=O^\dagger$), the observable error is denoted as 

\begin{equation}
\epsilon_O=\abs{\langle \psi| U^\dagger_0 O U_0- U^\dagger O U \ket{\psi}}.
\end{equation}
Now let set $U^\dagger_0 O U_0- U^\dagger O U=B$, one has
$\epsilon_O=\abs{\langle \psi| B \ket{\psi}}\le \sqrt{\langle \psi| B^\dagger B \ket{\psi}}$.
Let us expand the above inequality,
\begin{equation}
\label{ech}
\begin{split}
   \epsilon_O&\le \sqrt{\bra{\psi}U^\dagger_0 O^2 U_0\ket{\psi}+\bra{\psi}U^\dagger O^2 U\ket{\psi}-\bra{\psi}U^\dagger_0 O U_0U^\dagger O U+U^\dagger O UU^\dagger_0 O U_0\ket{\psi}}.\\
\end{split}
\end{equation}

Set $U=U_0(1+\mathcal{M})$ and $O'=U_0^\dagger OU_0$, with $-(\mathcal{M}^\dagger  +\mathcal{M} )=\mathcal{M}  \mathcal{M}^\dagger  =\mathcal{M}^\dagger  \mathcal{M}  $, one has 
\begin{equation}
\begin{split} \label{sm1}
    \tilde{\epsilon}_O^2&=\bra{\psi}O'^2+(1+{\mathcal{M}^\dagger})O'^2(1+{\mathcal{M}})-O'(1+\mathcal{M}^\dagger)O'(1+\mathcal{M})-(1+\mathcal{M}^\dagger)O'(1+\mathcal{M})O'\ket{\psi}\\
    &=\bra{\psi}{\mathcal{M}^\dagger O'^2}+{O'^2\mathcal{M}}+\mathcal{M}^\dagger O'^2\mathcal{M}-O'\mathcal{M}^\dagger O'-{O'^2\mathcal{M}}-O'\mathcal{M}^\dagger O'\mathcal{M}-{\mathcal{M}^\dagger O'^2}-O'\mathcal{M}O'-\mathcal{M}^\dagger O'\mathcal{M}O'\ket{\psi}\\
    &=\bra{\psi}\mathcal{M}^\dagger O'^2\mathcal{M}-O'(\mathcal{M}+\mathcal{M}^\dagger ) O'-O'\mathcal{M}^\dagger O'\mathcal{M}-\mathcal{M}^\dagger O'\mathcal{M}O'\ket{\psi}\\
    &=\bra{\psi}\mathcal{M}^\dagger O'^2\mathcal{M}+O'\mathcal{M}^\dagger \mathcal{M}O'-O'\mathcal{M}^\dagger O'\mathcal{M}-\mathcal{M}^\dagger O'\mathcal{M}O'\ket{\psi}.
\end{split}
\end{equation}


In the above formula, the terms with the same color can be canceled.
Note that $\bra{\psi} O' \mathcal{M}  \mathcal{M}^\dagger     O' \ket{\psi}$ and $\bra{\psi}\mathcal{M}  O'^2 \mathcal{M}^\dagger  \ket{\psi}$ are positive for $\mathcal{M} \ne 0$. 
Now let rewrite the $\tilde{\epsilon}^2_O$ as 
\begin{equation}
 \tilde{\epsilon}^2_O  =   \bra{\psi} |[O',\mathcal{M} ]|^2  \ket{\psi} 
=\bra{\psi} [O',\mathcal{M}  ]^\dagger [O',\mathcal{M}  ]  \ket{\psi}.\label{eq:one-step_scrambling_bound}
\end{equation}
This formula has a similar form to quantum operator scrambling. Thus, we have the following Lemma for observable error:

\begin{lemma} [Operator scrambling-based bound] \label{sm:scrambling}
Given an ideal quantum circuit $U_0$ and approximate unitary $U$, the error of the observable $O$ is bounded by   
\begin{equation}
 \epsilon_O=\abs{\langle \psi| U^\dagger_0 O U_0- U^\dagger O U \ket{\psi}} \le \sqrt{\bra{\psi} |[O_{U_0},\mathcal{M} ]|^2  \ket{\psi} }=\norm{[O_{U_0},\mathcal{M} ]\ket{\psi}},
\end{equation}
where $U=U_0(I+\mathcal{M})$, and  $O_{U_0}=U^\dagger_0 O U_0$ denote the ideal evolution of observable $O$.
\end{lemma}
Suppose $\mathcal{M}=\sum_j \mathcal M_j$, where $\mathcal M_j$ is the local term ($\mathcal M_j/\norm{\mathcal M_j}$ is a Pauli operator), one can further bound the observable error as the summation of the operator scrambling of the local error term $\mathcal M_j$ by using the triangle inequality.
Note that, here $\mathcal M_j$ is not one of the leading-order error terms $M_j$ defined in the quantum simulation setting.

\begin{corollary}
[Operator scrambling-based bound (local case)] \label{scrambling_l}
Given an ideal quantum circuit $U_0$ and approximate unitary $U$, the error of the observable $O$ is bounded by 
 \begin{equation}
    \epsilon_O\le \norm{[O_{U_0},\mathcal{M} ]\ket{\psi}} \le \sum_j \norm{[O_{U_0},\mathcal M_j ]\ket{\psi}}=\sum_j \sqrt{\bra{\psi} |[O_{U_0},\mathcal M_j ]|^2  \ket{\psi} },
\end{equation}   
where $U=U_0(I+\mathcal{M})$, $\mathcal{M}=\sum_j \mathcal M_j$ and  $O_{U_0}=U^\dagger_0 O U_0$ denote the ideal evolution of observable $O$.
\end{corollary}
It should be noted that similar results hold if $O=\sum_i O_i$, where $O_i=\alpha_i P_i$ ($P_i$ is the Pauli operator).

\subsection{Operator scrambling-based bound of observable error for a single segment}

In this section, we primarily focus on short-time evolution $e^{-iH \delta t}$ and establish a connection between the Trotter error in a single segment and operator scrambling.
Considering the target quantum evolution $U_0 =e^{-iH\delta t}$,
a $p$th-order Trotter approximation $U_{p}$, and a pure quantum state $\ket{\psi}$, according to \cref{scramblingPF}, the Trotter error bound can be expressed as follows:
   \begin{equation}
           \epsilon_O=\abs{\langle \psi| e^{iH\delta t}Oe^{-iH\delta t}- U_p^\dagger O U_p \ket{\psi}}\le  \norm{[O(\delta t),\sum_j M_j ] \ket{\psi}}\delta t^{p+1} + 2 \norm{O} \norm{\mathcal{M}_{Re}}, 
        \end{equation}
Here $O (\delta t)=e^{iH\delta t}Oe^{-iH\delta t}$, $M_j$ are the leading-order local error terms and $\mathcal{M}_{Re}$ is the sum of the higher-order remaining terms, with $\mathcal{M}_{Re}=\mathcal O(\delta t^{p+2})$.

\begin{theorem}[Operator scrambling-based bound for observable error]\label{scramblingPF}
      Given ideal unitary $U_0$ and approximate unitary $U_p$ satisfying $U_p=U_0(I+\mathcal{M})$,  for the Hamiltonian $H=\sum_{l=1}^L H_l$ and an input pure state $\ket{\psi}$, the additive Trotter error of the $p$th-order product formula of an observable $O$ can be bound as
        \begin{equation}
           \epsilon_O=\abs{\langle \psi| e^{iH\delta t}Oe^{-iH\delta t}- U_p^\dagger O U_p \ket{\psi}}\le  \norm{[O(\delta t),\sum_j M_j] \ket{\psi}}\delta t^{p+1} + 2 \norm{O} \norm{\mathcal{M}_{Re}}, 
        \end{equation}
        where $O (\delta t)=e^{iH\delta t}Oe^{-iH\delta t}$, $\norm{\mathcal{M}_{Re}}=\bigO(\alpha_{p+2} \delta t^{p+2})$, $\alpha_{p+2}= \sum_{l_1,...,l_{p+2}=1}^L \norm{[H_{l_1},[H_{l_2},...,[H_{l_{p+1}},H_{l_{p+2}}]]]} $ and $\mathcal{M}=\sum_j M_j \delta t^{p+1}+\mathcal{M}_{Re}$.
\end{theorem}

\begin{proof}
       According to our above analysis (see \autoref{Mterm}), the multiplicative error of $p$th-order product formula is  
       $$\mathcal{M} = \mathcal{M}_{p+1} +\mathcal{M}_{Re}=\sum_j M_j \delta t^{p+1} +\mathcal{M}_{Re},$$
       where $\mathcal{M}_{p+1}=\sum_j M_j \delta t^{p+1}$, $\norm{\mathcal{M}_{Re}}=\bigO(\alpha_{p+2} \delta t^{p+2})$ and $\alpha_{p+2}= \sum_{l_1,...,l_{p+2}=1}^L \norm{[H_{l_1},[H_{l_2},...,[H_{l_{p+1}},H_{l_{p+2}}]]]} $.

 From  Lemma \ref{sm:scrambling},
        $\epsilon^2_O\le \bra{\psi} [O(\delta t),\mathcal{M}  ]^\dagger [O(\delta t),\mathcal{M}  ]  \ket{\psi}$ where $O(\delta t)=e^{iH\delta t} O e^{-iH \delta t}$. Since $\norm{A\ket{\psi}}=\sqrt{\bra{\psi}A^\dagger A \ket{\psi}}$, and expanding $\mathcal{M}  $, we have
        \begin{equation}
        \begin{split}
                    \epsilon^2_O&\le \bra{\psi} [O(\delta t),\mathcal{M}_{p+1}]^\dagger [O(\delta t),\mathcal{M}_{p+1}]  \ket{\psi}+\bra{\psi} [O(\delta t),\mathcal{M}_{Re}]^\dagger [O(\delta t),\mathcal{M}_{Re}]  \ket{\psi}\\
                    &\quad + \bra{\psi} [O(\delta t),\mathcal{M}_{p+1}]^\dagger [O(\delta t),\mathcal{M}_{Re}]  \ket{\psi} +\bra{\psi} [O(\delta t),\mathcal{M}_{Re}]^\dagger [O(\delta t),\mathcal{M}_{p+1}]  \ket{\psi}\\
                    &\le  \norm{[O(\delta t ),\mathcal{M}_{p+1}]\ket{\psi}}^2+\norm{[O(\delta t),\mathcal{M}_{Re}]\ket{\psi}}^2 + 2 \norm{[O(\delta t),\mathcal{M}_{p+1}]\ket{\psi}}\norm{[O(\delta t),\mathcal{M}_{Re}]\ket{\psi}}\\
                    &= (\norm{[O(\delta t),\mathcal{M}_{p+1}]\ket{\psi}}+\norm{[O(\delta t),\mathcal{M}_{Re}]\ket{\psi}})^2.
        \end{split}
        \end{equation}

Therefore, the upper bound of observable error for $p$th-order product formula is 
\begin{equation}
\begin{split}
        \epsilon_O &\le \norm{[O (\delta t),\mathcal{M}_{p+1}]\ket{\psi}}+\norm{[O(\delta t),\mathcal{M}_{Re}]\ket{\psi}} \\
        &= \norm{[O(\delta t),\sum_j M_j]\ket{\psi}} \delta t^{p+1} + \norm{[O(\delta t),\mathcal{M}_{Re}]\ket{\psi}}\\
         &\le \norm{[O(\delta t),\sum_j M_j]\ket{\psi}} \delta t^{p+1} + 2 \norm{O} \norm{\mathcal{M}_{Re}}\\
        &=\norm{[O(\delta t),\sum_j M_j]\ket{\psi}} \delta t^{p+1} + \bigO( \norm{O} \delta t^{p+2} ).
\end{split}
\end{equation}
Traditionally, one may set $\norm{O}=1$. 
By using the triangle inequality, one can also bound the observable error as 
\begin{equation}
    \epsilon_O\le \sum_j\norm{[O(\delta t), M_j]\ket{\psi}} \delta t^{p+1}+\bigO( \norm{O}\delta t^{p+2} ).
\end{equation}
That is, the upper bound of observable error is the accumulated operator scrambling for all local $M_j$.
\end{proof}

\subsection{Operator scrambling-based bound of observable error for long-time evolution}
We are interested in the ideal unitary of $U^r_0$ and the approximate unitary of $U^r_p$. This is the $p$-th product formula for quantum simulation. 
Since $\abs{\langle \psi| B \ket{\psi}}\le \sqrt{\langle \psi| B^\dagger B \ket{\psi}}$, we have 
$$
\epsilon_O=\abs{\langle \psi| {U^\dagger_0}^r O U^r_0- {U^\dagger_p}^r O {U^r_p}\ket{\psi}}\le \norm{({U^\dagger_0}^r O U^r_0- {U^\dagger_p}^r O {U^r_p} )\ket{\psi}}.
$$
A naive upper bound for observables is to use the triangle inequality to bound the $l_2$ norm \cite{zhao2025entanglement} in which one should treat ${U^\dagger_0}^r O U^r_0$  as a whole operation. However, in this case, the observable-driven error mitigation no longer holds, and the error is independent of the observable, reducing to the case presented in \cite{zhao2025entanglement}.

Instead of directly using Cauchy–Schwarz inequality at first,  one should bound the observable error of the product formula as 

\begin{equation}
\begin{split}
\epsilon_O &=\abs{\langle \psi| {U^{\dagger}_p}^r O U^r_p- {U^\dagger_0}^r O {U^r_0}\ket{\psi}}\\
&= \big|\langle \psi| {U^{\dagger}_p}^r O U^r_p-  {U^{\dagger}_p}^{r-1} U^\dagger_0 O U_0 U^{r-1}_p    + {U^{\dagger}_p}^{r-1} U^\dagger_0 O U_0 U^{r-1}_p - {U^{\dagger}_p}^{r-2} {U^\dagger_0}^2 O U^2_0 U^{r-2}_p  + \quad \cdots \quad \\
& \quad \cdots 
+{U^{\dagger}_p} {U^\dagger_0}^{r-1} O U^{r-1}_0 U_p   -{U^\dagger_0}^r O {U^r_0}\ket{\psi}\big| \\
&\le \abs{\langle \psi| {U^{\dagger}_p}^r O U^r_p-  {U^{\dagger}_p}^{r-1} U^\dagger_0 O U_0 U^{r-1}_p \ket{\psi}} + \abs{\langle \psi| {U^{\dagger}_p}^{r-1} U^\dagger_0 O U_0 U^{r-1}_p-  {U^{\dagger}_p}^{r-2} {U^\dagger_0}^2 O U^2_0 U^{r-2}_p \ket{\psi}} + \cdots \\
& \quad \cdots +\abs{\bra{\psi}{U^{\dagger}_p} {U^\dagger_0}^{r-1} O U^{r-1}_0 U_p   -{U^\dagger_0}^r O {U^r_0} \ket{\psi}}\\
& = \sum_{k=0}^{r-1} \abs{\langle \psi| {U^{\dagger}_p}^{r-k} {U^\dagger_0}^{k} O U^{k}_0 U^{r-k}_p-  {U^{\dagger}_p}^{r-k-1} {U^\dagger_0}^{k+1} O U^{k+1}_0 U^{r-k-1}_p \ket{\psi}}\\
&= \sum_{k=0}^{r-1} \abs{\langle \psi| {U^{\dagger}_p}^{r-k-1}  ( {U^{\dagger}_p}{U^\dagger_0}^{k} O U^{k}_0 U_p -   {U^\dagger_0}^{k+1} O U^{k+1}_0)  U^{r-k-1}_p \ket{\psi}}\\
&= \sum_{k=0}^{r-1} \abs{\langle \psi_{k+1}|  ({U^{\dagger}_p} {O}_k U_p -   {U^\dagger_0}{O}_k {U_0} )  \ket{\psi_{k+1}}},
\label{longtime}
\end{split}
\end{equation} 
where $\ket{\psi_{k+1}}=U^{r-(k+1)}_p \ket{\psi}$ and ${O}_k={U^\dagger_0}^{k} O U^{k}_0$. 

Suppose $U_p=U_0(I+\mathcal{M}  )$ (we still denote the additive error term is $\mathcal{M} $).
Based on our previous analysis, we bound observable errors as 

\begin{equation}
\begin{split}
    \epsilon_O 
    &\le   \sum_{k=0}^{r-1}  \norm{({U^{\dagger}_p} {O}_k U_p -   {U^\dagger_0}{O}_k {U_0} )  \ket{\psi_{k+1}}}\\
    &=  \sum_{k=0}^{r-1} \norm{[O_{k+1},\mathcal{M}]\ket{\psi_{k+1}}} \\
    &=\sum_{k=0}^{r-1} \sqrt{ \bra{\psi_{k+1} }[O_{k+1},\mathcal{M}]^\dagger  [O_{k+1},\mathcal{M}] \ket{\psi_{k+1}}}\\
    &= \sum_{k=0}^{r-1} \sqrt{C^{(k+1)}_{O,{M}}},
    \end{split}
    \label{eq:long-term_scrambling_bound}
\end{equation}
where $C^{(k+1)}_{O,{M}}=\bra{\psi_{k+1} }[O_{k+1},\mathcal{M} ]^\dagger  [O_{k+1},\mathcal{M} ] \ket{\psi_{k+1}}$ is the quantity of quantum scrambling for evolution of time step $k+1$ with initial state $\ket{\psi_{k+1}}=U^{r-(k+1)}_p \ket{\psi}$.
The error of the observable changes with the evolution, where the state and the observable evolve. The only difference is that the evolution of the state is carried out according to the approximate $p$-order PF $U^{r-(k+1)}_p$, while ${O}_{k+1}={U^\dagger_0}^{k+1} O U^{k+1}_0$ is carried out according to the ideal one. 
This suggests that operator scrambling and state growth simultaneously determine the circuit error.

Since $\mathcal{M}=\sum_j M_j \delta t^{p+1} + \mathcal{M}_{Re}$ in which each $M_j$ is a local operator, i.e., $n$-qubit Pauli operator. According to Lemma \ref{scramblingPF}, the long-time observable error is given as 

\begin{theorem}[Accumulated scrambling-based bound (Theorem 1 in the main text)]
      Given ideal unitary $U^r_0=e^{-iH  t }$, where $t=r\delta t$ and approximate unitary $U^r_p$ with $U_p=U_0(I+\mathcal{M})$,  for the Hamiltonian $H=\sum_{l=1}^L H_l$ and an input pure state $\ket{\psi}$, the additive Trotter error of an observable $O$ can be bound as
        \begin{equation}
           \epsilon_O=\abs{\langle \psi| e^{iH  t } O e^{-iH t}- {U_p^\dagger}^r O U^r_p \ket{\psi}}\le  \sum_{k=1}^{r}\sqrt{C^{(k)}_{O,{M}}}\delta t^{p+1} + 2 r\norm{O} \norm{\mathcal{M}_{Re}}, 
        \end{equation}
        where 
        $$C^{(k)}_{O,{M}}=\bra{\psi_{k}} |[O(k \delta t),M]|^2 \ket{\psi_k},$$ is the  quantum scrambling for evolution of time step $k$ with initial state $\ket{\psi_{k}}=U^{r-k}_p \ket{\psi}$,
        $O(k\delta t)=e^{iH k\delta t} O e^{-iH k \delta t}$, 
       and $\norm{\mathcal{M}_{Re}}=\bigO(\sum_{l_1,...,l_{p+2}=1}^L \norm{[H_{l_1},[H_{l_2},...,[H_{l_{p+1}},H_{l_{p+2}}]]]} \delta t^{p+2})$. Here $\mathcal{M}=M \delta t^{p+1}+ \mathcal{M}_{Re}=\sum_j M_j \delta t^{p+1}+\mathcal{M}_{Re}$.   
\end{theorem}

\subsection{Observable error of scrambling-based bound with random inputs}\label{Sec:Average}

Our previous analysis of the scrambling-based bound focused on a fixed input state. In this subsection, we analyze the average performance of observable error for input states chosen randomly from a given ensemble.

For an ensemble $\mc{E}$ of quantum states, the average error between $U$ and $U_0$ is 
\begin{equation}
\begin{aligned}
 R^{\mc{E}}(U_0,U;O)&:= \mathbb{E}_{\mathcal{E}}\big [|\langle \psi| U^\dagger_0 O U_0- U^\dagger O U \ket{\psi}| \big].
\end{aligned}
\end{equation}    
Here $\mathcal{E}$ can be a discrete or continuous ensemble. 
While it is simple to consider Haar-random inputs, this distribution is challenging to realize in practice. Fortunately, the analysis of average case holds under weaker assumptions on the distribution of inputs---in particular, it suffices for the input to be a 1-design. In general, a complex projective $t$-design is a distribution that agrees with the Haar ensemble for any homogeneous degree-$t$ polynomial of the state and its conjugate. Formally, a probability distribution over quantum states $\mathcal{E}=\{(p_{i}, \psi_{i})\}$ is a complex projective $t$-design if
\begin{equation}
\begin{aligned}\label{Eq:tdesign}
\sum_{\psi_i \in \mathcal{E}} p_i\ket{\psi_i}\bra{\psi_i}^{\otimes t}=\int_{\mathrm{Haar}} \ket{\psi}\bra{\psi}^{\otimes t} \d \psi,
\end{aligned}
\end{equation}
where the sum on the left is over the (discrete) ensemble $\mathcal{E}$, and the integral on the right is over the Haar measure. It is clear that if $\mathcal{E}$ is a $t$-design then it is also a $(t-1)$-design \cite{low2010pseudorandomnesslearningquantumcomputation}.

Let $\mc{H}_d$ denote a $d$-dimensional Hilbert space.
The integral over Haar-random states of the projector onto $t$ copies of the state is proportional to the projector $\Pi_+$ onto the symmetric subspace of $\mc{H}_d^{\otimes t}$ \cite{harrow2013churchsymmetricsubspace}:
\begin{equation}\label{Eq:tdesignsym}
\begin{aligned}
\int_{\mathrm{Haar}} \ket{\psi}\bra{\psi}^{\otimes t} \d \psi=\frac{\Pi_+}{D_+}&=\frac{\sum_{\pi\in S_t}W_{\pi}}{t!D_+}.
\end{aligned}
\end{equation}
Here $D_+:=\binom{d+t-1}{t}$ is the dimension of the symmetric subspace, $S_t$ is the symmetric group of order $t$, and $W_{\pi}$ is the unitary representation of $\pi\in S_t$ that permutes the states of each copy in $\mc{H}_d^{\otimes t}$ according to $\pi$.
For $t=1$, the only element in the symmetric group is the identity, so
\begin{equation}\label{Eq:onedesignsym}
\begin{aligned}
\int_{\mathrm{Haar}} \ket{\psi}\bra{\psi}\d \psi=\id/d.
\end{aligned}
\end{equation}
For $t=2$, $S_2$ has two elements, the permutations $(1)(2)$ and $(1,2)$. We have $W_{(1)(2)}=\id^{\otimes 2}$ and $W_{(1,2)}=\s$ (the swap operator on the 2-copy space, with $\s \ket{\psi_1}\ket{\psi_2}= \ket{\psi_2}\ket{\psi_1}$). Therefore
\begin{equation}\label{Eq:twodesignsym}
\begin{aligned}
\int_{\mathrm{Haar}} \ket{\psi}\bra{\psi}^{\otimes 2}\d \psi=\frac{\id^{\otimes 2}+\s}{d(d+1)}.
\end{aligned}
\end{equation}

Here, we introduce our average error as follows.
\begin{lemma}\label{Lemma:S_averagecase}
For two $d$-dimensional unitaries $U$ and $U_0$ with $U=U_0(\id+\mathcal{M})$, the expectation of $\epsilon_O$ with respect to a 1-design ensemble $\mathcal{E}$ is
\begin{equation}\label{eq:avSresult}
  R^{\mathcal{E}}(U_0,U;O)\le \sqrt{\frac{\Tr( |[O_{U_0},\mathcal{M}  ]|^2)}{d}  } ={\norm{[O_{U_0},\mathcal{M}]}_F},
\end{equation}
where $O_{U_0}=U^\dagger_0 O U_0$.

\end{lemma}

\begin{proof}
Now let us relate the average error for a Haar-random state to the average over a $t$-design ensemble for observable error. The average error of observables is 
\begin{equation}\label{eq:l2upp}
\begin{aligned}
R^{\mathcal{E}}(U_0,U;O)&=\mathbb{E}_{\psi \in \mathcal{E}}\big [|\langle \psi| U^\dagger_0 O U_0- U^\dagger O U \ket{\psi}| \big]\\
&\le\mathbb{E}_{\psi \in \mathcal{E}} \sqrt{\bra{\psi}  |[O_{U_0},\mathcal{M}  ]|^2  \ket{\psi} } \\
&\le  \sqrt{\mathbb{E}_{\psi \in \mathcal{E}}\bra{\psi}  |[O_{U_0},\mathcal{M}  ]|^2  \ket{\psi} } \\
&= \sqrt{\frac{\Tr( |[O_{U_0},\mathcal{M} ]|^2)}{d}  } ={\norm{[O_{U_0},\mathcal{M}]}_F}
\end{aligned}
\end{equation}
where the first bound is from Lemma \ref{sm:scrambling} and the second bound is due to the Cauchy–Schwarz inequality.
\end{proof}



For the long-time evolution,
  $ \text{average}(\epsilon_O)\le \sum_{k=1}^{r} \sqrt{\frac{\Tr(|[O_k,\mathcal{M}]|^2)  }{d}}=\sum_{k=1}^{r}{\norm{[O_k,\mathcal{M}]}_F}$,
where $O_k={U^\dagger_0}^k O U^k_0$.

{
Note that compared to the exact worst-case commutator norm
$\norm{[O_k,\mathcal{M}]}$, the normalized Frobenius norm
$\norm{[O_k,\mathcal{M}]}_F$ exhibits a $\sqrt{N}$-type
reduction. Specifically, expanding
$[O_k,\mathcal{M}]=\sum_{j=1}^{J} \alpha_j P_j$ in
distinct Pauli strings with $\|P_j\|=1$ and
$\Tr(P_i^\dagger P_j)/d=\delta_{ij}$, one has
\begin{equation}
\norm{[O_k,\mathcal{M}]}
    \le
    \sum_{j=1}^{J} |\alpha_j|
    = O(J),
\end{equation}
whereas
\begin{equation}
\norm{[O_k,\mathcal{M}]}_F
    =
    \sqrt{\frac{1}{d}\Tr\!\left(
    \sum_{i,j} \alpha_i^*\alpha_j P_i^\dagger P_j
    \right)}
    =
\left(\sum_{j=1}^{J}|\alpha_j|^2\right)^{1/2}
    = O(\sqrt{J}),
\end{equation}
where the last equalities in both lines assume comparable
coefficients. Since the number of contributing local terms
typically scales as $J=O(N)$, this gives a $\sqrt{N}$
separation between the two norms. Therefore, the
scrambling-based bound---which in the random-state or
sufficiently entangled regime is governed by the Frobenius
scale---improves over the exact worst-case commutator norm
by a factor of $\sqrt{N}$ in system-size scaling.}
 

\section{Entanglement-based bound of observable error}

\subsection{Vector norm-based bound of observable error}
The operator scrambling bound suggests
$\epsilon_O\le\sqrt{\bra{\psi} [O',\mathcal{M} ]^\dagger [O',\mathcal{M} ] \ket{\psi} }$ where $O'=U^\dagger_0 O U_0$. However, in many practical scenarios, this upper bound becomes infeasible to compute for large system sizes. To efficiently estimate the error bound, it may be necessary to adopt a more relaxed bound.

Here we expand $\tilde{\epsilon}_O :=\sqrt{\bra{\psi} [O',\mathcal{M} ]^\dagger [O',\mathcal{M} ] \ket{\psi} }$ and bound $\epsilon_O$ as
\begin{equation}
\epsilon_O\le   \tilde{\epsilon}_O \le \sqrt{ \bra{\psi} O' \mathcal{M}^\dagger   \mathcal{M}    O' \ket{\psi}+\bra{\psi}\mathcal{M}^\dagger  {O'}^2 \mathcal{M}  \ket{\psi} +|\bra{\psi} O' \mathcal{M}^\dagger  O' \mathcal{M}  \ket{\psi}|+ |\bra{\psi} \mathcal{M}^\dagger  O'\mathcal{M}  O' \ket{\psi}|}.
\label{upperbound}
\end{equation}
It can be seen that $O'$ and $\mathcal{M}$ have equal footing as it is symmetric about $O'$ and $\mathcal{M}$.
  Note that the vector norm is defined as $\norm{A\ket{\psi}}=\sqrt{\bra{\psi}A^\dagger A \ket{\psi}}$.

For the first term in $\tilde{\epsilon}_O$, 
without loss of generality, one can rewrite $\bra{\psi} O' \mathcal{M}^\dagger   \mathcal{M}    O' \ket{\psi}$ as
\begin{equation}
    \bra{\psi} O' \mathcal{M}^\dagger  \mathcal{M}     O' \ket{\psi} = \bra{\psi}{{O^\prime}^\dagger} O'\ket{\psi} \bra{\psi_{O'}} \mathcal{M}^\dagger  \mathcal{M}   \ket{\psi_{O'}}= \norm{O'\ket{\psi}}^2 \norm{\mathcal{M}\ket{\psi_{O'}}}^2,
\end{equation}
where $\ket{\psi_{O'}}=\frac{O'\ket{\psi}}{\norm{O'\ket{\psi}}}$.
For the second term $\bra{\psi}\mathcal{M}^\dagger   {{O'}^2} \mathcal{M} \ket{\psi}$, it can be written as 
\begin{equation}
     \bra{\psi}\mathcal{M}^\dagger   {{O'}^2} \mathcal{M} \ket{\psi}=  
      {\bra{\psi }}\mathcal{M}^\dagger \mathcal{M}\ket{\psi}  \bra{\psi_{\mathcal{M}  }} {{O'}^2} \ket{\psi_{\mathcal{M} }}= \norm{\mathcal{M} \ket{\psi}}^2\norm{O'\ket{\psi_{\mathcal{M} }}}^2,
\end{equation}
where $\ket{\psi_{\mathcal{M}  }}=\frac{\mathcal{M} \ket{\psi}}{\norm{\mathcal{M} \ket{\psi}}}$.
For the third and fourth terms $\eta=|\bra{\psi} O' \mathcal{M}^\dagger   O' \mathcal{M} \ket{\psi}|+|\bra{\psi} \mathcal{M}^\dagger  O'\mathcal{M}  O' \ket{\psi}|$, we can bound it as
\begin{equation}
    \begin{split}
    \eta&=|\bra{\psi} O' \mathcal{M}^\dagger   O' \mathcal{M} \ket{\psi}|+|\bra{\psi} \mathcal{M}^\dagger  O'\mathcal{M}  O' \ket{\psi}|\\
    &\le 2\sqrt{\bra{\psi} O' \mathcal{M}^\dagger \mathcal{M} O'\ket{\psi} \bra{\psi}\mathcal{M}^\dagger   {{O'}^2} \mathcal{M} \ket{\psi}}\\
    &=2\norm{\mathcal{M} O'\ket{\psi}}\norm{O'\mathcal{M}  \ket{\psi}}.
    \end{split}
\end{equation}

Putting all error terms together,  we have the upper error bound of the observable $O$  
\begin{equation}
\begin{split}
       \epsilon_O &\le \sqrt{ \norm{\mathcal{M} O'\ket{\psi}}^2+ \norm{O'\mathcal{M}  \ket{\psi}}^2+  2\norm{\mathcal{M} O'\ket{\psi}}\norm{O'\mathcal{M}  \ket{\psi}}}\\
       &= \norm{\mathcal{M} O'\ket{\psi}}+ \norm{O'\mathcal{M}  \ket{\psi}}\\
       &=\norm{O'\ket{\psi}} \norm{\mathcal{M}\ket{\psi_{O'}}}+\norm{O'\ket{\psi_{\mathcal{M} }}}\norm{\mathcal{M} \ket{\psi}}.
\end{split}
\end{equation}

Note that $\norm{O'\ket{\psi}}=\sqrt{\bra{\psi}U^\dagger_0 O U_0 U^\dagger_0 O U_0 \ket{\psi}}=\sqrt{(\bra{\psi}U^\dagger_0)O^2 (U_0\ket{\psi})}=\norm{O \ket{\psi_{U_0}}}$, where $\ket{\psi_{U_0}}=U_0\ket{\psi}$.




\begin{lemma}
    Given observable $O$, input state $\ket{\psi}$, ideal unitary $U_0$ and approximate unitary $U=U_0(I+\mathcal{M} )$, the observable error is bounded as
\begin{equation}
\begin{split}
        \epsilon_O &\le 
        \norm{O \ket{\psi_{U_0}}} \norm{\mathcal{M}\ket{\psi_{O'}}}+\norm{\mathcal{M} \ket{\psi}}\norm{O \ket{\psi_{U_0 \mathcal{M} }}},
\end{split}
\end{equation}
where $O'=U^\dagger_0 OU_0$, $\ket{\psi_{O'}}=\frac{O'\ket{\psi}}{\norm{O'\ket{\psi}}}$,  $\ket{\psi_{U_0 \mathcal{M} }}=U_0 \ket{\psi_{\mathcal{M} }}$, and $\ket{\psi_{\mathcal{M} }}=\frac{\mathcal{M}  \ket{\psi}}{\norm{\mathcal{M}  \ket{\psi}}}$.
\end{lemma}

\begin{lemma} [Vector norm-based bound of short-time of observable error]\label{Vectornormbound}
       For the Hamiltonian $H=\sum_{l=1}^L H_l$ and an input pure state $\ket{\psi}$, the additive Trotter error of the $p$th-order product formula of an observable $O$ can be bound as
        \begin{equation}
        \begin{split}
           \epsilon_O\le  
          (\norm{O \ket{\psi (\delta t)}} \norm{M\ket{\psi_{O(\delta t)}}}+\norm{M\ket{\psi}}\norm{O \ket{\psi_{U_0(\delta t)M}}}) \delta t^{p+1}+ 2 \norm{O} \norm{\mathcal{M}_{Re}}, 
                   \end{split}
        \end{equation}
        where $M\delta t^{p+1}=\sum_j M_j\delta t^{p+1}$ is the leading error of $\mathcal{M}(\delta t)$. Here  $\norm{\mathcal{M}_{Re}}=\bigO(\alpha_{p+2} \delta t^{p+2})$, $\ket{\psi(\delta t)}= U_0(\delta t) \ket{\psi} =e^{-iH\delta t} \ket{\psi}$,
$\ket{\psi_{O(\delta t)}}=\frac{O(\delta t)\ket{\psi}}{\norm{O(\delta t)\ket{\psi}}}$,  $\ket{\psi_{M}}=\frac{M\ket{\psi}}{\norm{M\ket{\psi}}}$, $\ket{\psi_{U_0(\delta t)M}}=U_0(\delta t) \ket{\psi_{M}}$, and
\[\alpha_{p+2}= \sum_{l_1,...,l_{p+2}=1}^L \norm{[H_{l_1},[H_{l_2},...,[H_{l_{p+1}},H_{l_{p+2}}]]]}. \]
\end{lemma}

\subsection{Entanglement-based bound of observable error}
We would like to investigate how the entangled state can influence the observable error. 
Here we introduce the following mathematical Lemma. 

\begin{lemma}(Distance-based bound \cite{zhao2025entanglement}) \label{SMLemma:local_d}
Let $A=\sum_j A_j$ act on $N$ qubits, where $A_j$ acts nontrivially on the subsystem with $\support(A_j)$. Then
\begin{align}
  | \bra{\psi}A^{\dagger}A \ket{\psi}|\le \|A\|_F^2+    \sum_{j,j'} \|A_j^\dagger A_{j'}\| ~\tr|\rho_{j,j'}- \mathbb{I}_{\support(A_j^{\dag}A_{j'})}/d_{\support(A_j^{\dag}A_{j'})} |,
\end{align}
where $\|A\|^2_F:= \tr(A^{\dagger}A) /d$ is the (square) of the normalized Frobenius norm, and $\rho_{j,j'}:=\tr_{[N]\setminus \support(A_j A_{j'})}(\ket{\psi}\bra{\psi})$ is the reduced density matrix of $\ket{\psi}\bra{\psi}$ on the subsystem of $\support(A^\dagger_jA_{j'})$.
\end{lemma}

\begin{proof}
The term $A_j^{\dag}A_{j'}$ in the expression for $A^{\dag}A$ only acts nontrivially on $\support(A_j^{\dag}A_{j'})$. We denote its nontrivial part by $L_{j,j'}:=\frac{1}{d_{N-\support(A^\dagger_j A_{j'})}}\tr_{[N]\setminus \support(A^\dagger_j A_{j'})} (A_j^{\dag}A_{j'})$. Since $2^{-N}\tr(A_j^{\dagger} A_{j'})=d_{\support(A_j^{\dag}A_{j'})}^{-1}\tr(L_{j,j'})$, we have
\begin{equation}
\begin{aligned}
\bra{\psi}A_j^{\dagger}A_{j'} \ket{\psi}&=\tr(L_{j,j'}\rho_{j,j'})=\tr[L_{j,j'}(\rho_{j,j'}- \mathbb{I}_{\support(A_j^{\dag}A_{j'})}/d_{\support(A_j^{\dag}A_{j'})})]+\tr(L_{j,j'})/d_{\support(A_j^{\dag}A_{j'})}\\
&= \tr[L_{j,j'}(\rho_{j,j'}- \mathbb{I}_{\support(A_j^{\dag}A_{j'})}/d_{\support(A_j^{\dag}A_{j'})})]+\tr(A_j^{\dagger}A_{j'})/2^N\\
&\leq \|L_{j,j'}\|\tr|\rho_{j,j'}- \mathbb{I}_{\support(A_j^{\dag}A_{j'})}/d_{\support(A_j^{\dag}A_{j'})}|+\tr(A_j^{\dagger}A_{j'})/2^N\\
&= \|A_j^{\dag}A_{j'}\| \tr|\rho_{j,j'}- \mathbb{I}_{\support(A_j^{\dag}A_{j'})}/d_{\support(A_j^{\dag}A_{j'})}|+\tr(A_j^{\dagger}A_{j'})/2^N,
\end{aligned}
\end{equation}
and the result follows by summing the indices $j,j'$.
\end{proof}

Moreover, the trace distance of $\rho_{j,j'}$ and $\mathbb{I}_{\support(A_j^{\dag}A_{j'})}/d_{\support(A_j^{\dag}A_{j'})}$ can be bounded by the relative entropy, the quantum Pinsker inequality, as
\begin{align}
    \tr|\rho_{j,j'}- \mathbb{I}_{\support(A_j^{\dag}A_{j'})}/d_{\support(A_j^{\dag}A_{j'})}| \le \sqrt{2S(\rho_{j,j'}\|\mathbb{I}_{\support(A_j^{\dag}A_{j'})}/d_{\support(A_j^{\dag}A_{j'})}) }=
    \sqrt{2\log(d_{\support(A_j^{\dag}A_{j'})})-2S(\rho_{j,j'})},
\end{align}
which leads to the following entanglement-based bound,

\begin{lemma} (Entanglement-based bound \cite{zhao2025entanglement})\label{SMLemma:local}
Let $A=\sum_j A_j$ act on $N$ qubits, where $A_j$ acts nontrivially on the subsystem with $\support(A_j)$. Then
\begin{align}
  | \bra{\psi}A^{\dagger}A \ket{\psi}|=\norm{A\ket{\psi}}^2\le \|A\|_F^2+    \sum_{j,j'} \|A_j^\dagger A_{j'}\| ~\sqrt{2\log(d_{\support(A_j^{\dag}A_{j'})})-2S(\rho_{j,j'})},,
\end{align}
where $\|A\|^2_F:= \tr(A^{\dagger}A) /d$ is the (square) of the normalized Frobenius norm,  $\rho_{j,j'}:=\tr_{[N]\setminus \support(A^\dagger_j A_{j'})}(\ket{\psi}\bra{\psi})$ is the reduced density matrix of $\ket{\psi}\bra{\psi}$ on the subsystem of $\support(A_jA_{j'})$, and $\text{S}(\rho_{j,j'})$ is the entanglement entropy of $\rho_{j,j'}$.
\end{lemma}

Combined with the Lemma \ref{Vectornormbound} and the Lemma \ref{SMLemma:local}, we have the following theorem:

\begin{theorem}
    [Entanglement-based bound for short-time evolution of observable error (Theorem 2 in the main text)]\label{SM:entanglement-short}
    For the Hamiltonian $H=\sum_{l=1}^L H_l$ and an input pure state $\ket{\psi}$, 
    $U_p (\delta t)=U_0(\delta t)(I+M \delta t^{p+1}+\bigO(\delta t^{p+2}))$ with ($M=\sum_j M_j $), the additive Trotter error of the $p$th-order product formula of an observable $O=\sum_j O_j$ can be bound as 
    \begin{equation}
    \begin{split}
         \epsilon_O \le  &\Bigg[\sqrt{\norm{O}^2_F+\Delta_{O}(\ket{\psi(\delta t)})} \sqrt{\norm{M}^2_F+\Delta_{M}( \ket{\psi_{O(\delta t)}})}+\sqrt{\norm{O}^2_F+\Delta_{O}(\ket{\psi_{U_0(\delta t)M}})} \sqrt{\norm{M}^2_F+\Delta_{M}(\ket{\psi})}\bigg] \delta t^{p+1}\\ +&\bigO (\delta t^{p+2}),
    \end{split}
    \end{equation}
    where  
    $\|A\|^2_F:= \tr(A^{\dagger}A) /d$ is the (square) of the normalized Frobenius norm,
    $$\Delta_A(\ket{\chi})= \sum_{j,j'} \|A_j^\dagger A_{j'}\| ~\sqrt{2\log(d_{\support(A_j^{\dag}A_{j'})})-2S(\rho_{j,j'})},$$ 
    and   $\rho_{j,j'}:=\tr_{[N]\setminus \support(A^\dagger _j A_{j'})}(\ket{\chi}\bra{\chi})$ is the reduced density matrix of $\ket{\chi}\bra{\chi}$ on the subsystem of $\support(A^\dagger_jA_{j'})$, and $\text{S}(\rho_{j,j'})$ is the entanglement entropy of $\rho_{j,j'}$. Here  $M\delta t^{p+1}=\sum_j M_j\delta t^{p+1}$ is the leading error of $\mathcal{M}(\delta t)$, $\norm{\mathcal{M}_{Re}}=\bigO(\alpha_{p+2} \delta t^{p+2})$, $\ket{\psi(\delta t)}= U_0(\delta t) \ket{\psi} =e^{-iH\delta t} \ket{\psi}$,
        $\ket{\psi_{O(\delta t)}}=\frac{O(\delta t)\ket{\psi}}{\norm{O(\delta t)\ket{\psi}}}$,  $\ket{\psi_{M}}=\frac{M\ket{\psi}}{\norm{M\ket{\psi}}}$, $\ket{\psi_{U_0(\delta t)M}}=U_0(\delta t) \ket{\psi_{M}}.$
\end{theorem}

Since $\sqrt{x^2+y^2}\le \abs{x}+\abs{y}$, the upper bound of the observable error can be rewritten as,

    \begin{equation}
    \begin{split}
         \epsilon_O \le&  \Bigg\{\bigg[\norm{O}_F+\sqrt{\Delta_{O}(\ket{\psi(\delta t)})}\bigg] \bigg[\norm{M}_F+\sqrt{\Delta_{M}( \ket{\psi_{O(\delta t)}})}\bigg] +\bigg[\norm{O}_F+\sqrt{\Delta_{O}(\ket{\psi_{U_0(\delta t)M}})} \bigg]\bigg[\norm{M}_F+\sqrt{\Delta_{M}(\ket{\psi})}\bigg]\bigg\} \delta t^{p+1} \\+&\bigO (\delta t^{p+2}).
    \end{split}
    \end{equation}

{Define $
  \|A\|_\chi
  :=
  \sqrt{\|A\|_{\rm F}^2+\Delta_A(\chi)}, 
$ with the simplified entropy-deficit correction 
$
  \Delta_A(\chi)
  =
  \sum_{\alpha,\beta}
  \|A_\alpha^\dagger A_\beta\|
  \sqrt{2D_{\alpha\beta}(\chi)},
$ and $
  D_{\alpha\beta}(\chi)
  =
  \log d_{\alpha\beta}-S(\rho^\chi_{\alpha\beta}),
$
the one-step observable error in 
\autoref{SM:entanglement-short} can be rewritten as
\begin{equation}
  \epsilon_O(\delta t)
  \le
  \delta t^{p+1}
  \left[
  \|O\|_{\psi(\delta t)}\,\|M\|_{\psi_{O(\delta t)}}
  +
  \|M\|_\psi\,\|O\|_{\psi_{U_0(\delta t) M}}
  \right]
  +\mathcal O(\delta t^{p+2}).
\end{equation}
This expression is consistent with the main text. When the entropy-deficit corrections for all four effective
states are negligible compared with the corresponding
Frobenius terms, the bound reduces to
\begin{equation}
\epsilon_O(\delta t)
\lesssim
2\|O\|_{\rm F}\|M\|_{\rm F}\delta t^{p+1}
+
\mathcal O(\delta t^{p+2}).
\end{equation}}


For long-time evolution analysis, utilizing Eq. \ref{eq:long-term_scrambling_bound} and the Lemma \ref{Vectornormbound}, the following result is obtained.

\begin{lemma}[Vector norm-based bound for long-time evolution of observable error (Lemma 2 in the main text)]\label{eq:bound_for_lt} For a quantum circuit $U_0^r$ and an approximate unitary $U^r_p$, 
   the error of observable $O$ with  input state $\ket{\psi}$ satisfies
\begin{equation}
\begin{split}
        \epsilon_O &\le \sum_{k=1}^{r} \norm{O_k\ket{\psi_k}} \norm{\mathcal{M}  \ket{{\psi_k}_{O_k}}}+\norm{\mathcal{M}  \ket{\psi_k}}\norm{O_k\ket{{\psi_k}_{\mathcal{M} }}}\\
        &=\sum_{k=1}^{r} \norm{O \ket{\tilde{\psi}_{r_k}}} \norm{\mathcal{M} \ket{{\psi_k}_{O_k}}} + \norm{\mathcal{M}  \ket{\psi_k}}\norm{O \ket{{\tilde{\psi_{r_k}}}_\mathcal{M}}} ,
\end{split}
\end{equation}
where $\ket{\tilde{\psi}_{r_k}}=(U^k_0U^{r-k}_p) \ket{\psi}$ and $\ket{\tilde{\psi_{r_k}}_\mathcal{M}}=U^k_0\mathcal{M}U^{r-k}_p \ket{\psi}/ \norm{U^k_0\mathcal{M}U^{r-k}_p \ket{\psi}}$ with
 $U_p=U_0(1+\mathcal{M})$,
$\ket{\psi_k}=U^{r-k}_p\ket{\psi}$,  $O_k={U_0^\dagger}^k O U^k_0$, 
$\ket{{\psi_k}_{O_k}}=\frac{O_k\ket{\psi_k}}{\norm{O_k\ket{\psi_k}}}=\frac{{U^\dagger_0}^k O\ket{\tilde{\psi_{r_k}}}}{\norm{O\ket{\tilde{\psi_{r_k}}}}}$ and  $\ket{{\psi_k}_{\mathcal{M} }}=\frac{\mathcal{M} \ket{\psi_k}}{\norm{\mathcal{M} \ket{\psi_k}}}$.
\end{lemma}

\begin{theorem}[Entanglement-based bound for long-time evolution of observable error]\label{SM:entanglement-long}
    For the Hamiltonian $H=\sum_{l=1}^L H_l$ and an input pure state $\ket{\psi}$, 
    $U_p=U_0(I+M \delta t^{p+1}+\bigO(\delta t^{p+2}))$ with ($M=\sum_j M_j $), the additive Trotter error of the $p$th-order product formula of an observable $O=\sum_j O_j$ can be bound as 
    \begin{equation}
    \begin{split}
         \epsilon_O \le & \sum_{k=0}^{r-1} \Bigg[
\sqrt{\norm{O}^2_F+\Delta_O(\ket{\tilde{\psi}_{r_k}})}
         \sqrt{\norm{M}^2_F+\Delta_{M}( \ket{{\psi_k}_{O_k}})}+
\sqrt{\norm{O}^2_F+\Delta_O(\ket{\tilde{\psi_{r_k}}_{M}})}
         \sqrt{\norm{M}^2_F+\Delta_{M}(\ket{\psi_k})}\bigg] \delta t^{p+1} \\+&\bigO (r \delta t^{p+2}),
    \end{split}
    \end{equation}
    where  
    $\|A\|^2_F:= \tr(A^{\dagger}A) /d$ is the (square) of the normalized Frobenius norm,
    $$\Delta_A(\ket{\chi})= \sum_{j,j'} \|A_j^\dagger A_{j'}\| ~\sqrt{2\log(d_{\support(A_j^{\dag}A_{j'})})-2S(\rho_{j,j'})},$$ 
    $\rho_{j,j'}:=\tr_{[N]\setminus \support(A^\dagger _j A_{j'})}(\ket{\chi}\bra{\chi})$ is the reduced density matrix of $\ket{\chi}\bra{\chi}$ on the subsystem of $\support(A^\dagger_jA_{j'})$, and $\text{S}(\rho_{j,j'})$ is the entanglement entropy of $\rho_{j,j'}$. Here $\ket{\tilde{\psi}_{r_k}}=(U^k_0U^{r-k}_p) \ket{\psi}$ and $\ket{\tilde{\psi_{r_k}}_{M}}=U^k_0{M}U^{r-k}_p \ket{\psi}/ \norm{U^k_0{M}U^{r-k}_p \ket{\psi}}$,
$\ket{\psi_k}=U^{r-k}_p\ket{\psi}$,  $O_k={U_0^\dagger}^k O U^k_0$, 
$\ket{{\psi_k}_{O_k}}=\frac{O_k\ket{\psi_k}}{\norm{O_k\ket{\psi_k}}}=\frac{{U^\dagger_0}^k O\ket{\tilde{\psi_{r_k}}}}{\norm{O\ket{\tilde{\psi_{r_k}}}}}$ and  $\ket{{\psi_k}_{{M} }}=\frac{{M} \ket{\psi_k}}{\norm{{M} \ket{\psi_k}}}$.
\end{theorem}

Remarkably, the error term associated with the vector norm of $O$ in Lemma \ref{eq:bound_for_lt} is largely independent of the specific segment and is primarily determined by the final state. This is fundamentally different from the state error described by \cite{zhao2025entanglement}. Specifically, even if the entanglement is minimal during intermediate stages, a large final-state entanglement where the corresponding entropy-deficit corrections can remain small across the propagated states may allow us to consistently leverage the advantage conferred by the Frobenius norm of $O$. In this case, for any segment $k$, we have
$\Delta_{O}(\ket{\tilde{\psi}_{r_k}})\approx\Delta_{O}(\ket{\tilde{\psi_{r_k}}_{M}})\approx0$.

\section{Concrete upper bounds of observable error for \texorpdfstring{PF$1$ and PF$2$}{PF1 and PF2}}

We give the concrete scrambling-based bounds for the PF1 and PF2 in the following.

\subsection{Concrete upper bound for PF1}

\begin{lemma}\label{SMLemma:AB}
(Operator scrambling-based bound for PF1)
For a two-term Hamiltonian $H=A+B$, consider the first-order product formula $U_1(\delta t)=e^{-iA\delta t}e^{-iB\delta t}$ with initial state $\ket{\psi}$. Let $M=-\frac{1}{2}[A,B]=\sum_j M_j$. Then the Trotter error for the observable $O$ is upper bounded as
     \begin{equation}
           \epsilon_O\le  \frac{1}{2}\norm{[O(\delta t), [A,B]] \ket{\psi}}\delta t^{2} + \mathscr{O}_{Re}, 
        \end{equation}
        where $O(\delta t)=e^{iH\delta t}Oe^{-iH\delta t}$ and
   $$\mathscr{O}_{Re}=2\norm{O}(1+\frac{1}{2}\norm{[A,B]} \delta t^2) \left( \frac{\delta t^3}{6}\|[A,[A,B]]\|+ \frac{\delta t^3}{3}\|[B,[A,B]]\| \right) + \frac{\delta t^4}{2} \norm{O}\norm{[A,B]}^2.$$   
   \end{lemma}
\begin{proof}

By Ref. \cite{chailds2021TrottertheoryPhysRevX.11.011020} and our analysis before, the multiplicative error $\mathcal{M}(\delta t)$ satisfying $U_1(\delta t)=U_0(I+\mathcal{M}(\delta t))$ is
\begin{equation}
\begin{aligned}\label{Mhigh}
\mathcal{M}(\delta t)&=\int_0^{\delta t} \d\tau_1 \,  e^{iH\tau_1} 
e^{-iA\tau_1}e^{-iB\tau_1}\left(e^{iB\tau_1}e^{iA\tau_1}iBe^{-iA\tau_1}e^{-iB\tau_1}    -iB\right)\\
&=\int_0^{\delta t} \d\tau_1 \, e^{iH\tau_1} e^{-iA\tau_1}e^{-iB\tau_1} [iA,iB]\tau_1 \\
&\quad+
\int_0^{\delta t} \d\tau_1
\int_0^{\tau_1} \d\tau_2 \,  e^{iH\tau_1} e^{-iA\tau_1}e^{-iB\tau_1} e^{iB(\tau_1-\tau_2)}[iB,[iA, iB]]e^{-iB(\tau_1-\tau_2)}   \tau_1 \\
&\quad+
\int_0^{\delta t} \d\tau_1  \int_0^{\tau_1} \d\tau_2 \, e^{iH\tau_1} e^{-iA\tau_1}e^{-iB\tau_1} e^{iB\tau_1}e^{iA(\tau_1-\tau_2)} [iA,[iA,iB]]e^{-iA(\tau_1-\tau_2)}   e^{-iB\tau_1}  \tau_2\\
&= \mathcal{M}_1+ \mathcal{M}_{\re},
\end{aligned}
\end{equation}
where $\mathcal{M}_1=\int_0^{\delta t} \d\tau_1 \, e^{iH\tau_1} e^{-iA\tau_1}e^{-iB\tau_1} [iA,iB]\tau_1$ and $\mathcal{M}_{\re}$ is the sum of the other integrals.
Here the second line is due to the equation
\begin{equation}
\begin{aligned}
 e^{iB\tau_1}e^{iA\tau_1}iBe^{-iA\tau_1}e^{-iB\tau_1}
&=  iB+e^{iB\tau_1}[iA, iB] \tau_1 e^{-iB\tau_1} + \int_0^{\tau_1} \d\tau_2  e^{iB\tau_1}e^{iA(\tau_1-\tau_2)} [iA,[iA,iB]] \tau_2 e^{-iA(\tau_1-\tau_2)}   e^{-iB\tau_1}.\\
&=iB+[iA, iB] \tau_1  +
\int_0^{\tau_1} \d\tau_2  e^{iB(\tau_1-\tau_2)}[iB,[iA, iB]]e^{-iB(\tau_1-\tau_2)}   \tau_1 \\
&\quad+\int_0^{\tau_1} \d\tau_2  e^{iB\tau_1}e^{iA(\tau_1-\tau_2)} [iA,[iA,iB]] \tau_2 e^{-iA(\tau_1-\tau_2)}   e^{-iB\tau_1}.
\end{aligned}
\end{equation}

Since $\norm{UXV}=\norm{X}$ where $U$ and $V$ is unitary, the second term $ \mathcal{M}_{\re}$ can be bounded as
\begin{align}\label{eq:remain_upper_pf1}
 \|\mathcal{M}_{\re}\|\le \frac{\delta t^3}{6}\|[A,[A,B]]\|+ \frac{\delta t^3}{3}\|[B,[A,B]]\|.
\end{align}
Similarly,  $\mathcal{M}_1$ can be bounded as $\norm{\mathcal{M}_1} \le \frac{\delta t^2}{2} \norm{ [A,B]}$, resulting
\begin{equation}
  \norm{\mathcal{M}(\delta t)}\le  \|\mathcal{M}_1\|+\|\mathcal{M}_{\re}\| \le \frac{\delta t^2}{2} \norm{ [A,B]}+\frac{\delta t^3}{6}\|[A,[A,B]]\|+ \frac{\delta t^3}{3}\|[B,[A,B]]\|.
\end{equation}
For the sufficiently small step sizes considered here, we use the perturbative monotonicity, for $0\le\tau_1\le \delta t$, $\norm{\mathcal{M}(\tau_1)}\le \norm{\mathcal{M}(\delta t)}$.
Since $U_1(\tau_1)=e^{-iA\tau_1}e^{-iB\tau_1}=e^{-iH\tau_1} (I+\mathcal{M}(\tau_1))$,
\begin{equation}
    \begin{split}
  \mathcal{M}_1&=\int_0^{\delta t} \d\tau_1 \, e^{iH\tau_1} e^{-iA\tau_1}e^{-iB\tau_1} [iA,iB]\tau_1 \\
&=  \int_0^{\delta t} \d\tau_1 \, (I+\mathcal{M}{(\tau_1)}) [iA,iB]\tau_1\\
&=-\frac{1}{2}[A,B]\delta t^2+ \int_0^{\delta t} \d\tau_1 \, \mathcal{M}{(\tau_1)} [iA,iB]\tau_1\\
&= -\frac{1}{2}[A,B]\delta t^2+\mathcal{M}^{*}_{\re},
\end{split}
\end{equation}
where $\mathcal{M}^{*}_{\re}$ can be bounded as
\begin{equation}
\begin{split}
        \norm{\mathcal{M}^{*}_{\re}}&\le  \int_0^{\delta t} \d\tau_1 \, \norm{\mathcal{M}( \tau_1)} \norm{[iA,iB]}\tau_1 
         \le  \norm{\mathcal{M}( \delta t)} \int_0^{\delta t} \d\tau_1 \,   \norm{[iA,iB]}\tau_1 \\
        &\le \frac{\delta t^4}{4} \norm{[A,B]}^2 + \frac{\delta t^5}{12}\|[A,[A,B]]\|\norm{[A,B]}+ \frac{\delta t^5}{6}\|[B,[A,B]]\|\norm{[A,B]}.
\end{split}
\end{equation}
Therefore, the multiplicative error of PF1 is
 $$\mathcal{M}(\delta t)= -\frac{1}{2}[A,B]\delta t^2+ \mathcal{M}^{*}_{\re} +\mathcal{M}_{\re}$$ and its leading error term is $-\frac{1}{2}[A,B]\delta t^2=\sum_j M_j \delta t^2$. 
Define $\mathcal{M}_{Re}=\mathcal{M}^*_{\text{re}}+\mathcal{M}_{\text{re}}$, so 
according to \autoref{scramblingPF}, for PF1, we have
    \begin{equation}
    \begin{split}
              \epsilon_O&=\abs{\langle \psi| U^\dagger_0 O U_0- U_1^\dagger O U_1 \ket{\psi}}\le  \frac{1}{2}\norm{[O(\delta t), [A,B] ] \ket{\psi}}\delta t^2 + 2 \norm{O} \norm{\mathcal{M}_{Re}}\\
              &\le  \frac{1}{2}\norm{[O(\delta t), [A,B] ] \ket{\psi}}\delta t^2 + 2 \norm{O} (\norm{\mathcal{M}_{\re}}+\norm{\mathcal{M}^*_{\re}})\\
              &\le \frac{1}{2}\norm{[O(\delta t), [A,B] ] \ket{\psi}}\delta t^2 \\
              &\quad +2\norm{O}(1+\frac{1}{2}\norm{[A,B]} \delta t^2) \left( \frac{\delta t^3}{6}\|[A,[A,B]]\|+ \frac{\delta t^3}{3}\|[B,[A,B]]\| \right) + \frac{\delta t^4}{2} \norm{O}\norm{[A,B]}^2.
    \end{split}
        \end{equation}
\end{proof}

\subsection{Concrete upper bound for PF2}

\begin{lemma}\label{SMLemma:AB_pf2}
(Operator scrambling-based bound for PF2)
For a two-term Hamiltonian $H=A+B$, consider the second-order product formula $U_2(\delta t)=e^{-iA\delta t/2}e^{-iB\delta t}e^{-iA\delta t/2}$ with initial state $\ket{\psi}$. Let $M=\frac{1}{12}[-iB,[-iB,-iA]]+\frac{1}{24}[iA,[iA,iB]]=\sum_j M_j$. Then the Trotter error for the observable $O$ is upper bounded as
     \begin{equation}
           \epsilon_O\le  \frac{1}{12}\norm{[O(\delta t), [B,[B,A]]] \ket{\psi}}\delta t^{3} +\frac{1}{24}\norm{[O(\delta t), [A,[A,B]]] \ket{\psi}}\delta t^{3} + \mathscr{O}_{Re}, 
        \end{equation}
        where $O(\delta t)=e^{iH\delta t}Oe^{-iH\delta t}$ and
  $\mathscr{O}_{Re}= 2\norm{O} (\zeta_{\re,1}+\zeta_{\re,2}+\zeta^*_{\re,(2,1)}+\zeta^*_{\re,(2,2)})
  $ with
  \begin{equation}      
  \begin{split}
     & \zeta_{\re,1}=\frac{{\delta t}^4}{32}\|\left[A,\left[B,\left[B,A\right]\right]\right]\|+  \frac{{\delta t}^4}{12} \|\left[B,\left[B,\left[B,A\right]\right]\right]\|, \\
     &\zeta_{\re,2}= \frac{{\delta t}^4}{32}\|\left[B,\left[A,\left[A,B\right]\right]\right]\|+  \frac{{\delta t}^4}{48}\|\left[A,\left[A,\left[A,B\right]\right]\right]\|,\\
     &\zeta^*_{\re,(2,1)}= \big( \frac{{\delta t}^6}{144} \| \left[B,\left[B,A\right]\right]\|+ \frac{{\delta t}^6}{288} \|\left[A,\left[A,B\right]\right]\|+\frac{{\delta t}^7}{384}\|\left[A,\left[B,\left[B,A\right]\right]\right]\|+  \frac{{\delta t}^7}{144} \|\left[B,\left[B,\left[B,A\right]\right]\right]\|\\
 &\quad + \frac{{\delta t}^7}{384}\|\left[B,\left[A,\left[A,B\right]\right]\right]\|+  \frac{{\delta t}^7}{576}\|\left[A,\left[A,\left[A,B\right]\right]\right]\| \big) \norm{\left[B,\left[B,A\right]\right] }\\
 &\zeta^*_{\re,(2,2)}=\big( \frac{{\delta t}^6}{288} \| \left[B,\left[B,A\right]\right]\|+ \frac{{\delta t}^6}{576} \|\left[A,\left[A,B\right]\right]\|+\frac{{\delta t}^7}{768}\|\left[A,\left[B,\left[B,A\right]\right]\right]\|+  \frac{{\delta t}^7}{288} \|\left[B,\left[B,\left[B,A\right]\right]\right]\|\\
 &\quad + \frac{{\delta t}^7}{768}\|\left[B,\left[A,\left[A,B\right]\right]\right]\|+  \frac{{\delta t}^7}{1152}\|\left[A,\left[A,\left[A,B\right]\right]\right]\| \big) \norm{\left[A,\left[A,B\right]\right]}.
  \end{split}
\end{equation}
   \end{lemma}

\begin{proof}
According to Appendix L of  Ref.~\cite{childs2019nearly,chailds2021TrottertheoryPhysRevX.11.011020} and 
\begin{align}
    e^{iA\tau_1}B e^{-iA\tau_1}= B + \int_0^{\tau_1}  \d\tau_2 e^{iA(\tau_1-\tau_2)} [A, B] e^{-iA(\tau_1-\tau_2)},
\end{align}
the multiplicative error $\mathcal{M}(\delta t)$ for PF2 is
\begin{equation}
\begin{aligned}\label{Eq:PF2additive}
&\mathcal{M}(\delta t)=e^{iH\delta t}(U_2(\delta t)-e^{-iH\delta t})\\
&=\int_0^{\delta t} \d\tau_1\int_0^{\tau_1}\d\tau_2\int_0^{\tau_2} \d\tau_3 \, e^{i \tau_1 H}e^{-i\tau_1A/2}
e^{-i\tau_1B}
e^{-i\tau_1A/2}
e^{i\tau_1A/2}e^{i\tau_1B}\\
&\quad\cdot \left( e^{-i\tau_3B}\left[-iB,\left[-iB,-i\frac{A}{2}\right]\right]e^{i\tau_3B}+ e^{i\tau_3A/2}\left[i\frac{A}{2},\left[i\frac{A}{2},iB\right]\right]e^{-i\tau_3A/2}\right) e^{-i\tau_1B} e^{-i\tau_1A/2}\\
&= \int_0^{\delta t} \d\tau_1\int_0^{\tau_1}\d\tau_2\int_0^{\tau_2} \d\tau_3 \, e^{i\tau_1H}e^{-i\tau_1A/2}
e^{-i\tau_1B}
e^{-i\tau_1A/2}\left(\left[-iB,\left[-iB,-i\frac{A}{2}\right]\right]+ \left[i\frac{A}{2},\left[i\frac{A}{2},iB\right]\right]+ R_1+R_2 \right)\\
&= \mathcal{M}_{2,1}+\mathcal{M}_{2,2} + \mathcal{M}_{\mathrm{re},1}+ \mathcal{M}_{\mathrm{re},2},
\end{aligned}
\end{equation}
where
\begin{equation}
\begin{aligned}
    R_1&= e^{i\tau_1A/2}e^{i\tau_1B}  e^{-i\tau_3B}\left[-iB,\left[-iB,-i\frac{A}{2}\right]\right]e^{i\tau_3B} e^{-i\tau_1B} e^{-i\tau_1A/2} - \left[-iB,\left[-iB,-i\frac{A}{2}\right]\right],\\
    R_2&=e^{i\tau_1A/2}e^{i\tau_1B} e^{i\tau_3A/2}\left[i\frac{A}{2},\left[i\frac{A}{2},iB\right]\right]e^{-i\tau_3A/2} e^{-i\tau_1B} e^{-i\tau_1A/2}- \left[i\frac{A}{2},\left[i\frac{A}{2},iB\right]\right],
\end{aligned}
\end{equation}
and
\begin{equation}
\begin{aligned}
    &\mathcal{M}_{2,1}= \int_0^{\delta t} \d\tau_1\int_0^{\tau_1}\d\tau_2\int_0^{\tau_2} \d\tau_3 \, e^{i\tau_1 H}e^{-i\tau_1A/2}
e^{-i\tau_1B}
e^{-i\tau_1A/2}\left[-iB,\left[-iB,-i\frac{A}{2}\right]\right],\\
&\mathcal{M}_{2,2}= \int_0^{\delta t} \d\tau_1\int_0^{\tau_1}\d\tau_2\int_0^{\tau_2} \d\tau_3 \, e^{i\tau_1 H}e^{-i\tau_1A/2}
e^{-i\tau_1B}
e^{-i\tau_1A/2}  \left[i\frac{A}{2},\left[i\frac{A}{2},iB\right]\right],\\
&\mathcal{M}_{\mathrm{re},1}= \int_0^{\delta t} \d\tau_1\int_0^{\tau_1}\d\tau_2\int_0^{\tau_2} \d\tau_3 \, e^{i\tau_1 H}e^{-i\tau_1A/2}
e^{-i\tau_1B}e^{-i\tau_1A/2}R_1,\\
&\mathcal{M}_{\mathrm{re},2}=\int_0^{\delta t} \d\tau_1\int_0^{\tau_1}\d\tau_2\int_0^{\tau_2} \d\tau_3 \, e^{i\tau_1 H}e^{-i\tau_1A/2}
e^{-i\tau_1B}e^{-i\tau_1A/2}R_2 .
\end{aligned}
\end{equation}
We have
\begin{equation}
\begin{aligned}
  R_1&= e^{i\tau_1A/2}e^{i\tau_1B}  e^{-i\tau_3B}\left[-iB,\left[-iB,-i\frac{A}{2}\right]\right]e^{i\tau_3B} e^{-i\tau_1B} e^{-i\tau_1A/2 }-\left[-iB,\left[-iB,-i\frac{A}{2}\right]\right]\\
  & =e^{i\tau_1A/2}e^{i\tau_1B}  \left[-iB,\left[-iB,-i\frac{A}{2}\right]\right]e^{-i\tau_1B} e^{-i\tau_1A/2}-\left[-iB,\left[-iB,-i\frac{A}{2}\right]\right]\\
   &\quad+ \int_0^{\tau_3}\d\tau_4 \,
 e^{i\tau_1A/2}e^{i\tau_1B}  e^{-i(\tau_3-\tau_4)B}\left[-iB,\left[-iB,\left[-iB,-i\frac{A}{2}\right]\right]\right]e^{i(\tau_3-\tau_4)B} e^{-i\tau_1B} e^{-i\tau_1A/2}\\
 &=e^{i\tau_1A/2}\left[-iB,\left[-iB,-i\frac{A}{2}\right]\right] e^{-i\tau_1A/2}-\left[-iB,\left[-iB,-i\frac{A}{2}\right]\right]\\
 &\quad+ \int_0^{\tau_1} \d \tau_2 \, e^{i\tau_1A/2}e^{i(\tau_1-\tau_2)B}\left[iB,\left[-iB,\left[-iB,-i\frac{A}{2}\right]\right]\right] e^{-i(\tau_1-\tau_2)B}e^{-i\tau_1A/2}\\
 &\quad+ \int_0^{\tau_3}\d\tau_4 \,
 e^{i\tau_1A/2}e^{i\tau_1B}  e^{-i(\tau_3-\tau_4)B}\left[-iB,\left[-iB,\left[-iB,-i\frac{A}{2}\right]\right]\right]e^{i(\tau_3-\tau_4)B} e^{-i\tau_1B} e^{-i\tau_1A/2}\\
&=\int_0^{\tau_1} \d \tau_1' \, e^{i(\tau_1-\tau_1')A/2}\left[i\frac{A}{2},\left[-iB,\left[-iB,-i\frac{A}{2}\right]\right]\right] e^{-i(\tau_1-\tau_1')A/2}\\
&\quad+ \int_0^{\tau_1} \d \tau_1' \, e^{i\tau_1A/2}e^{i(\tau_1-\tau_1')B}\left[iB,\left[-iB,\left[-iB,-i\frac{A}{2}\right]\right]\right] e^{-i(\tau_1-\tau_1')B}e^{-i\tau_1A/2}\\
 &\quad+ \int_0^{\tau_3}\d\tau_4 \,
 e^{i\tau_1A/2}e^{i\tau_1B}  e^{-i(\tau_3-\tau_4)B}\left[-iB,\left[-iB,\left[-iB,-i\frac{A}{2}\right]\right]\right]e^{i(\tau_3-\tau_4)B} e^{-i\tau_1B} e^{-i\tau_1A/2}.
\end{aligned}
\end{equation}
Therefore
\begin{equation}
     \|R_1\| \le \tau_1 \left \|\left[i\frac{A}{2},\left[-iB,\left[-iB,-i\frac{A}{2}\right]\right]\right]\right\|+ (\tau_1+ \tau_3) \left \|\left[iB,\left[-iB,\left[-iB,-i\frac{A}{2}\right]\right]\right]\right\|,
\end{equation}
which implies
\begin{equation}      \|\mathcal{M}_{\mathrm{re},1}\| \le \frac{{\delta t}^4}{32}\|\left[A,\left[B,\left[B,A\right]\right]\right]\|+  \frac{{\delta t}^4}{12} \|\left[B,\left[B,\left[B,A\right]\right]\right]\|=: \zeta_{\re,1}.
\label{eq:remain_pf2_re1}
\end{equation}
Similarly, we have
\begin{equation}
\begin{aligned}
R_2&=e^{i\tau_1A/2}e^{i\tau_1B} e^{i\tau_3A/2}\left[i\frac{A}{2},\left[i\frac{A}{2},iB\right]\right]e^{-i\tau_3A/2} e^{-i\tau_1B} e^{-i\tau_1A/2}-\left[i\frac{A}{2},\left[i\frac{A}{2},iB\right]\right] \\
&=e^{i\tau_1A/2}e^{i\tau_1B} \left[i\frac{A}{2},\left[i\frac{A}{2},iB\right]\right] e^{-i\tau_1B} e^{-i\tau_1A/2}-\left[i\frac{A}{2},\left[i\frac{A}{2},iB\right]\right]
\\
&\quad+ \int_0^{\tau_3}\d \tau_4 \, e^{i\tau_1A/2}e^{i\tau_1B}  e^{i(\tau_3-\tau_4)A/2}   \left[i\frac{A}{2}, \left[i\frac{A}{2},\left[i\frac{A}{2},iB\right]\right]\right]e^{-i(\tau_3-\tau_4)A/2} e^{-i\tau_1B} e^{-i\tau_1A/2}\\
&=e^{i\tau_1A/2}\left[i\frac{A}{2},\left[i\frac{A}{2},iB\right]\right] e^{-i\tau_1A/2}-\left[i\frac{A}{2},\left[i\frac{A}{2},iB\right]\right] \\
&\quad+ \int_0^{\tau_1} \d \tau_1' \,  e^{i\tau_1A/2}e^{i(\tau_1-\tau_1')B}\left[iB, \left[i\frac{A}{2},\left[i\frac{A}{2},iB\right]\right] \right] e^{-i(\tau_1-\tau_1')B} e^{-i\tau_1A/2}\\
&\quad+ \int_0^{\tau_3}\d \tau_4 \, e^{i\tau_1A/2}e^{i\tau_1B}  e^{i(\tau_3-\tau_4)A/2}  \left[i\frac{A}{2}, \left[i\frac{A}{2},\left[i\frac{A}{2},iB\right]\right]\right]e^{-i(\tau_3-\tau_4)A/2} e^{-i\tau_1B} e^{-i\tau_1A/2}\\
&=\int_0^{\tau_1} \d \tau_1' \, e^{i(\tau_1-\tau_1')A/2} \left[i\frac{A}{2}, \left[i\frac{A}{2},\left[i\frac{A}{2},iB\right]\right]\right] e^{-i(\tau_1-\tau_1')A/2} \\
&\quad+ \int_0^{\tau_1} \d \tau_1' \,  e^{i\tau_1A/2}e^{i(\tau_1-\tau_1')B}\left[iB, \left[i\frac{A}{2},\left[i\frac{A}{2},iB\right]\right] \right] e^{-i(\tau_1-\tau_1')B} e^{-i\tau_1A/2}\\
&\quad+ \int_0^{\tau_3}\d \tau_4 \, e^{i\tau_1A/2}e^{i\tau_1B}  e^{i(\tau_3-\tau_4)A/2}  \left[i\frac{A}{2}, \left[i\frac{A}{2},\left[i\frac{A}{2},iB\right]\right]\right]e^{-i(\tau_3-\tau_4)A/2} e^{-i\tau_1B} e^{-i\tau_1A/2},
\end{aligned}
\end{equation}
so
\begin{equation}
\|R_2\| \le (\tau_1+ \tau_3) \left\|\left[i\frac{A}{2}, \left[i\frac{A}{2},\left[i\frac{A}{2},iB\right]\right]\right] \right\|+ \tau_1\left\|\left[B, \left[i\frac{A}{2},\left[i\frac{A}{2},iB\right]\right] \right]\right\|,
\end{equation}
which implies
\begin{equation}
\|\mathcal{M}_{\mathrm{re},2}\| \le \frac{{\delta t}^4}{32}\|\left[B,\left[A,\left[A,B\right]\right]\right]\|+  \frac{{\delta t}^4}{48}\|\left[A,\left[A,\left[A,B\right]\right]\right]\|=: \zeta_{\re,2}.
\label{eq:remain_pf2_re2}
\end{equation}

Similarly, $\mathcal{M}_{2,1}$ and $\mathcal{M}_{2,2}$ can be bounded as

\begin{align}
\norm{\mathcal{M}_{2,1}}
&\le \frac{{\delta t}^3}{12} \| \left[B,\left[B,A\right]\right]\|, \\
\norm{\mathcal{M}_{2,2}}
&\le \frac{{\delta t}^3}{24} \|\left[A,\left[A,B\right]\right]\|.
\end{align}

By the triangle inequality,
\begin{equation}
\begin{aligned}
\|{\mathcal{M}(\delta t)}\|&\le \|\mathcal{M}_{2,1}(\delta t)\| + \|{\mathcal{M}_{2,2}(\delta t)}\| +\|\mathcal{M}_{\mathrm{re},1}\|+\|\mathcal{M}_{\mathrm{re},2}\|\\
&\le\frac{{\delta t}^3}{12} \| \left[B,\left[B,A\right]\right]\|+ \frac{{\delta t}^3}{24} \|\left[A,\left[A,B\right]\right]\|+\frac{{\delta t}^4}{32}\|\left[A,\left[B,\left[B,A\right]\right]\right]\|+  \frac{{\delta t}^4}{12} \|\left[B,\left[B,\left[B,A\right]\right]\right]\|\\
 &\quad + \frac{{\delta t}^4}{32}\|\left[B,\left[A,\left[A,B\right]\right]\right]\|+  \frac{{\delta t}^4}{48}\|\left[A,\left[A,\left[A,B\right]\right]\right]\|.
\end{aligned}  \label{eq:triangle_pf2}
\end{equation}
Thus, for $0\le\tau_1\le \delta t$, $\norm{\mathcal{M}(\tau_1)}\le \norm{\mathcal{M}(\delta t)}$.
Since $e^{i\tau_1 H}e^{-i\tau_1A/2}
e^{-i\tau_1B}
e^{-i\tau_1A/2}=e^{i\tau_1 H}U_2(\tau_1)=(I+\mathcal{M}(\tau_1))$,
$\mathcal{M}_{2,1}$ can be rewritten as
\begin{equation}
    \begin{split}
        \mathcal{M}_{2,1}&=\int_0^{\delta t} \d\tau_1\int_0^{\tau_1}\d\tau_2\int_0^{\tau_2} \d\tau_3 \, (I+\mathcal{M}(\tau_1))\left[-iB,\left[-iB,-i\frac{A}{2}\right]\right]\\
        &=\int_0^{\delta t} \d\tau_1\int_0^{\tau_1}\d\tau_2\int_0^{\tau_2} \d\tau_3 \, \left[-iB,\left[-iB,-i\frac{A}{2}\right]\right]+\int_0^{\delta t} \d\tau_1\int_0^{\tau_1}\d\tau_2\int_0^{\tau_2} \d\tau_3 \, \mathcal{M}(\tau_1)\left[-iB,\left[-iB,-i\frac{A}{2}\right]\right]\\
        &=\frac{\delta t^3}{12}\left[-iB,\left[-iB,-iA\right]\right]+\int_0^{\delta t} \d\tau_1 \frac{\tau^2_1}{4}\, \mathcal{M}(\tau_1)\left[-iB,\left[-iB,-iA\right]\right]\\
        &=\frac{\delta t^3}{12}\left[-iB,\left[-iB,-iA\right]\right] + \mathcal{M}^*_{\re,(2,1)},
    \end{split}
\end{equation}
where $\mathcal{M}^*_{\re,(2,1)}$ can be bounded as
\begin{equation}
\begin{split}
      \| \mathcal{M}^*_{\re,(2,1)}\| &\le \int_0^{\delta t} \d\tau_1 \frac{\tau^2_1}{4}\, \|\mathcal{M}(\tau_1)\| \norm{\left[-iB,\left[-iB,-iA\right]\right] } \\
      &\le \|\mathcal{M}(\delta t)\| \int_0^{\delta t} \d\tau_1 \frac{\tau^2_1}{4}\,  \norm{\left[-iB,\left[-iB,-iA\right]\right] } \\
      &=\frac{\delta t^3}{12}\|\mathcal{M}(\delta t) \|\norm{\left[-iB,\left[-iB,-iA\right]\right] }\\
      &\le \big( \frac{{\delta t}^6}{144} \| \left[B,\left[B,A\right]\right]\|+ \frac{{\delta t}^6}{288} \|\left[A,\left[A,B\right]\right]\|+\frac{{\delta t}^7}{384}\|\left[A,\left[B,\left[B,A\right]\right]\right]\|+  \frac{{\delta t}^7}{144} \|\left[B,\left[B,\left[B,A\right]\right]\right]\|\\
 &\quad + \frac{{\delta t}^7}{384}\|\left[B,\left[A,\left[A,B\right]\right]\right]\|+  \frac{{\delta t}^7}{576}\|\left[A,\left[A,\left[A,B\right]\right]\right]\| \big) \norm{\left[B,\left[B,A\right]\right] }=: \zeta^*_{\re,(2,1)}.
\end{split}
\end{equation}
Similarly, $\mathcal{M}_{2,2}$ can be expressed as 
\begin{equation}
    \begin{split}
        \mathcal{M}_{2,2}&= \int_0^{\delta t} \d\tau_1\int_0^{\tau_1}\d\tau_2\int_0^{\tau_2} \d\tau_3 \, e^{i\tau_1 H}e^{-i\tau_1A/2}
e^{-i\tau_1B}
e^{-i\tau_1A/2}  \left[i\frac{A}{2},\left[i\frac{A}{2},iB\right]\right]\\
&=\frac{1}{4}\int_0^{\delta t} \d\tau_1\int_0^{\tau_1}\d\tau_2\int_0^{\tau_2} \d\tau_3 \, \left[iA,\left[iA,iB\right]\right]+\frac{1}{4}\int_0^{\delta t} \d\tau_1\int_0^{\tau_1}\d\tau_2\int_0^{\tau_2} \d\tau_3 \, \mathcal{M}(\tau_1)\left[iA,\left[iA,iB\right]\right]\\
        &=\frac{\delta t^3}{24}\left[iA,\left[iA,iB\right]\right]+\int_0^{\delta t} \d\tau_1 \frac{\tau^2_1}{8}\, \mathcal{M}(\tau_1)\left[iA,\left[iA,iB\right]\right]\\
        &=\frac{\delta t^3}{24}\left[iA,\left[iA,iB\right]\right] + \mathcal{M}^*_{\re,(2,2)},
    \end{split}
\end{equation}
where 

\begin{equation}
    \begin{split}
     \|\mathcal{M}^*_{\re,(2,2)} \|&\le \norm{\mathcal{M}(\delta t)} \int_0^{\delta t} \d\tau_1 \frac{\tau^2_1}{8}\, \norm{\left[iA,\left[iA,iB\right]\right]}\\
        &=\frac{\delta t^3}{24}  \norm{\mathcal{M}(\delta t)}\norm{\left[iA,\left[iA,iB\right]\right]}\\
        &\le \big( \frac{{\delta t}^6}{288} \| \left[B,\left[B,A\right]\right]\|+ \frac{{\delta t}^6}{576} \|\left[A,\left[A,B\right]\right]\|+\frac{{\delta t}^7}{768}\|\left[A,\left[B,\left[B,A\right]\right]\right]\|+  \frac{{\delta t}^7}{288} \|\left[B,\left[B,\left[B,A\right]\right]\right]\|\\
 &\quad + \frac{{\delta t}^7}{768}\|\left[B,\left[A,\left[A,B\right]\right]\right]\|+  \frac{{\delta t}^7}{1152}\|\left[A,\left[A,\left[A,B\right]\right]\right]\| \big) \norm{\left[A,\left[A,B\right]\right]}=:\zeta^*_{\re,(2,2)}.
    \end{split}
\end{equation}

Therefore, the multiplicative error $\mathcal{M}(\delta t)$ of PF2 is 
\begin{equation}
    \mathcal{M}(\delta t) = \frac{\delta t^3}{12}\left[-iB,\left[-iB,-iA\right]\right] + \frac{\delta t^3}{24}\left[iA,\left[iA,iB\right]\right] + \mathcal{M}^*_{\re,(2,1)}+ \mathcal{M}^*_{\re,(2,2)} +\mathcal{M}_{\re,1}+\mathcal{M}_{\re,2}.
\end{equation}
Clearly, $ \frac{\delta t^3}{12}\left[-iB,\left[-iB,-iA\right]\right] + \frac{\delta t^3}{24}\left[iA,\left[iA,iB\right]\right]$ is the leading error term of $\mathcal{M}(\delta t)$ since 
 $$\mathcal{M}_{Re}:=\mathcal{M}^*_{\re,(2,1)}+ \mathcal{M}^*_{\re,(2,2)} +\mathcal{M}_{\re,1}+\mathcal{M}_{\re,2}$$ with
$\|\mathcal{M}_{Re}\| \le \|\mathcal{M}^*_{\re,(2,1)}\|+\|\mathcal{M}^*_{\re,(2,2)} \|+\norm{\mathcal{M}_{\re,1}}+\norm{\mathcal{M}_{\re,2}}=\bigO(\delta t^4)$. 

According to \cref{scramblingPF} and the leading term of $\mathcal{M}(\delta t)$,  we obtain
\begin{equation}
\begin{split}
       \epsilon_O &\le \frac{1}{12}\norm{[O(\delta t), [B,[B,A]]] \ket{\psi}}\delta t^{3} +\frac{1}{24}\norm{[O(\delta t), [A,[A,B]]] \ket{\psi}}\delta t^{3} + 2\norm{O}\norm{\mathcal{M}_{Re}} \\
       & \le \frac{1}{12}\norm{[O(\delta t), [B,[B,A]]] \ket{\psi}}\delta t^{3} +\frac{1}{24}\norm{[O(\delta t), [A,[A,B]]] \ket{\psi}}\delta t^{3} \\
       &\quad + 2\norm{O} \big(\|\mathcal{M}^*_{\re,(2,1)}\|+\|\mathcal{M}^*_{\re,(2,2)} \|+\norm{\mathcal{M}_{\re,1}}+\norm{\mathcal{M}_{\re,2}} \big)\\
      &= \frac{1}{12}\norm{[O(\delta t), [B,[B,A]]] \ket{\psi}}\delta t^{3} +\frac{1}{24}\norm{[O(\delta t), [A,[A,B]]] \ket{\psi}}\delta t^{3} + 2\norm{O}(\zeta_{\re,1}+\zeta_{\re,2}+\zeta^*_{\re,(2,1)}+\zeta^*_{\re,(2,2)}).
\end{split}
    \end{equation}

\end{proof}

\section{Numerical results}
 
\begin{figure*}[htb]
    \centering
    \subfloat[]{
    \includegraphics[width=0.45\textwidth]{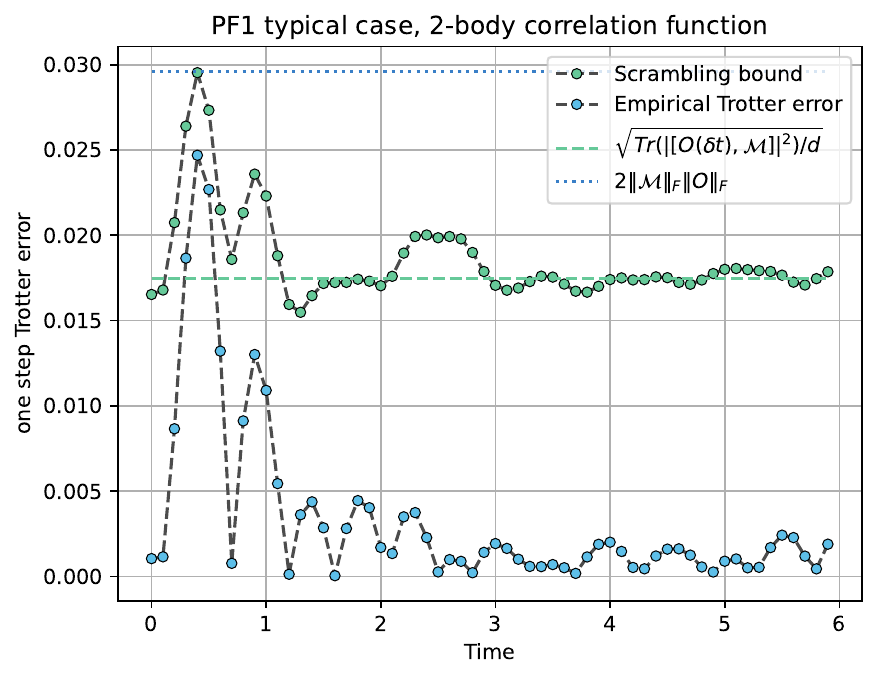}
    }
    \subfloat[]{
    \includegraphics[width=0.46\textwidth]{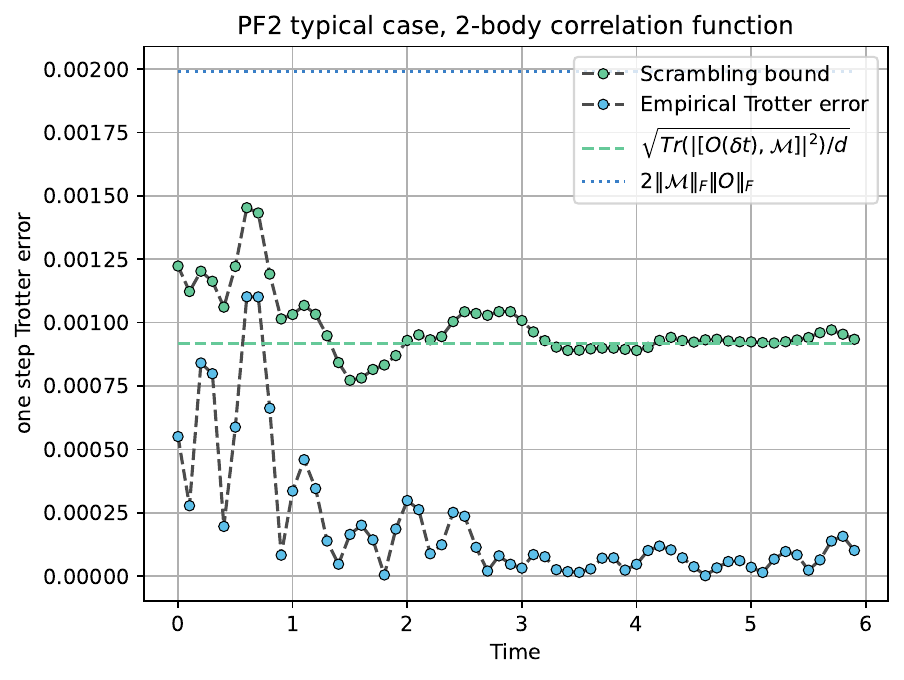}
    }
    
    \subfloat[]{
    \includegraphics[width=0.45\textwidth]{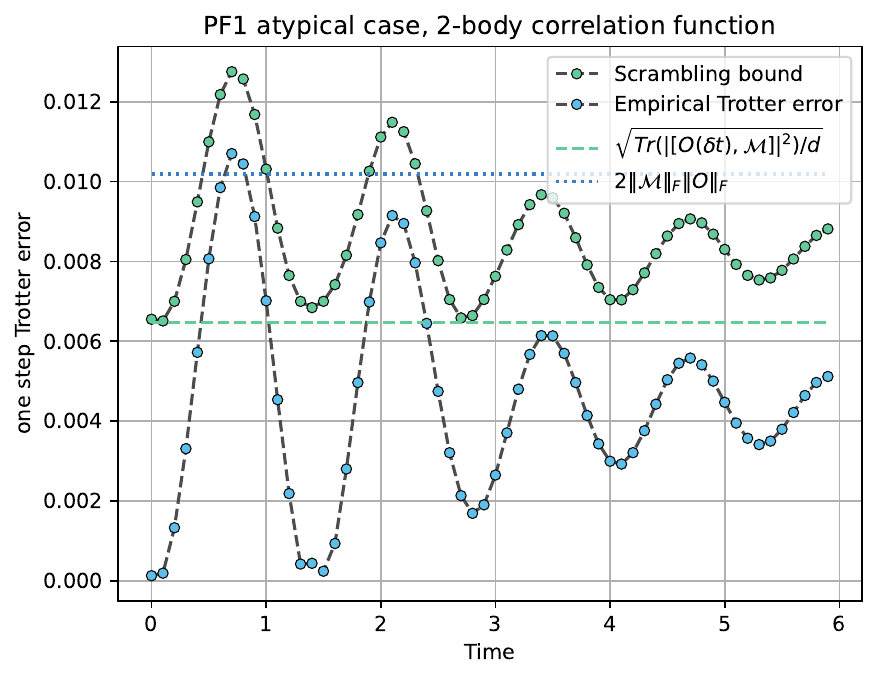}
    }
    \subfloat[]{
    \includegraphics[width=0.46\textwidth]{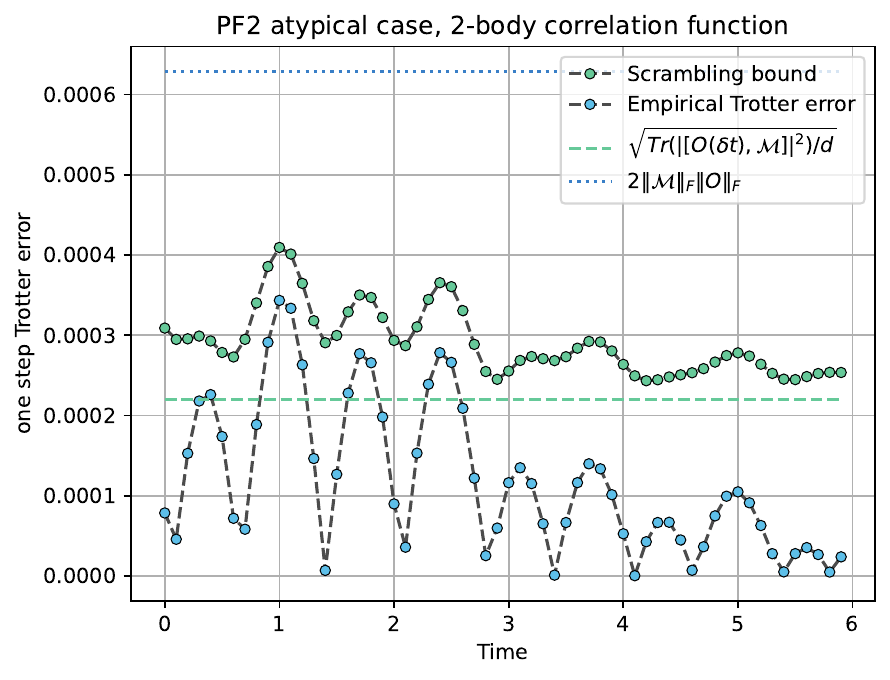}
    }
    \caption{One-step Trotter error of a 10-qubit 2-local observable $\langle\psi(t)|\mathscr U_p^\dagger O\mathscr U_p-U_0^\dagger OU_0|\psi(t)\rangle$ and the scrambling bound given by Eq.\ \ref{eq:one-step_scrambling_bound}, where O is set as $O=\dfrac{\sum_{i=1}^{N-1}X_iX_{i+1}}{\|\sum_{i=1}^{N-1}X_iX_{i+1}\|_{\infty}}$. The state is initialized as $|01\rangle^{\otimes N/2}$ and then evolved by a Hamiltonian of QIMF. The length of one segment is $t/r=0.1$. (a) and (b) show results of PF1 and PF2 method, respectively, with a typical set of parameters $(h_x,h_y,J)=(0.809, 0.9045, 1)$. (c) and (d) show results of PF1 and PF2 method, respectively, with an atypical set of parameter $(h_x,h_y,J)=(0,0.9045,0.4)$, which prevent the entanglement entropy of state $|\psi(t)\rangle$ increasing to its maximum. } 
    \label{fig:scramble}
\end{figure*}

We consider a one-dimensional quantum Ising spin model with mixed field (QIMF), governed by the Hamiltonian:
\begin{equation}
H=h_x\sum_{j=1}^{N}X_j+h_y\sum_{j=1}^NY_j+J\sum_{j=1}^{N-1}X_jX_{j+1},
\label{main:Hamil}
\end{equation}
where $X_j$ and $Y_j$ denote Pauli operators on site $j$, $h_x$ and $h_y$ represent transverse and longitudinal field strengths, respectively, and $J$ is the nearest-neighbor spin coupling. The evolution operator $\mathscr U_0=e^{-iH\delta t}$ can be approximated using a first-order product formula (PF1) $\mathscr U_1=e^{-iA\delta t}e^{-iB\delta t}$ or a second-order product formula (PF2) $\mathscr U_2=e^{-iA\delta t/2}e^{-iB\delta t}e^{-iA\delta t/2}$, where $A=h_x\sum_{j=1}^N X_j+J\sum_{j=1}^{N-1}X_jX_{j+1}$ and $B=h_y\sum_{j=1}^NY_j$.
The multiplicative Trotter error $\mathcal{M}$ is calculated according to $\mathcal{M}=U_0^\dagger \mathscr U_1-I$ and $\mathcal{M}=U_0^\dagger \mathscr U_2-I$ for PF1 and PF2 respectively. The leading error term $\sum_j M_j$ consists of terms $M_j$ proportional to commutators. For PF1, $M=\sum_jM_j=-\dfrac12[A,B]$ where $M=-ih_xh_y\sum_{j=1}^NZ_j-iJh_y\sum_{j=1}^{N-1}(Z_jX_{j+1}+X_jZ_{j+1})$ and $M_j$ are 2-local terms. For PF2, $M=\frac{1}{12}[-iB,[-iB,-iA]]+\frac{1}{24}[iA,[iA,iB]]=\sum_j M_j$, where the leading term of multiplicative error is 3-local: 
\begin{equation}
\begin{split}
    [A,[A,B]]=&4h_x^2h_y\sum_{j=1}^NY_j+4J^2h_y\sum_{j=1}^{N-1}Y_j+4J^2h_y\sum_{j=2}^NY_j\\&+8Jh_xh_y\sum_{j=1}^{N-1}(Y_jX_{j+1}+X_jY_{j+1})+8J^2h_y\sum_{j=1}^{N-2}(X_jY_{j+1}X_{j+2}),\\
    [B,[A,B]]=&-4h_xh_y^2\sum_{j=1}^NX_j+8Jh_y^2\sum_{j=1}^{N-1}(Z_jZ_{j+1}-X_jX_{j+1}).
\end{split}
\end{equation}
For numerical convenience, we use $\mathcal{M}$ directly rather than the leading error term $M \delta t^{p+1}$ in the following; this choice has a negligible effect when $\delta t$ is sufficiently small, i.e., $\delta t\le 0.1$.

\subsection{The one-step Trotter error of various observables}

Here we show some more numerical results analogous to those in the main part, but with other different observables. 
Figure~\ref{fig:scramble} shows the error in the expectation value of the results for observable $O=\dfrac{\sum_j{X_jX_{j+1}}}{\|\sum_j{X_jX_{j+1}}\|}$. These figures share some similarities. In typical cases, both the simulation error and scrambling decrease rapidly due to the increasing entanglement. The scrambling bound converges to the average value $\sqrt{\dfrac{\text{Tr}(|[O(\delta t),\mathcal M]|^2}d}$. In atypical cases, an atypical set of parameters $(h_x,h_y,J)=(0,0.9045,0.4)$ is chosen, which prevents the entanglement entropy from increasing to its maximum. As a result, the scrambling bound remains a bit larger than the average case.

\begin{figure}[tb]
\centering
\subfloat[]{
\includegraphics[width=0.45\textwidth]{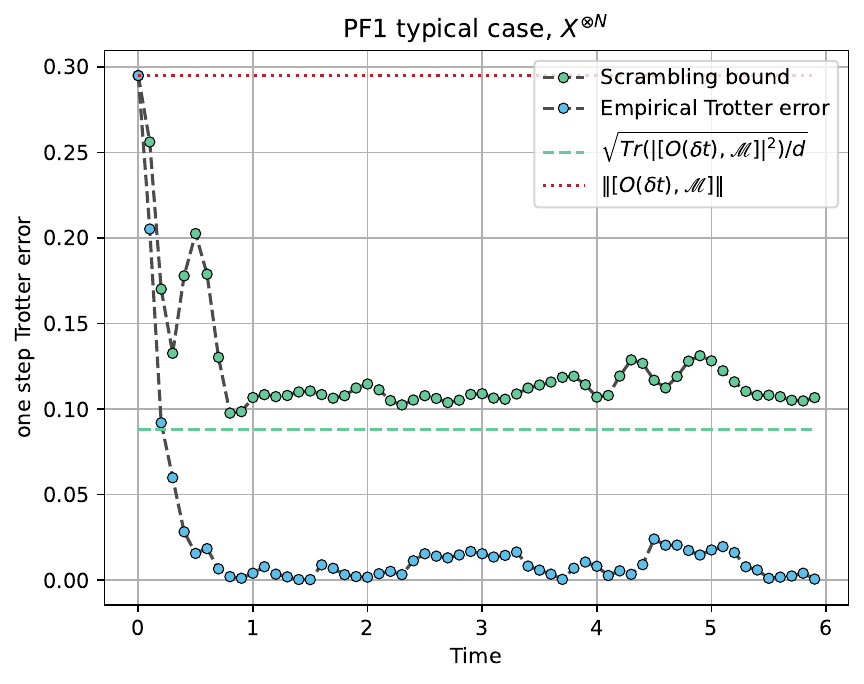}
}
\subfloat[]{
\includegraphics[width=0.45\textwidth]{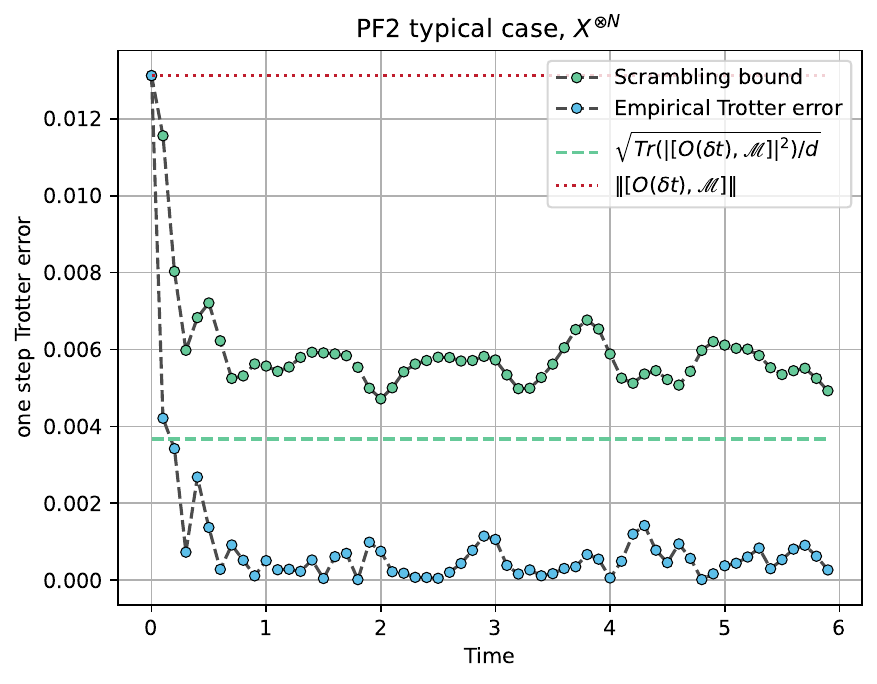}
}

\subfloat[]{
\includegraphics[width=0.45\textwidth]{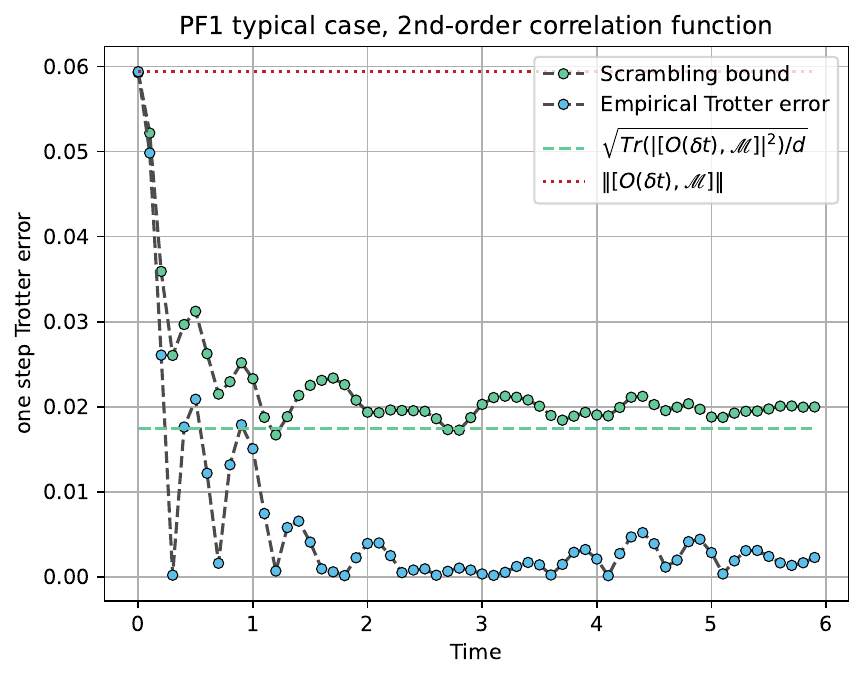}
}
\subfloat[]{
\includegraphics[width=0.45\textwidth]{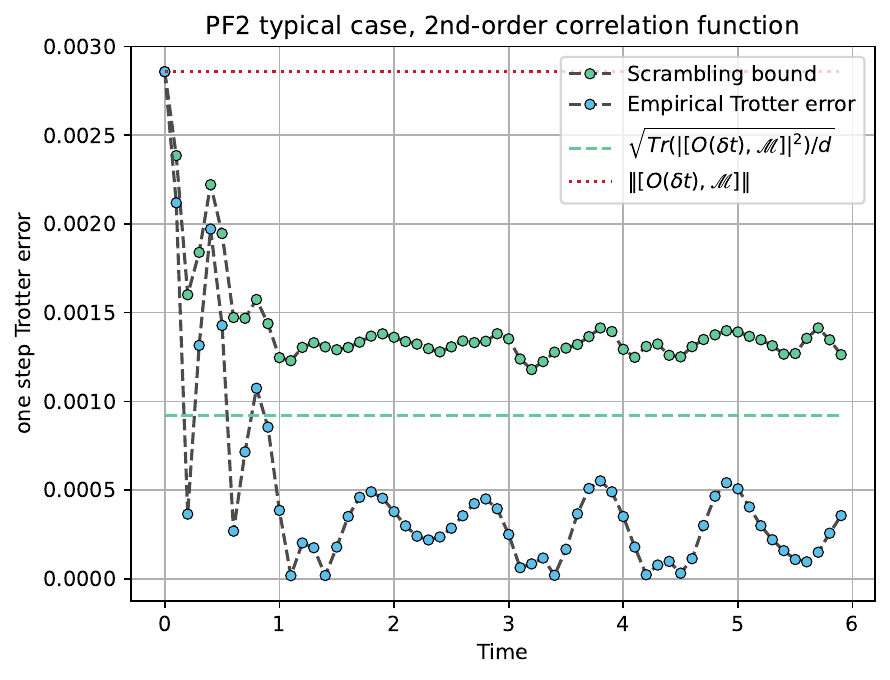}
}
\caption{One-step Trotter error of observable $O$ defined as $\bra{\psi(t)}\mathscr U_p^\dagger(t)O\mathscr U_p(t)-U_0^\dagger (t)OU_0(t)\ket{\psi(t)}$, comparing to our scrambling bound given by Eq.\ \ref{eq:one-step_scrambling_bound} and the worst-case bound (LB bound) $\|[O(\delta t),\mathcal M]\|$. The state is initialized as the worst-case state, i.e., the eigenstate of $U_p^\dagger(t)O\mathscr U_p(t)-U_0^\dagger (t)OU_0(t)$ with the eigenvalue with maximal absolute value, and then evolved by a Hamiltonian of QIMF with parameter $(h_x,h_y,J)=(0.809,0.9045,1)$. The length of one segment is t/r=0.1. Figures (a) and (b) show the results of the PF1 and PF2 methods, respectively, with the setting $O=X^{\otimes N}$. Figures (c) and (d) show the results of the PF1 and PF2 methods, respectively, with the setting $O=\sum_{i=1}^{N-1}X_iX_{i+1}/{\|\sum_{i=1}^{N-1}X_iX_{i+1}\|_\infty}$.}
\label{fig:LB_bound}
\end{figure}

Figure~\ref{fig:LB_bound} shows the improvement of our scrambling bound over the worst-case commutator bound. Setting the initial state as the worst-case state, we observe that as the system evolves, the error reduces, and our scrambling bound correspondingly reduces. 
This indicates that our scrambling bound provides a tighter bound than the worst-commutator bound \cite{LiPhysRevA.110.062614}.

\subsection{Operator-induced entanglement}
To show that induced entropy can cause a reduction in Trotter error, we compare the Trotter error and expectation value for the same set of observabless as shown in Figure~\ref{fig:comparison}. Figure~\ref{fig:comparison}(a) shows the one-step Trotter error for Pauli operators $P_1=X^{\otimes N},P_2=X_1\otimes X_2$ and for a 2-local operators $O_1=\dfrac{\sum_{i=1}^{N-1}X_iX_{i+1}}{\|\sum_{i=1}^{N-1}X_iX_{i+1}\|_{\infty}}, O_2=\dfrac{\sum_{i=1}^{N}X_i}{\|\sum_{i=1}^{N}X_i\|_{\infty}}$. Figure~\ref{fig:comparison}(b) shows expectation values of Pauli operators $P_1,P_2$ and operators $O_1,O_2$ during the simulation. The initial state is set as $|+\rangle^{\otimes N}$ so that its entanglement entropy remains at a low level. Although the expectation values of observables $O_1,O_2$ are larger than those of the Pauli operators $P_1,P_2$. This difference is caused by $\mathcal M$-induced entropy (See Figure~\ref{fig:comparison}(c)): although the entanglement of quantum state $\ket{\psi(t)}$ is weak, the entanglement of $\ket{\tilde{\psi_{r_k}}_\mathcal M}$ suppresses the simulation error of $O_1,O_2$.

\begin{figure}[tb]
    \centering
    \subfloat[]{
    \includegraphics[width=0.42\textwidth]{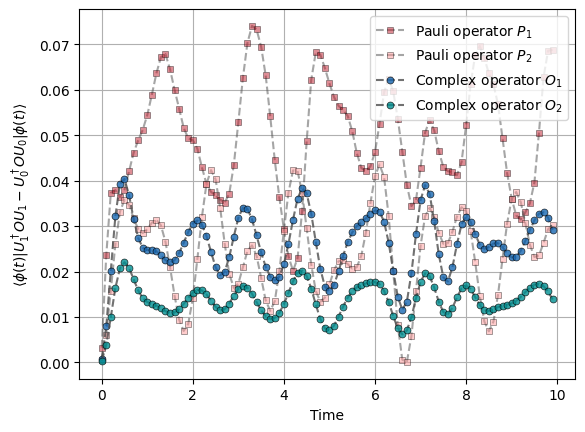}
    }
    \subfloat[]{
    \includegraphics[width=0.42\textwidth]{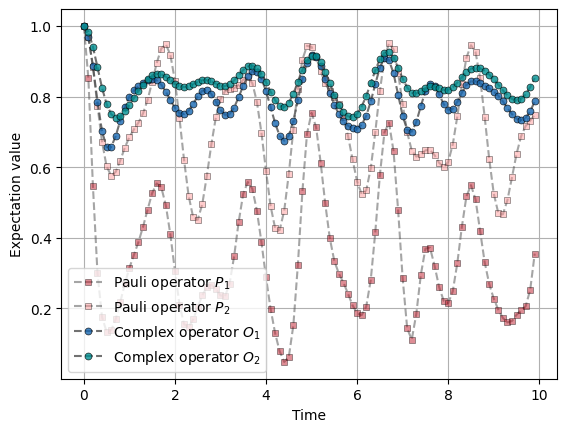}
    }
    
    \subfloat[]{
    \includegraphics[width=0.5\textwidth]{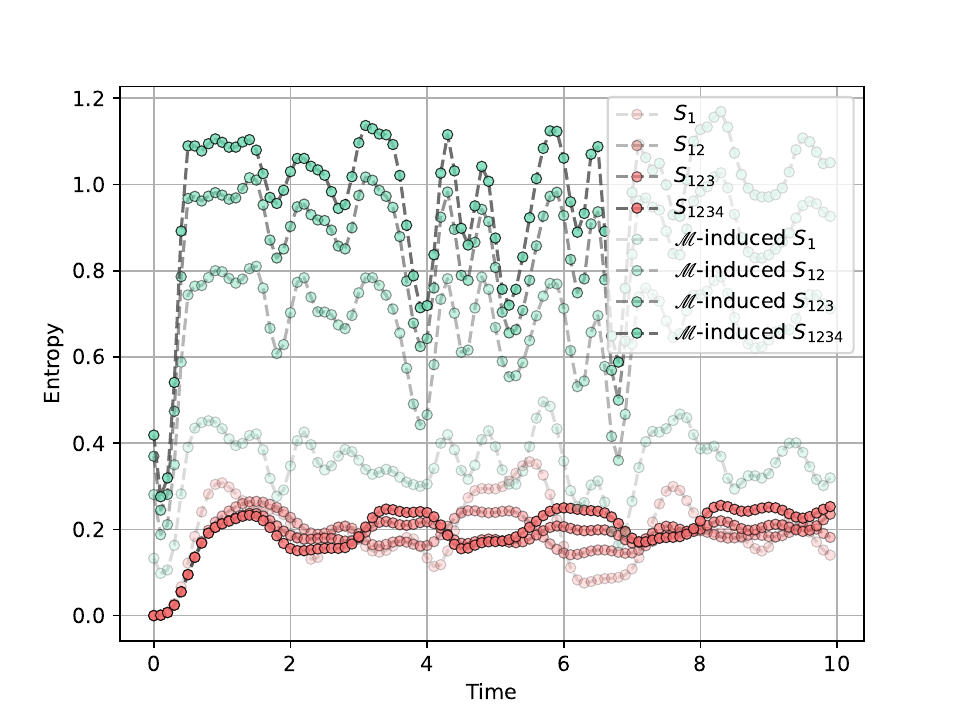}
    }
    
    \caption{Comparison of Pauli-string observables $P_1=X^{\otimes N},P_2=X_1\otimes X_2$ and normalized sums of local Pauli terms $O_1=\dfrac{\sum_{i=1}^{N-1}X_iX_{i+1}}{\|\sum_{i=1}^{N-1}X_iX_{i+1}\|_{\infty}}, O_2=\dfrac{\sum_{i=1}^{N}X_i}{\|\sum_{i=1}^{N}X_i\|_{\infty}}$. The initial state is set as $|\psi_0\rangle=|+\rangle^{\otimes N}$ and evolved by the Hamiltonian of the Ising model with typical parameters $(h_x,h_y,J)=(0.809,0.9045,1)$.(a)One-step Trotter error $|\langle\psi(t)|\mathscr U_1^\dagger O \mathscr U_1-U_0^\dagger OU_0|\psi(t)\rangle$ of these observables, where $|\psi(t)\rangle=e^{-iHt}\psi_0$. (b)Exact expectation value $\langle\psi(t)|O|\psi(t)\rangle$ of these observables. (c) shows the entanglement entropy of state $|\psi\rangle$ and $\mathcal M$-induced entropy.}
    \label{fig:comparison}
\end{figure}

\begin{figure}
    \centering
    \includegraphics[width=0.6\textwidth]{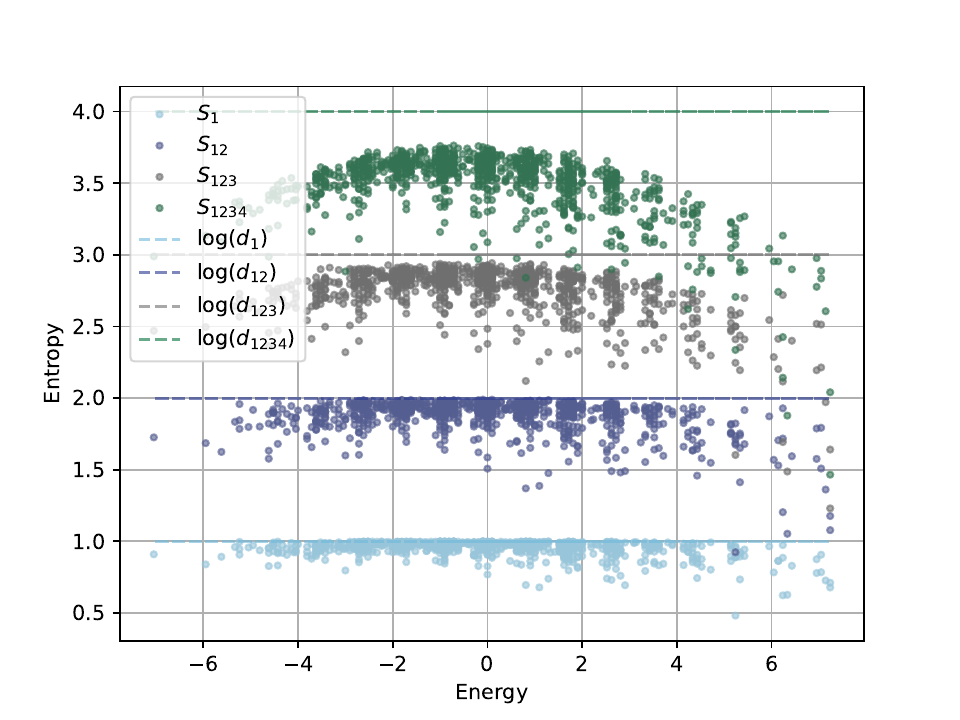}
    \caption{The relationship between energies and entanglement entropies of states. Tests of 1000 sampled states are shown. The dashed lines represent the maximal entangled entropies $\log d$, where $d$ denotes the dimension of the subsystem.}
    \label{fig:ee}
\end{figure}

\subsection{{More Hamiltonian with different parameters}}

{To further demonstrate that our conclusions do not rely on a single
hand-picked parameter point, we have added a new figure \ref{fig:placeholder} with four
randomly sampled parameter choices of the same mixed-field Ising
model for input state $\ket 0 ^{\otimes N}$. Since an overall rescaling of \(H\) only rescales time, we
fix \(J=1\) as the energy scale and sample the two remaining
dimensionless parameters \(h_x\) and \(h_y\) independently from
\([-2,2]\). The sampled points are
$
(h_x,h_y,J)=(1.214,-0.616,1),
(1.709,1.592,1),
$ $
(0.583,0.704,1),
(-0.493,-0.584,1).
$
These samples display different scrambling and entanglement behavior:
some approach the Frobenius/average scrambling scale rapidly, while
others remain farther from that scale during the simulated time
window. This variation is expected and is in fact useful: it shows
that the bound is informative not only in rapidly entangling regimes
but also in weaker-entanglement trajectories. In all four cases, the
state-dependent scrambling bound tracks the empirical observable
error and remains substantially below the state-independent
worst-case estimate. This supports the generality of the bound across
different dynamical regimes, without relying on a binary
chaotic/non-chaotic classification or on a single specially chosen
ETH-like point.}

\begin{figure}[htb]
    \centering
    \subfloat[]{
    \includegraphics[width=0.45\linewidth]{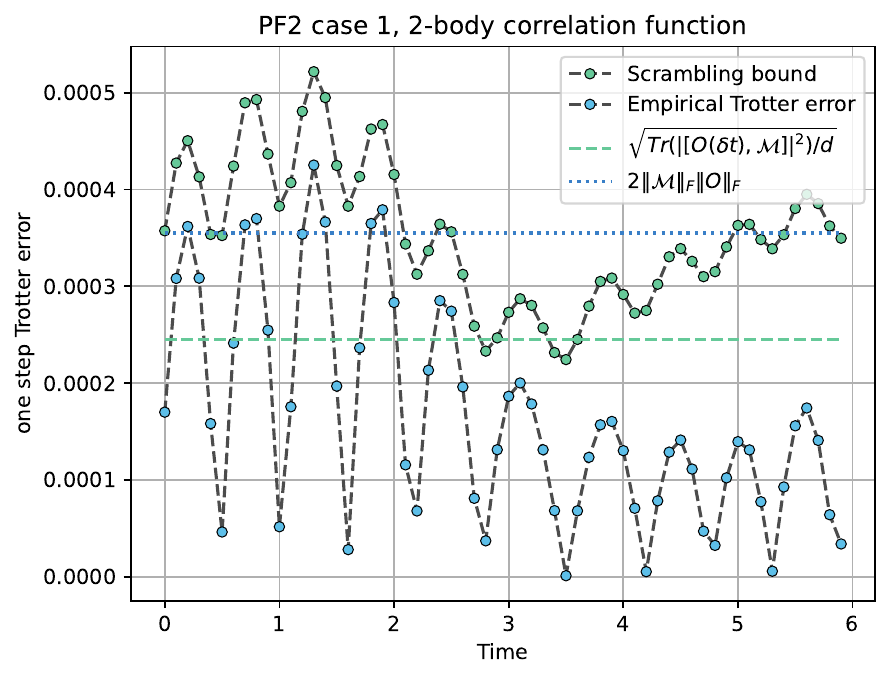}
    }
   \subfloat[]{
    \includegraphics[width=0.45\linewidth]{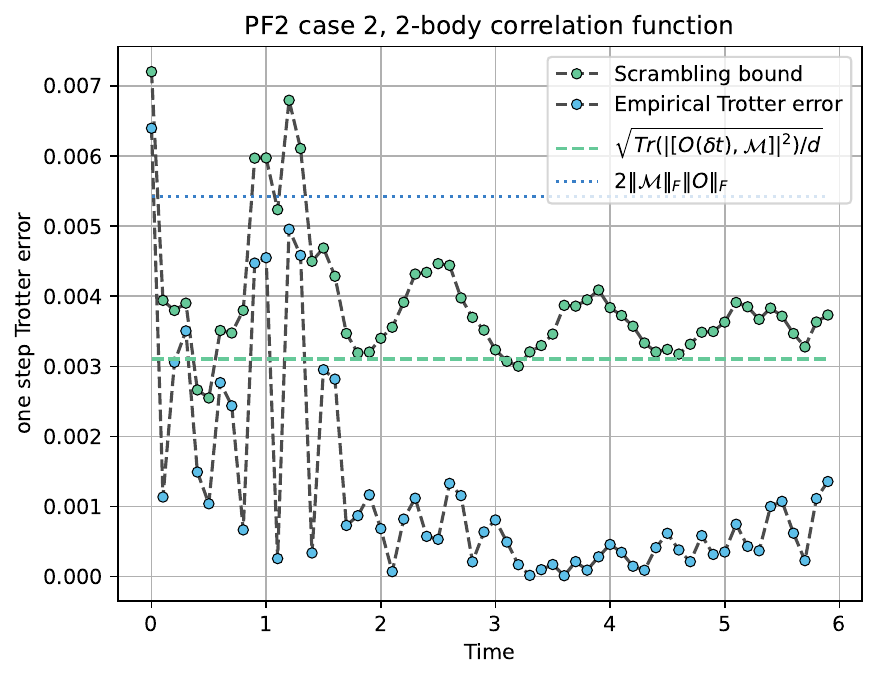}
    }
    
    \subfloat[]{
    \includegraphics[width=0.45\linewidth]{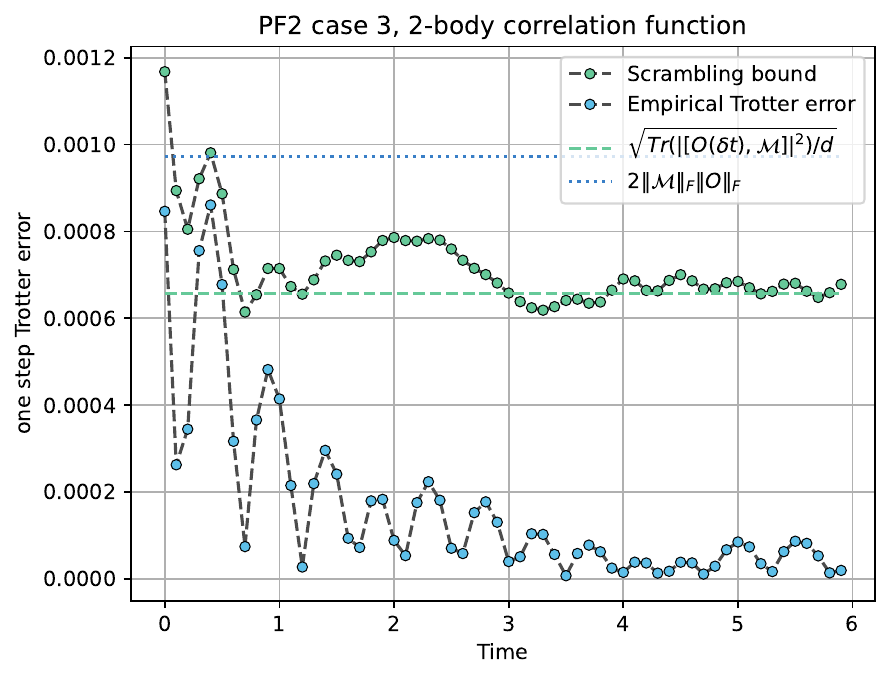}
    }
   \subfloat[]{
    \includegraphics[width=0.45\linewidth]{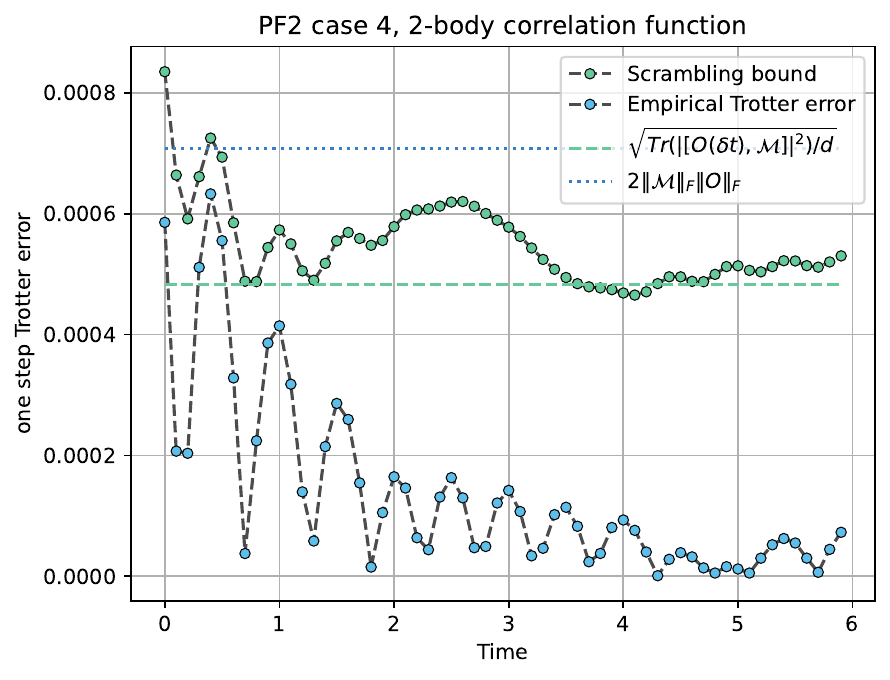}
    }
    \caption{Trotter error with Hamiltonian parameters $(h_x,h_y,J)=$ (a) $(1.214,-0.616,1)$, (b) $(1.709,1.592,1)$, (c) $(0.583,0.704,1)$, (d) $(-0.493,-0.584,1)$ and the initial state is set to be $\ket0^{\otimes N}$. The observable is set as the normalized Hamiltonian.}
    \label{fig:placeholder}
\end{figure}

\subsection{Energy and entropy in quantum simulation}
Here we numerically observe that certain initial states exhibit a saturation of entanglement entropy at values below the theoretical maximum $\log(d_{\text{supp}})$, where $d_{\text{supp}}$ denotes the dimension of the corresponding subsystem. In particular, the maximum attainable entanglement entropy during evolution depends critically on the initial energy of the state. For the QIMF, entanglement entropy can increase to the highest value when the energy of the initial state is zero (See Figure~\ref{fig:ee}). States oriented along the Z-axis, like $\ket 0^{\otimes N}$ and $\ket{01}^{\otimes N/2}$, are examples with zero energy that can get highly entangled after evolution. In contrast, states with high or low energy, like $\ket+^{\otimes N}$,  exhibit suppressed entanglement generation even under prolonged evolution.




\subsection{Long-time error and minimizing Trotter steps}

{It is precisely the quantity that 
appears in the accumulated scrambling bound. After a telescoping 
decomposition, the long-time observable error satisfies
\begin{equation}
      \epsilon_O(t)
  \le
  \sum_{k=1}^{r}
  \sqrt{C^{(k)}_{O,M}}\,\delta t^{p+1}
  +2r\|O\|\,\|\mathcal M_{\rm re}\|,
\end{equation}
where
$
  C^{(k)}_{O,M}
  =
  \bra{\psi_k}
  |[O(k\delta t),M]|^2
  \ket{\psi_k},
  \ket{\psi_k}=U_p^{r-k}\ket{\psi}.
$
Thus the accumulated error is controlled by the full sequence 
\(\{C^{(k)}_{O,M}\}_{k=1}^r\). Its growth, oscillation, and saturation 
directly determine how the segment-wise errors accumulate.}

{A useful way to understand the \(k\)-th term is the following. The 
observable side has evolved ideally for time \(k\delta t\), becoming 
\(O(k\delta t)\), while the state side has undergone the remaining 
Trotterized evolution, \(U_p^{r-k}\ket{\psi}\). These two pieces 
together cover the full physical time \(t=r\delta t\). For a generic 
local observable, \(O(k\delta t)\) grows in support as \(k\) 
increases. Since the leading error operator is itself a sum of local 
terms, \(M=\sum_j M_j\), more and more terms \(M_j\) can fail to 
commute with \(O(k\delta t)\) as the Heisenberg operator spreads. }

{We have added a segment-resolved numerical test (see \autoref{fig:segment}). We consider two kinds of observables, $\sum_{i=1}^{N-1} \frac{X_iX_{i+1}}{\|\sum_i^{N-1} X_iX_{i+1}\|}$ and $ O=\sum_{i=1}^{N} \frac{X_i}
{\left\|\sum_{i=1}^{N}X_i\right\|}$. 
The data (60 segments, $t=6$) show the expected
growth behavior as a function of the
backward propagation time \(k\delta t\). Equivalently, when
the horizontal axis is interpreted as the physical segment
index \(r-k\), later segments have a shorter remaining
propagation time and therefore contribute smaller
scrambling terms. This directly confirms that the operator scrambling
growth plays a role in the accumulated error estimate.
}

\begin{figure}[htb]
    \centering
    \subfloat[]{
    \includegraphics[width=0.45\linewidth]{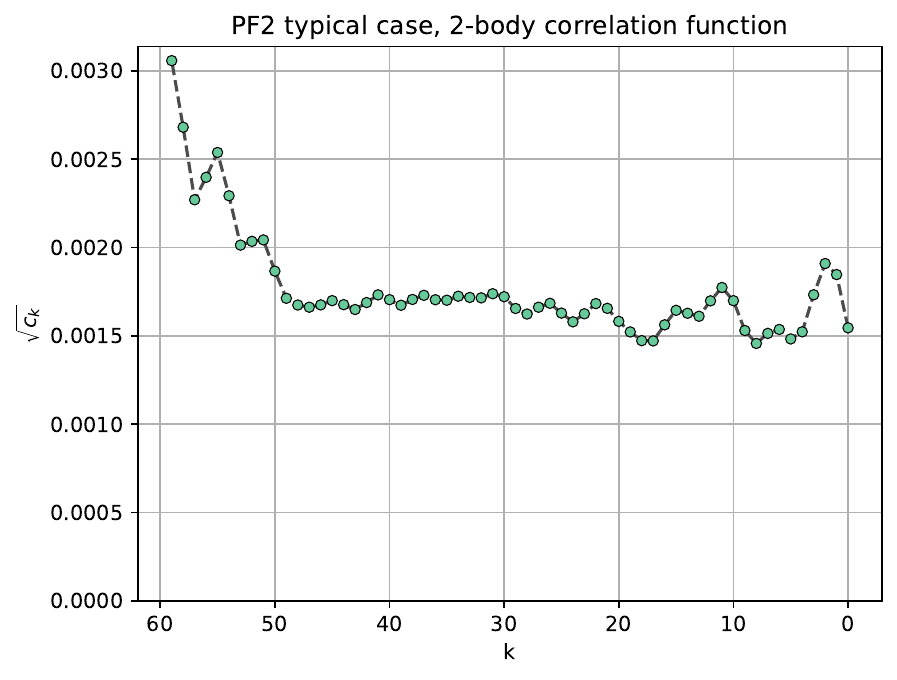}
    }
    \subfloat[]{
    \includegraphics[width=0.45\linewidth]{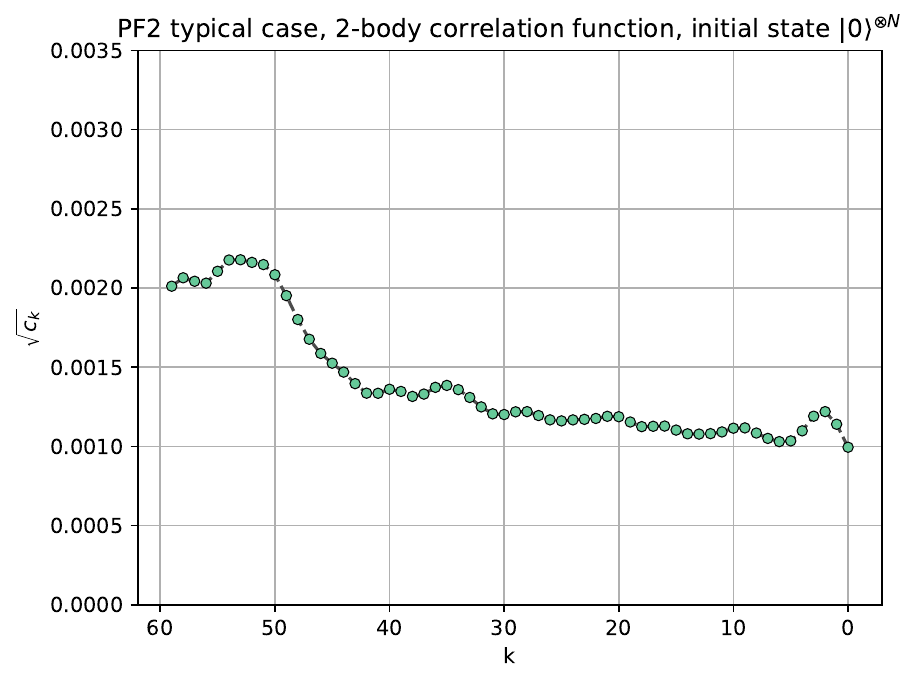}
    }

    \subfloat[]{
    \includegraphics[width=0.45\linewidth]{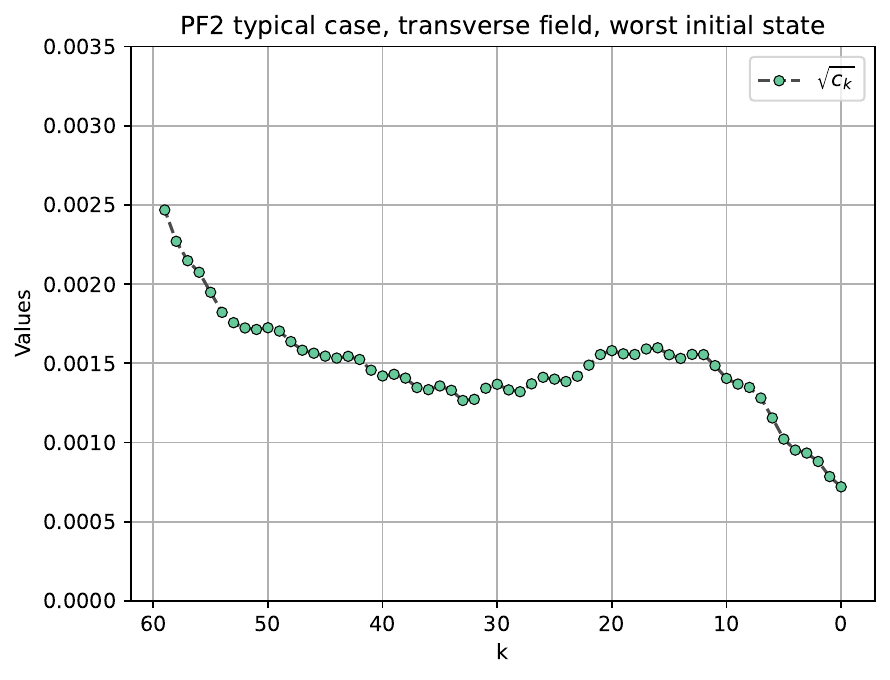}
    }
    \subfloat[]{
    \includegraphics[width=0.45\linewidth]{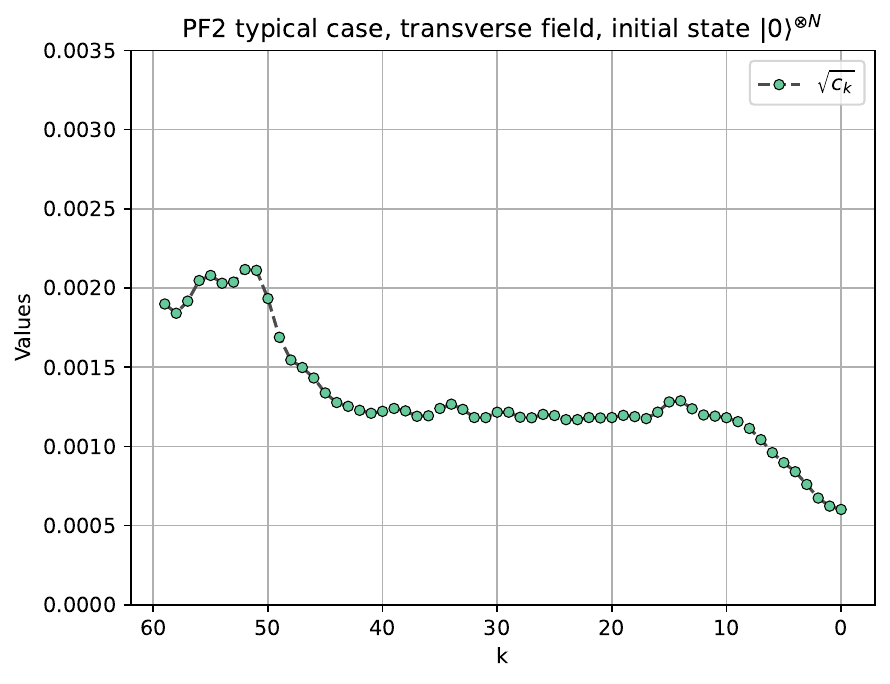}
    }
  \caption{
{Segment-resolved scrambling and one-segment Trotter error for a 
non-conserved observable. The observable is the normalized two-body 
correlator
$
  O={\sum_{i=1}^{N-1}X_iX_{i+1}}/
  {\left\|\sum_{i=1}^{N-1}X_iX_{i+1}\right\|}
$
for (a), (b) and 
$
 O=\sum_{i=1}^{N}X_i/
  {\left\|\sum_{i=1}^{N}X_i\right\|}
$ for (c), (d).
Here
\(O(k\delta t)=e^{iHk\delta t}Oe^{-iHk\delta t}\).
The index \(k\) is the remaining ideal propagation time
from the error insertion to the final observable. Thus,
larger \(k\) corresponds to an earlier error insertion and
a longer time for the error to scramble into the backward
light cone of \(O\). Panel (a)(c) uses the initial state that
maximizes the first-step Trotter error, while panel (b)(d)
uses the product state \(\ket{0}^{\otimes N}\). The data
show the expected growth of the
segment-resolved scrambling contribution with \(k\),
demonstrating that the error bound is 
captured by the accumulated bound.}
}  
    \label{fig:segment}
\end{figure}

    

Consider a set of observables of interest $\{O\}=\{O_{\mathcal{J}_1}, O_{\mathcal{J}_2}, ..., O_{\mathcal{J}_M} \}$ and input state $\ket{\psi}$ with Hamiltonian $H$, our task is to simulate these observables to the evolution $e^{-iHt}$ simultaneously. Here, we can define two qualities to quantify the performance of quantum simulation. The first is the worst simulation error in this observable set, 
    $\epsilon_{\{O\}}= \text{max}\epsilon_{O_{\mathcal{J}}}$.
We can also define the average simulation error of these observables,
    $\text{average }(\epsilon_{\{O\}})= \frac{1}{M} \sum_{\mathcal{J}=1}^M\epsilon_{O_{\mathcal{J}}} $.

Figure~\ref{fig:minimal_trotter_step} shows the relationship between the necessary Trotter step number and evolution time for target error $\varepsilon=10^{-4}$ in both worst cases and average cases. Our numerical experiments compare between a set $\{O\}$ that comprises combinations of Pauli operators, including Hamiltonian, magnetizations, low-order correlation functions, and some combinations of random local Pauli operators (these observables are normalized) and a set of randomly chosen Pauli operators $\{P\}$. In our example of QIMF, each Trotter step of the first-order Trotter-formula method can be implemented with 2 exponentials. Therefore, the Trotter step number indicates the gate complexity of the simulation. Theoretical bounds of the Trotter step number are also given according to the scrambling-based bound. Our results show that the gate cost is higher when dealing with a single Pauli operator compared to when dealing with a combination of various Pauli operators.

\begin{figure}[tb]
    \centering
    \subfloat[]{
    \includegraphics[width=0.45\textwidth]{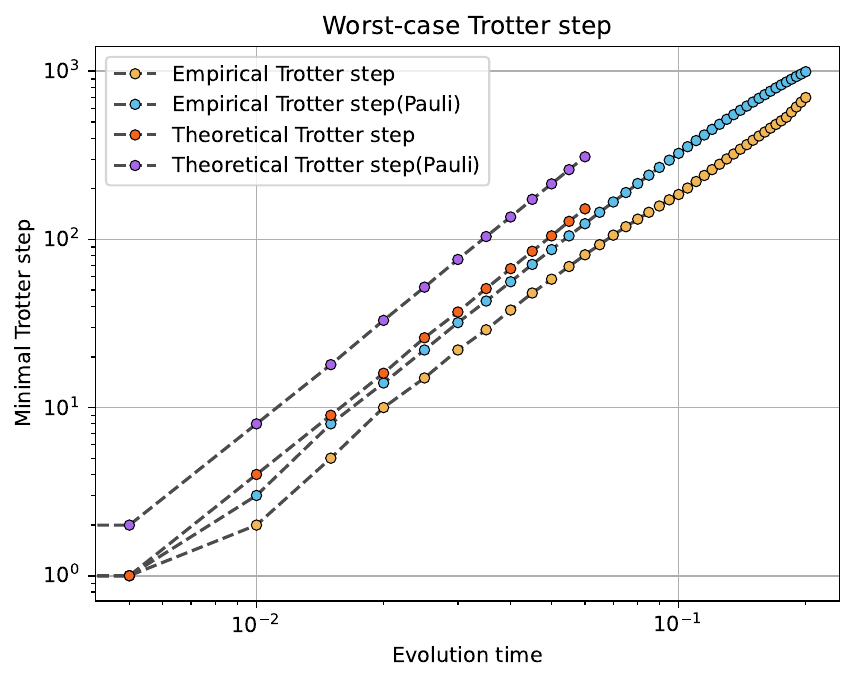}
    }
    \subfloat[]{
    \includegraphics[width=0.45\textwidth]{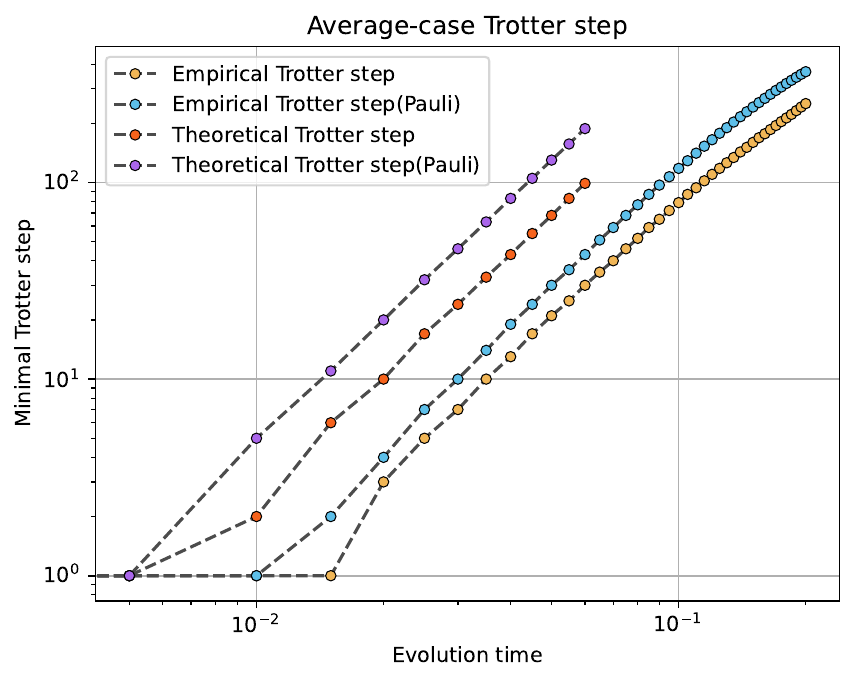}
    }
    \caption{Number of Trotter steps needed to achieve precision $\varepsilon=10^{-4}$ with first-order product formula for operators in a set $\{O\}$ and a set of random Pauli operators $\{P\}$. Observables in $\{O\}$ include Hamiltonian $\dfrac H{\|H\|_\infty}$, $\dfrac{\sum_iX_i}{\|\sum_iX_i\|_\infty}$, $\dfrac{\sum_iZ_i}{\|\sum_iZ_i\|_\infty}$, $\dfrac{\sum_iX_iX_{i+1}}{\|\sum_iX_iX_{i+1}\|_\infty}$, $\dfrac{\sum_iZ_iZ_{i+1}}{\|\sum_iZ_iZ_{i+1}\|_\infty}$, 10 combinations of random 2-local Pauli observables,10 combinations of random 3-local Pauli observables and 10 combinations of random 4-local Pauli observables. (a)Necessary number of Trotter steps for the worst-case observable in set $\{O\}$ and $\{P\}$. (b)Average necessary number of Trotter steps for the observable in set $\{O\}$ and $\{P\}$. Theoretical results are given according to the scrambling-based bound.}
    \label{fig:minimal_trotter_step}
\end{figure}
%
\end{document}